\documentclass[a4paper,UKenglish,cleveref, autoref, thm-restate]{lipics-v2021}

\title{Encoding Tight Typing in a Unified Framework}

\author{Delia Kesner}{Universit\'{e} de Paris, CNRS, IRIF and Institut Universitaire de France, France}{kesner@irif.fr}{https://orcid.org/0000-0003-4254-3129}{}

\author{Andr\'es Viso}{Inria, France}{andres-ezequiel.viso@inria.fr}{https://orcid.org/0000-0002-6822-8453}{}

\authorrunning{D. Kesner and A. Viso}

\Copyright{Delia Kesner and Andr\'es Viso}



\acknowledgements{We are specially grateful to G. Guerrieri for fruitful
discussions. We also thank B. Accattoli, A. Bucciarelli, and A. R\'{\i}os for
constructive remarks. This work has been partially supported by IRP SINFIN.}

\nolinenumbers 

\EventEditors{Florin Manea and Alex Simpson}
\EventNoEds{30}
\EventLongTitle{30th EACSL Annual Conference on Computer Science Logic (CSL 2022)}
\EventShortTitle{CSL 2022}
\EventAcronym{CSL}
\EventYear{2022}
\EventDate{February 14--19, 2022}
\EventLocation{G{\"o}ttingen, Germany}
\EventLogo{}
\SeriesVolume{216}
\ArticleNo{19}

\usepackage{float}
\usepackage{dsfont}
\usepackage{tikz}
\usetikzlibrary[positioning,decorations.pathmorphing,arrows]
\usepackage[normalem]{ulem}
\usepackage{amsmath}
\usepackage{amsthm}
\usepackage{amssymb}
\usepackage{color}
\usepackage{etoolbox}
\usepackage{mathrsfs}
\usepackage{paralist}
\usepackage{prooftree}
\usepackage[belowskip=-10pt,aboveskip=10pt]{caption}
\usepackage{rotating}
\usepackage{xargs}

\newcommand{\parrafo}[1]{\medskip \noindent {\sffamily\normalsize\bfseries{#1}.}}
\newcommand{\ttbf}[1]{{\ttfamily\normalsize\bfseries{#1}}}

\newcommand{\odelia}[1]{{\color{blue}\ifmmode\text{\sout{\ensuremath{#1}}}\else\sout{#1}\fi}}

\newcommand{\oandres}[1]{{\color{red}\ifmmode\text{\sout{\ensuremath{#1}}}\else\sout{#1}\fi}}

\newbool{compileReport}
\newcommand{\ifreport}[2][]{\ifbool{compileReport}{#2}{#1}}

\newcommand{\eg}{{\em e.g.~}}
\newcommand{\ie}{{\em i.e.~}}
\newcommand{\cf}{{\em cf.~}}

\newcommand{\coloneq}{\ensuremath{:=}}
\newcommand{\Coloneq}{\ensuremath{::=}}
\newcommand{\dbckslash}{\ensuremath{\setminus\!\!\setminus}}

\newcommand{\eqdef}{\mathrel{\raisebox{-1pt}{\ensuremath{\stackrel{\textit{\tiny{def}}}{=}}}}}

\newcommand{\emphdef}[1]{\textbf{\emph{#1}}}

\newcommand{\X}{\ensuremath{\mathcal{X}}}

\newcommand{\SysTightFlops}{\ensuremath{\mathcal{E}}}
\newcommand{\SysTight}{\ensuremath{\mathcal{B}}}
\newcommand{\SysTightCBN}{\ensuremath{\mathcal{N}}}
\newcommand{\SysTightCBV}{\ensuremath{\mathcal{V}}}

\newcommand{\ctxt}[1]{\ensuremath{\mathtt{#1}}}
\newcommand{\ctxtapp}[2]{\ensuremath{{#1}\langle{#2}\rangle}}

\newcommand{\TermExplicit}{\ensuremath{\mathcal{T}}}
\newcommand{\TermLambda}{\ensuremath{\mathcal{T_{\lambda}}}}

\newcommand{\TermVariable}{\ensuremath{\mathcal{X}}}

\newcommand{\exsubs}[2]{\ensuremath{[{#1}\backslash {#2}]}}

\newcommand{\termapp}[2]{\ensuremath{{#1}\,{#2}}}
\newcommand{\termabs}[2]{\ensuremath{\lambda{#1}.{#2}}}
\newcommand{\termbang}[1]{\ensuremath{\mathop{!}{#1}}}
\newcommand{\termder}[1]{\ensuremath{\mathop{\textnormal{der}}{#1}}}
\newcommand{\termsubs}[3]{\ensuremath{{#3}\exsubs{#1}{#2}}}

\newcommand{\functtype}[2]{\ensuremath{{#1}\to{#2}}}
\newcommand{\intertype}[2]{\ensuremath{\multiset{#1}_{#2}}}
\newcommand{\M}{\ensuremath{\mathcal{M}}}
\newcommand{\N}{\ensuremath{\mathcal{N}}}

\newcommand{\sequ}[2]{\ensuremath{{#1}\vdash{#2}}}

\newcommand{\assign}[2]{\ensuremath{{#1}:{#2}}}
\newcommand{\derivable}[3]{\ensuremath{{#1}\rhd_{#3}{#2}}}

\newcommand{\ctxtsum}[3]{\ensuremath{{#1}\mathrel{+_{#3}}{#2}}}
\newcommand{\ctxtres}[3]{\ensuremath{{#1}\mathrel{\dbckslash_{#3}}{#2}}}

\newcommand{\dom}[1]{\ensuremath{\mathtt{dom}({#1})}}

\newcommand{\bv}[1]{\ensuremath{\mathtt{bv}({#1})}}
\newcommand{\fv}[1]{\ensuremath{\mathtt{fv}({#1})}}

\newcommand{\rrule}[1]{\ensuremath{\mathrel{\mapsto_{#1}}}}
\newcommand{\rewrite}[1]{\rightarrow_{#1}}
\newcommand{\rewriten}[1]{\twoheadrightarrow_{#1}}

\newcommand{\bangweak}{\ensuremath{\mathtt{f}}}
\newcommand{\callbyname}{\ensuremath{\mathtt{n}}}
\newcommand{\callbyvalue}{\ensuremath{\mathtt{v}}}

\newcommand{\namestg}{\ensuremath{\mathtt{dn}}}

\newcommand{\valuestg}{\ensuremath{\mathtt{dv}}}
\newcommand{\loredv}{\rewrite{\valuestg}}
\newcommand{\loredname}{\rewrite{\namestg}}
\newcommand{\lorednamen}{\rewriten{\namestg}}
\newcommand{\loredvn}{\rewriten{\valuestg}}

\newcommand{\dB}{\ensuremath{\mathtt{dB}}}

\newcommand{\dBeta}{\ensuremath{\mathtt{dB}}}
\newcommand{\dBang}{\ensuremath{\mathtt{d!}}}
\newcommand{\sBang}{\ensuremath{\mathtt{s!}}}
\newcommand{\sTerm}{\ensuremath{\mathtt{sn}}}
\newcommand{\sVal}{\ensuremath{\mathtt{sv}}}

\newcommand{\mStep}{\ensuremath{\mathtt{m}}}
\newcommand{\eStep}{\ensuremath{\mathtt{e}}}

\newcommand{\set}[1]{\ensuremath{\{{#1}\}}}
\newcommand{\multiset}[1]{\ensuremath{[#1]}}

\newcommand{\substitute}[3]{\ensuremath{{{#3}\left\{{#1}\backslash{#2}\right\}}}}
\newcommand{\many}[2]{\ensuremath{({#1})_{#2}}}

\newcommand{\cbnname}{\ensuremath{\mathsf{cbn}}}
\newcommand{\cbnterm}[1]{\ensuremath{{#1}^{\callbyname}}}
\newcommand{\cbntype}[1]{\ensuremath{{#1}^{\callbyname}}}

\newcommand{\cbvname}{\ensuremath{\mathsf{cbv}}}
\newcommand{\cbvterm}[1]{\ensuremath{{#1}^{\callbyvalue}}}
\newcommand{\cbvtypeneg}[1]{\ensuremath{{#1}^{\overline{\callbyvalue}}}}
\newcommand{\cbvtypepos}[1]{\ensuremath{{#1}^{\callbyvalue}}}

\newcommand{\ruleName}[1]{\ensuremath{\mathtt{({#1})}}}

\newcommand{\ruleBDApp}{\ruleName{appt_{c}}}
\newcommand{\ruleBDArrowE}{\ruleName{app_{c}}}
\newcommand{\ruleBDArrowI}{\ruleName{abs_c}}
\newcommand{\ruleBDAxiom}{\ruleName{var_c}}
\newcommand{\ruleBDBang}{\ruleName{bg_c}}
\newcommand{\ruleBDDer}{\ruleName{dr_c}}
\newcommand{\ruleBDESubs}{\ruleName{es_c}}
\newcommand{\ruleBTArrowE}{\ruleName{app_p}}
\newcommand{\ruleBTArrowI}{\ruleName{abs_p}}
\newcommand{\ruleBTBang}{\ruleName{bg_p}}
\newcommand{\ruleBTDer}{\ruleName{dr_p}}
\newcommand{\ruleBTESubs}{\ruleName{es_p}}

\newcommand{\ruleNDArrowE}{\ruleName{app^{\mathcal{N}}_c}}
\newcommand{\ruleNDArrowI}{\ruleName{abs_c}}
\newcommand{\ruleNDAxiom}{\ruleName{var_c}}
\newcommand{\ruleNDESubs}{\ruleName{es^{\mathcal{N}}_c}}
\newcommand{\ruleNTArrowE}{\ruleName{app^{\mathcal{N}}_p}}
\newcommand{\ruleNTArrowI}{\ruleName{abs_p}}

\newcommand{\ruleVDArrowE}{\ruleName{app^{\mathcal{V}}_{c}}}
\newcommand{\ruleVDArrowI}{\ruleName{abs^{\mathcal{V}}_c}}
\newcommand{\ruleVDAxiom}{\ruleName{var^{\mathcal{V}}_c}}
\newcommand{\ruleVDESubs}{\ruleName{es_c}}
\newcommand{\ruleVDApp}{\ruleName{appt^{\mathcal{V}}_{c}}}
\newcommand{\ruleVTArrowE}{\ruleName{app^{\mathcal{V}}_p}}

\newcommand{\ruleVTArrowI}{\ruleName{abs^{\mathcal{V}}_p}}

\newcommand{\ruleVTAxiom}{\ruleName{val_p}}

\newcommand{\ruleVTAxiomVar}{\ruleName{var_p}}

\newcommand{\ruleVTESubs}{\ruleName{es_p}}

\newcommand{\Rule}[3]{
    \prooftree
         #1
    \justifies  
         #2
    \thickness=0.05em
    \using
         #3
    \endprooftree}

\newcommand{\ih}{{\em i.h.~}}

\newcommand{\wsize}[1]{\ensuremath{|{#1}|_{\bangweak}}}
\newcommand{\nsize}[1]{\ensuremath{|{#1}|_{\callbyname}}}
\newcommand{\valsize}[1]{\ensuremath{|{#1}|_{\callbyvalue}}}

\newcommand{\countvr}[1]{\ensuremath{\mathtt{e}\!\left({#1}\right)}}

\newcommand{\inversevr}[1]{\ensuremath{\widehat{\mathtt{e}}\!\left({#1}\right)}}

\newcommand{\sequT}[5]{\ensuremath{{#1}\vdash^{({#3},{#4},{#5})}{#2}}}

\newcommand{\typeabs}{\ensuremath{\mathtt{a}}}
\newcommand{\typebang}{\ensuremath{\mathtt{vl}}}
\newcommand{\typeneutral}{\ensuremath{\mathtt{n}}}

\newcommand{\typetight}{\ensuremath{\mathtt{tt}}}

\newcommand{\typevalue}{\ensuremath{\mathtt{vl}}}
\newcommand{\typevar}{\ensuremath{\mathtt{vr}}}

\newcommand{\iptight}{\ensuremath{\mathsf{tight}}}

\newcommand{\ptight}[1]{\ensuremath{\iptight(#1)}}

\newcommand{\ipabs}{\ensuremath{\mathsf{abs}}}
\newcommand{\ipvar}{\ensuremath{\mathsf{var}}}
\newcommand{\ipval}{\ensuremath{\mathsf{val}}}
\newcommand{\ipapp}{\ensuremath{\mathsf{app}}}

\newcommand{\pabs}[1]{\ensuremath{\ipabs(#1)}}
\newcommand{\pvar}[1]{\ensuremath{\ipvar(#1)}}
\newcommand{\pval}[1]{\ensuremath{\ipval(#1)}}
\newcommand{\papp}[1]{\ensuremath{\ipapp(#1)}}

\newcommand{\icfnrml}{\ensuremath{\mathsf{no}_{\cfz}}}
\newcommand{\icfntrl}{\ensuremath{\mathsf{ne}_{\cfz}}}
\newcommand{\icfnbang}{\ensuremath{\mathsf{nb}_{\cfz}}}
\newcommand{\icfnabs}{\ensuremath{\mathsf{na}_{\cfz}}}
\newcommand{\cfnrml}[1]{\ensuremath{{#1}\in\icfnrml}}

\newcommand{\cbeta}{\ensuremath{\mathit{b}}}
\newcommand{\cmult}{\ensuremath{\mathit{m}}}
\newcommand{\cexp}{\ensuremath{\mathit{e}}}
\newcommand{\csize}{\ensuremath{\mathit{s}}}

\newcommand{\iI}{i \in I}
\newcommand{\cfz}{{\tt scf}}
\newcommand{\emul}{\intertype{\, }{}}

\newcommand{\CBNNF}{\ensuremath{\mathsf{no}_{\callbyname}}}
\newcommand{\CBVNF}{\ensuremath{\mathsf{no}_{\callbyvalue}}}
\newcommand{\HCBNNF}{\ensuremath{\mathsf{ne}_{\callbyname}}}
\newcommand{\HCBVNF}{\ensuremath{\mathsf{ne}_{\callbyvalue}}}
\newcommand{\VarHCBVNF}{\ensuremath{\mathsf{vr}_{\callbyvalue}}}
\newcommand{\sep}{\hspace{.5cm}}
\newcommand{\BangRev}{\texorpdfstring{\ensuremath{\lambda !}}{lambda!}}

\newcommand{\Kterm}{{\tt K}}
\newcommand{\id}{{\tt I}}
\newcommand{\ctxleq}{\subseteq}

\newcommand{\relevant}{relevant}
\newcommand{\bangrelevant}{\ensuremath{\bangweak}-relevant}
\newcommand{\cbnrelevant}{\ensuremath{\callbyname}-relevant}
\newcommand{\cbvrelevant}{\ensuremath{\callbyvalue}-relevant}

\booltrue{compileReport}

\begin{document}

\maketitle

\begin{abstract}
This paper explores how the intersection type theories of call-by-name (CBN)
and call-by-value (CBV) can be \emph{unified} in a more general framework
provided by call-by-push-value (CBPV). Indeed, we propose \emph{tight} type
systems for CBN and CBV that can be both encoded in a unique \emph{tight} type
system for CBPV. All such systems are quantitative, \ie they provide
\emph{exact} information about the length of normalization sequences to normal
form as well as the size of these normal forms. Moreover, the length of
reduction sequences are discriminated according to their multiplicative and
exponential nature, a concept inherited from linear logic. Last but not least,
it is possible to extract quantitative measures for CBN and CBV from their
corresponding encodings in CBPV.
\end{abstract}


\section{Introduction}
\label{s:intro}

Every programming language implements a particular evaluation strategy which
specifies when and how parameters are evaluated during function calls. For
example, in {\bf call-by-value (CBV)}, the argument is evaluated before being
passed to the function, while in {\bf call-by-name (CBN)} the argument is
substituted directly into the function body, so that the argument may never be
evaluated, or may be re-evaluated several times. CBN and CBV have always been
studied independently, by developing different techniques for one and the
other, until the remarkable observation that they are two different instances
of a more general framework introduced  by Girard's {\bf Linear Logic (LL)}.
Their (logical) duality (``CBN is \emph{dual} to CBV'') was understood
later~\cite{DanosJS95, CurienH00}. And their rewriting semantics were finally
\emph{unified} by the {\bf call-by-push-value (CBPV)} paradigm ---introduced
by P.B.~Levy~\cite{Levy04,Levy06}---, a formalism being able to capture
different functional languages/evaluation strategies. 

A typical aspect that one wants to compare between two different evaluation
strategies is the \emph{number of steps} that are necessary to get a result.
Such numbers are extracted from different models of computation
and should be then measured by \emph{compatible} instruments, either by
means of common quantitative tools, or by a precise transformation between
them\footnote{Think for example about cm and inches.}. Thus, we provide a
\emph{uniform} tool to measure quantitative information extracted from the
evaluation of programs in different programming languages (namely CBN and
CBV).

More concretely, we introduce typing systems capturing {\bf qualitative} and
{\bf quantitative} information of programs. From a qualitative point of view,
the typing systems characterize termination of programs, \ie a program $p$ is
typable if and only if $p$ terminates. From a quantitative point of view, the
typing systems provide {\bf exact/tight} information about the {\bf number of
steps} needed to get a result, as well as the {\bf size} of this result. All
these questions are addressed in the general framework of CBPV, being able to
encode CBN and CBV, both from a static and a dynamic point of view.
Dynamically, CBN and CBV evaluation strategies are known to be encodable by the
rewriting semantics of CBPV. Statically, we define tight typing systems
providing quantitative information for CBN and CBV, that can be seen as
particular cases of the tight quantitative typing system behind the unified
framework of CBPV. We now explain all these concepts in more detail.

\parrafo{Quantitative Types} Quantitative typing systems are often specified by
non-idempotent intersection types inspired by the relational semantics of
LL~\cite{Girard88,BucciarelliE01}, as pioneered by~\cite{Gardner94}. This
connection makes non-idempotent types not only a qualitative typing tool to
reason about programming languages, but mainly a quantitative one, being able
to specify properties related to the consumption of resources, a remarkable
investigation pioneered by the seminal de Carvalho's PhD
thesis~\cite{Carvalho:thesis} (see also~\cite{Carvalho18}). Thus,
qualitatively, a non-idempotent typing system is able to fully {\bf
characterise} normalisation, in the sense that a term $t$ is typable if and
only if $t$ is normalising. More interestingly, quantitative typing systems
also provide  {\bf upper bounds}, in the sense that the length of any
reduction sequence from $t$ to normal form \emph{plus} the size of this normal
form is \emph{bounded} by the size of the type derivation of $t$. Therefore,
typability characterises normalisation in a qualitative as well as in a
quantitative way, but only provides upper bounds. Several papers explore
bounded measures of non-idempotent types for different higher order languages.
Some references
are~\cite{Ehrhard12,Guerrieri18,AccattoliGL19,AccattoliG18,KesnerV17,BucciarelliKV17,Carvalho16,GuerrieriPF16,CarvalhoPF11,CarvalhoF16,BucciarelliKR21,MazzaPV18}.

In this satisfactory enough? A major observation concerning $\beta$-reduction
in $\lambda$-calculus is that the size of normal forms can be
\emph{exponentially} bigger than the number of steps needed to reach these
normal forms. This means that bounding the sum of these two integers \emph{at
the same} time is too rough, and not very relevant from a quantitative point of
view. Fortunately, it is possible to extract better (\ie independent and exact)
measures from a non-idempotent intersection type system. A crucial point to
obtain {\bf exact measures}, instead of upper bounds, is to consider
\emph{minimal} type
derivations~\cite{Carvalho:thesis,BernadetL13,CarvalhoPF11}. Therefore,
\emph{upper bounds} for time \emph{plus} size can be refined into
\emph{independent exact measures} for time \emph{and}
size~\cite{AccattoliGK18}. More precisely, the quantitative typing systems are
now equipped with constants and \emph{counters}, together with an appropriate
notion of {\bf tightness}, which encodes minimality of type derivations. For
any \emph{tight} type derivation $\Phi$ of a term $t$ ending with (independent)
counters $(\cbeta,\csize)$, it is now possible to show that $t$ is normalisable
in $\cbeta$ steps and its normal form has size $\csize$, so that the type
system is able to \emph{guess} the number of steps to normal form as well as
the size of this normal form. The opposite direction also holds: if $t$
normalises in $\cbeta$ steps to a normal form size $\csize$, then it is
possible to tightly type $t$ by using (independent) counters $(\cbeta,\csize)$.

In this paper we design tight quantitative type systems that are also capable
to discriminate between {\bf multiplicative} and {\bf exponential} evaluation
steps to normal form, two conceptual notions coming from LL:
\emph{multiplicative} steps are essentially those that (linearly) reconfigure
proofs/terms/programs, while \emph{exponential} steps are the only ones that
are potentially able to erase/duplicate other objects. As a consequence, for
any \emph{tight} type derivation $\Phi$ of $t$ ending with (independent)
counters $(\cmult,\cexp, \csize)$, the term $t$ is normalisable in $\cmult$
multiplicative and $\cexp$ exponential steps to a normal form having size
$\csize$. The opposite direction also holds.

\parrafo{Call-by-push-value} CBPV extends the $\lambda$-calculus with two
primitives $\mathtt{thunk}$ and $\mathtt{force}$ distinguishing between
\emph{values} and \emph{computations}: the former freezes the execution of a
term (\ie turns a computation into a value) while the latter fires again a
frozen term (\ie turns a value into a computation). These primitives allow to
capture the duality between CBN and CBV by conveniently labelling a
$\lambda$-term with $\mathtt{thunk}/\mathtt{force}$ to pause/resume the
evaluation of a subterm. Thus, CBPV provides a \emph{unique} and \emph{general}
formalism capturing different functional strategies, and allowing to
\emph{uniformly} study operational and denotational semantics of different
programming languages through a single tool.

In this paper we model the CBPV paradigm by using the
$\BangRev$-calculus~\cite{BucciarelliKRV20}, based on the \emph{bang calculus}
introduced in~\cite{EhrhardG16}, which in turn extends ideas by
T.~Ehrhard~\cite{Ehrhard16}. The granularity of the $\BangRev$-calculus,
expressed with both {\bf explicit substitutions} and {\bf reduction at a
distance} (details in Sec.~\ref{s:bang}), clearly allows to differentiate
between \emph{multiplicative} and \emph{exponential} steps, as in LL. The
corresponding CBN and CBV strategies also follow this pattern: it is possible
to distinguish the multiplicative steps that only reconfigure pieces of syntax,
from the exponential steps used to implement erasure and duplication of terms.

\parrafo{Contributions} We first define deterministic strategies for CBN and
CBV that are able to discriminate between multiplicative and exponential steps
(Sec.~\ref{s:cbname-cbvalue}).

We then formulate tight typing systems for both CBN (Sec.~\ref{s:tight-cbname})
and CBV (Sec.~\ref{s:tight-value}), called respectively $\SysTightCBN$ and
$\SysTightCBV$. System $\SysTightCBN$ is a direct extension of Gardner's
system~\cite{Gardner94}, while system $\SysTightCBV$ is completely new, and
constitutes one of the major contributions of this paper. A key feature of
system $\SysTightCBV$ is its ability to distinguish between the two different
roles that variables may play in CBV depending on the context where they are
placed, \ie to be a placeholder for a \emph{value} or the head of a
\emph{neutral} term. We show that both systems implement independent measures
for time and size, and that they are quantitatively sound and complete. More
precisely, we show that tight (\ie minimal) typing derivations in such systems
(exactly) quantitatively characterise normalisation, \ie if $\Phi$ is a
\emph{tight} type derivation of $t$ in system $\SysTightCBN$ (resp.
$\SysTightCBV$), ending with counters $(\cmult, \cexp, \csize)$, then there
exists a CBN (resp. CBV) normal form $p$ of size $\csize$ such that $t$ reduces
to $p$ by using exactly $\cmult$ multiplicative steps and $\cexp$ exponential
steps. The converse, giving quantitative completeness of the approach, also
holds.

Sec.~\ref{s:bang} recalls the bang calculus at a distance $\BangRev$ and
Sec.~\ref{s:tight} presents its associated tight type system $\SysTight$,
together with their respective quantitative sound and complete properties. The
untyped CBN/CBV translations into $\BangRev$ are recalled in
Sec.~\ref{s:cbn-cbv-embeddings}, while the typed translations are defined and
discussed in Sec.~\ref{s:tight-translations}, the other major contribution of
this work. Through these typed encodings,
  the counters of the
  source  and target derivations are related.
  This makes it possible to give  the precise  cost of our
  typed translations, as well as to extract  quantitative
measures for CBN and CBV from their corresponding encodings in CBPV.

\ifreport[Detailed proofs can be found in~\cite{KesnerV21:arxiv}.]{}

\parrafo{Related Work} For CBN, non-idempotent types were introduced
by~\cite{Gardner94}, their quantitative power was extensively studied
in~\cite{Carvalho:thesis,Carvalho16}, and their tight extensions
in~\cite{BernadetL13,AccattoliGK18}. For CBV, non-idempotent types were
introduced in~\cite{Ehrhard12}, and extensively studied,
\eg\cite{Guerrieri18,AccattoliG18}. A tight extension being able to count
reduction steps was recently defined in~\cite{AccattoliGL19,LeberlePhD} for a
special version of (closed) CBV, but it is not clear how this could be encoded
in a linear logic based CBPV framework. Another non-idempotent type system was
also recently introduced for CBV~\cite{ManzonettoPR19,KerinecMR21}, it is not
tight and does not translate to CBPV.

A (non-tight) quantitative type system for the bang calculus, based on a
relational model, can be found in~\cite{GuerrieriM18}. Another relational model
that can be seen as a non-tight system for CBPV was introduced
by~\cite{ChouquetT20}. Following ideas
in~\cite{Carvalho:thesis,BernadetL13,AccattoliGK18}, a type system
$\SysTightFlops$ was proposed in~\cite{BucciarelliKRV20} to fully exploit tight
quantitative aspects of the $\BangRev$-calculus: \emph{independent exact
measures} for time and size are guessed by the types system. However, no
relation between tightness for CBN/CBV and tightness for the
$\BangRev$-calculus are studied in {\it op.cite}. This papers fills this gap.

The discrimination between multiplicative and exponential steps by means of
tight quantitative types can be found \eg in CBN~\cite{AccattoliGK18},
call-by-need~\cite{AccattoliGL19}, and languages with pattern-matching
primitives~\cite{AlvesKV19}.


\section{Call-by-Name and Call-by-Value}
\label{s:cbname-cbvalue}

This section introduces the CBN and CBV specifications being able to
distinguish between multiplicative and exponential steps, as in linear logic.
Given a countably infinite set $\TermVariable$ of variables $x, y, z, \ldots$,
we consider the following grammars for terms ($\TermLambda$), values and
contexts:
\begin{center}
\begin{tabular}{c@{\!\!}c}
\hspace{-1em}
\begin{tabular}{r@{\enspace}r@{\enspace}c@{\enspace}l}
\emphdef{(Terms)}  & $t,u,r$ & $\Coloneq$  & $v \mid \termapp{t}{u} \mid \termsubs{x}{u}{t}$ \\
\emphdef{(Values)} & $v$   & $\Coloneq$  & $x \in \TermVariable \mid \termabs{x}{t}$ \\
\emphdef{(Contexts)} & $\ctxt{C}$ & $\Coloneq$ & $\ctxt{L} \mid \ctxt{N} \mid \ctxt{V}$ \\
\emphdef{(List Contexts)}  & $\ctxt{L}$  & $\Coloneq$  & $\Box \mid \termsubs{x}{t}{\ctxt{L}}$ \\
\emphdef{(CBN Contexts)}   & $\ctxt{N}$  & $\Coloneq$  & $\Box \mid \termapp{\ctxt{N}}{t} \mid \termabs{x}{\ctxt{N}} \mid \termsubs{x}{u}{\ctxt{N}}$ \\
\emphdef{(CBV Contexts)}   & $\ctxt{V}$  & $\Coloneq$  & $\Box \mid \termapp{\ctxt{V}}{t} \mid \termapp{t}{\ctxt{V}} \mid \termsubs{x}{u}{\ctxt{V}} \mid \termsubs{x}{\ctxt{V}}{t}$ 
\end{tabular}
\end{tabular}
\end{center}
A terms of the form $\termsubs{x}{u}{t}$ is a  \emphdef{closure}, and
$\exsubs{x}{u}$ an \emphdef{explicit substitution} (ES). Special terms are
$\id = \termabs{z}{z}$, $\Kterm = \termabs{x}{\termabs{y}{x}}$, $\Delta =
\termabs{x}{\termapp{x}{x}}$, and $\Omega = \termapp{\Delta}{\Delta}$. We use
$\ctxtapp{\ctxt{C}}{t}$ for the term obtained by replacing the hole $\Box$ of
$\ctxt{C}$ by $t$. \emphdef{Free} and \emphdef{bound} variables, as well as
$\alpha$-conversion, are defined as expected. In particular,
$\fv{\termsubs{x}{u}{t}} \eqdef \fv{t} \setminus \set{x} \cup \fv{u}$,
$\fv{\termabs{x}{t}} \eqdef \fv{t} \setminus \set{x}$,
$\bv{\termsubs{x}{u}{t}} \eqdef \bv{t} \cup \set{x} \cup \bv{u}$ and
$\bv{\termabs{x}{t}} \eqdef \bv{t} \cup \set{x}$. The notation
$\substitute{x}{u}{t}$ is used for the (capture-free) \emphdef{meta-level}
substitution operation, defined, as usual, modulo $\alpha$-conversion. Special
predicates are used to distinguish different kinds of terms surrounded by ES:
$\pabs{t}$ iff $t = \ctxtapp{\ctxt{L}}{\termabs{x}{u}}$, $\papp{t}$ iff $t =
\ctxtapp{\ctxt{L}}{\termapp{r}{u}}$ and $\pvar{t}$ iff $t =
\ctxtapp{\ctxt{L}}{x}$. Finally, $\pval{t}$ iff $\pabs{t}$ or $\pvar{t}$.

As mentioned in the introduction, our aim is to count the reduction steps by
distinguishing their multiplicative and exponential nature. To achieve this,
the standard specifications of CBN/CBV are not adequate, so  we need to
consider alternative appropriate definitions~\cite{AccattoliP12} making use of the following three
different rewriting rules:
\begin{center}
$
\begin{array}{rlll}
  \mbox{(\ttbf{d}istant \ttbf{B}eta)} & \termapp{\ctxtapp{\ctxt{L}}{\termabs{x}{t}}}{u}  & \rrule{\dBeta} & \ctxtapp{\ctxt{L}}{\termsubs{x}{u}{t}} \\ 
  \mbox{(\ttbf{s}ubstitute term)}& \termsubs{x}{u}{t}                               & \rrule{\sTerm} & \substitute{x}{u}{t} \\ 
  \mbox{(\ttbf{s}ubstitute \texttt{v}alue)} & \termsubs{x}{\ctxtapp{\ctxt{L}}{v}}{t}           & \rrule{\sVal}  & \ctxtapp{\ctxt{L}}{\substitute{x}{v}{t}}
\end{array}
$
\end{center}
Rule $\dB$ fires
$\beta$-reduction \emph{at a distance} by combining the two more elementary
rules: $\termapp{(\termabs{x}{t})}{u} \rrule{} \termsubs{x}{u}{t}$ and
$\termapp{\ctxtapp{\ctxt{L}}{t}}{u} \rrule{}
\ctxtapp{\ctxt{L}}{\termapp{t}{u}}$, where the second one is a
structural/permutation rule pushing out ES that may block $\beta$-redexes. Rule
$\sTerm$ implements standard substitution, while $\sVal$ restricts substitution
to values and acts \emph{at a distance} by combining the two more elementary
rules: $\termsubs{x}{v}{t} \rrule{} \substitute{x}{v}{t}$ and
$\termsubs{x}{\ctxtapp{\ctxt{L}}{v}}{t} \rrule{}
\ctxtapp{\ctxt{L}}{\termsubs{x}{v}{t}}$. The \emphdef{call-by-name} reduction
relation $\rewrite{\callbyname}$ is the closure by  contexts $\ctxt{N}$ of the
rules $\dBeta$ and $\sTerm$, while the \emphdef{call-by-value} reduction
relation $\rewrite{\callbyvalue}$ is the closure by  contexts $\ctxt{V}$ of the
rules $\dBeta$ and  $\sVal$. Equivalently,
\begin{center}
$\mathbin{\rewrite{\callbyname}} \coloneq \ctxt{N}(\rrule{\dBeta} \cup
\rrule{\sTerm})\quad$ and $\quad\mathbin{\rewrite{\callbyvalue}} \coloneq
\ctxt{V}(\rrule{\dBeta} \cup \rrule{\sVal})$
\end{center}
The resulting CBN/CBV formulations are now based on distinguished
\emphdef{multiplicative} (\cf $\dBeta$) and \emphdef{exponential} (\cf
$\sTerm$ and $\sVal$) steps, called resp. $\mStep$-steps and $\eStep$-steps,
thus inheriting the nature of cut elimination rules in LL. Notice that the
number of $\mStep$ and $\eStep$-steps in a normalization sequence is not always
the same: \eg $\termsubs{x}{y}{x} \rewrite{\callbyname} y$ has only
one $\eStep$-step, and $\termapp{(\termabs{x}{x})}{(\termapp{z}{\id})}
\rewrite{\callbyvalue} \termsubs{x}{\termapp{z}{\id}}{x}$ has only one
$\mStep$-step. We write $t \not\rewrite{\callbyname}$ (resp. $t \not
\rewrite{\callbyvalue}$), and call $t$ an \emphdef{$\callbyname$-normal form}
(resp. \emphdef{$\callbyvalue$-normal form}), if $t$ cannot be reduced by means
of $\rewrite{\callbyname}$ (resp. $\rewrite{\callbyvalue}$). Both CBN and CBV
are non-deterministic: $t \rewrite{\callbyname} u$ and $t \rewrite{\callbyname}
s$ does not necessarily implies $u = s$. But both calculi enjoy confluence,
notably because their rules are orthogonal~\cite{Klop80,Kha92}. Moreover, in
each calculus, it is easy to show that any two different reduction paths to
normal form have the same number of multiplicative and exponential steps.

CBN is to be understood as \emph{head}
reduction~\cite{Barendregt84}, \ie reduction does not take place in
arguments of applications, while CBV corresponds to \emph{open} CBV
reduction~\cite{AccattoliP12,AccattoliG16}, \ie reduction does not
take place inside abstractions.   The sets of
$\callbyname/\callbyvalue$-normal forms can be alternatively
characterised by the following grammars~\cite{BucciarelliKRV20}:
\begin{center}
\vspace{-1em}
\begin{tabular}{c@{\!\!}c}
\hspace{-1em}
\begin{tabular}{r@{\enspace}r@{\enspace}c@{\enspace}l}
\emphdef{(CBN Neutral)}   & $\HCBNNF$      & $\Coloneq$ & $x \in \TermVariable \mid \termapp{\HCBNNF}{t}$ \\
\emphdef{(CBN Normal)}    & $\CBNNF$       & $\Coloneq$ & $\termabs{x}{\CBNNF} \mid \HCBNNF$
\end{tabular}
&
\begin{tabular}{r@{\enspace}r@{\enspace}c@{\enspace}l}
\emphdef{(CBV Variable)}  & $\VarHCBVNF$   & $\Coloneq$ & $x \in \TermVariable \mid \termsubs{x}{\HCBVNF}{\VarHCBVNF}$ \\
\emphdef{(CBV Neutral)}   & $\HCBVNF$      & $\Coloneq$ & $\termapp{\VarHCBVNF}{\CBVNF} \mid \termapp{\HCBVNF}{\CBVNF} \mid \termsubs{x}{\HCBVNF}{\HCBVNF}$ \\
\emphdef{(CBV Normal)}    & $\CBVNF$       & $\Coloneq$ & $\termabs{x}{t} \mid \VarHCBVNF \mid \HCBVNF \mid \termsubs{x}{\HCBVNF}{\CBVNF}$
\end{tabular}
\end{tabular}
\end{center}

In contrast to CBN, variables are left out of the definition of neutral terms
for the CBV case, since they are now considered as values. However, even if
CBV variables are not neutral terms, neutral terms are necessarily headed by
a variable, so that variables play a double role which is difficult to be
distinguished by means of an intersection type system. We will come back to
this point in Sec.~\ref{s:tight-value}. Excluding variables from the set of
values brings a remarkable speed up in implementations of
CBV~\cite{LeberlePhD}, but goes beyond the logical Girard's
translation of CBV into LL, which is the main topic of this paper.
Our chosen approach allows both CBN and CBV neutral terms to translate to
neutral terms of the $\BangRev$-calculus (\cf Sec.~\ref{s:bang}).

\parrafo{Deterministic Strategies for CBN and CBV}
As a technical tool, in order to count the reduction steps of CBN/CBV we first
fix a deterministic version for them. The reduction relation $\loredname$ is a
deterministic version of $\rewrite{\callbyname}$ defined as: \[
  \begin{prooftree}
    \vphantom{\ctxt{L}\loredname}
    \justifies{\termapp{\ctxtapp{\ctxt{L}}{\termabs{x}{t}}}{u} \loredname
      \ctxtapp{\ctxt{L}}{\termsubs{x}{u}{t}}}
  \end{prooftree} \sep
  \begin{prooftree}
    \vphantom{\ctxt{L}\loredname}
    \justifies{\termsubs{x}{u}{t} \loredname \substitute{x}{u}{t}}
  \end{prooftree} \sep
  \begin{prooftree}
   t \loredname s \mbox{ and } \neg\pabs{t}
    \justifies{\termapp{t}{u} \loredname \termapp{s}{u}}
  \end{prooftree} \sep
  \begin{prooftree}
   t \loredname s
    \justifies{\termabs{x}{t} \loredname \termabs{x}{u}}
  \end{prooftree} \]
Similarly, the reduction relation $\loredv$ is a deterministic version of
$\rewrite{\callbyvalue}$ defined as:
\[
\begin{array}{c}
\Rule{\vphantom{\pabs{t}}}
     {\termapp{\ctxtapp{\ctxt{L}}{\termabs{x}{t}}}{u} \loredv \ctxtapp{\ctxt{L}}{\termsubs{x}{u}{t}}}
     {}
\qquad
\Rule{\vphantom{\pabs{t}}}
     {\termsubs{x}{\ctxtapp{\ctxt{L}}{v}}{t} \loredv \ctxtapp{\ctxt{L}}{\substitute{x}{v}{t}}}
     {}
\quad
\Rule{t \loredv s \quad \neg\pabs{t}}
     {\termapp{t}{u} \loredv \termapp{s}{u}}
     {}
\\ \\ 
\Rule{t \loredv s \quad u \in \HCBVNF \cup \VarHCBVNF}
     {\termapp{u}{t} \loredv \termapp{u}{s}}
     {}
\quad
\Rule{t \loredv s \quad \neg\pval{t}}
     {\termsubs{x}{t}{u} \loredv \termsubs{x}{s}{u}}
     {}
\qquad
\Rule{t \loredv s \quad u \in \HCBVNF}
     {\termsubs{x}{u}{t} \loredv \termsubs{x}{u}{s}}
     {}
\end{array}
\]

As a matter of notation, for $\X \in \set{\namestg,\valuestg}$, we write $t
\rewriten{\X}^{(\cmult,\cexp)} u$ if $t \rewriten{\X} u$ using $\cmult$
multiplicative steps and $\cexp$ exponential steps.

The normal forms of the non-deterministic and the deterministic versions of
CBN/CBV are the same, in turn characterised by the grammars
$\CBNNF$/$\CBVNF$~\cite{BucciarelliKRV20}.

\begin{proposition}
\label{l:cbn-cbv-normal-forms}
Let $t \in \TermLambda$. Then, $t \not\rewrite{\callbyname}$ iff $t
\not\loredname$ iff $t \in \CBNNF$ and $t \not\rewrite{\callbyvalue}$ iff
$t \not\loredv$ iff $t \in \CBVNF$.
\end{proposition}

\ifreport{\begin{proof} \mbox{}
\begin{enumerate}
  \item $t \not\rewrite{\callbyname}$ implies $t \not\loredname$ follows from
  $\rewrite{\namestg} \subset \rewrite{\callbyname}$.

  \item $t \in \CBNNF$ iff $t \not\rewrite{\callbyname}$ is shown
  simultaneously by considering the following statements:
  \begin{enumerate}
    \item\label{ldix:uno} $t \in \HCBNNF$ iff $t \not\rewrite{\callbyname}$
    and $\neg\pabs{t}$.

    \item\label{ldix:dos} $t \in \CBNNF$ iff $t \not\rewrite{\callbyname}$.
  \end{enumerate}
  \begin{enumerate}
    \item\label{case-non-lambda} The left-to-right implication is
    straightforward. For the right-to-left implication we reason by induction
    on terms. Suppose $t \not\rewrite{\callbyname}$ and $\neg\pabs{t}$. If
    $t$ is a substitution, then rule $\sTerm$ is applicable, contradiction.
    Then $t = \termapp{u}{u'}$, where $\neg\pabs{u}$ (otherwise $\dBeta$ would
    be applicable) and $u \not\rewrite{\callbyname}$ (otherwise $t$ would be
    $\callbyname$-reducible). The \ih (\ref{ldix:uno}) gives $u \in \HCBNNF$
    and thus we conclude $t = \termapp{u}{u'} \in \HCBNNF$.
    
    \item The left-to-right implication is straightforward. For the
    right-to-left implication we reason by induction on terms. If $t =
    \termabs{x}{u}$, then $u \not\rewrite{\callbyname}$, and the \ih
    (\ref{ldix:dos}) gives $u \in \CBNNF$, which implies in turn
    $\termabs{x}{u} \in \CBNNF$. Otherwise, we apply the previous case. 
  \end{enumerate}

  \item $t \not\loredname$ implies $t \in \CBNNF$ follows a straightforward
  adaptation of the previous right-to-left implication.

  \item $t \not\rewrite{\callbyvalue}$ implies $t \not\loredvn$ follows from
  $\rewrite{\valuestg} \subset \rewrite{\callbyvalue}$.

  \item $t \in \CBVNF$ iff $t \not\rewrite{\callbyvalue}$ is shown
  simultaneously by considering the following statements:
  \begin{enumerate}
    \item\label{ldix:dos-uno} $t \in \VarHCBVNF$ iff $t
    \not\rewrite{\callbyvalue}$ and $\neg\pabs{t}$ and $\neg\papp{t}$.
    
    \item\label{ldix:dos-dos} $t \in \HCBVNF$ iff $t
    \not\rewrite{\callbyvalue}$ and $\neg\pabs{t}$ and $\neg\pvar{t}$.
    
    \item\label{ldix:dos-tres} $t \in \CBVNF$ iff $t
    \not\rewrite{\callbyvalue}$.
  \end{enumerate}
  \begin{enumerate}
    \item The left-to-right implication is straightforward. For the
    right-to-left implication we reason by induction on terms. Suppose
    $t \not\rewrite{\callbyvalue}$ and $\neg\pabs{t}$ and $\neg\papp{t}$.
    Then necessarily $\pvar{t}$. If $t = \termsubs{x}{u'}{u}$, then $u, u'
    \not\rewrite{\callbyvalue}$ (otherwise $t$ would be
    $\callbyvalue$-reducible), $\neg\pabs{u'}$ and $\neg\pvar{u'}$ (otherwise
    $t$ would be $\sVal$-reducible), $\neg\pabs{u}$ and $\neg\papp{u}$. The \ih
    (\ref{ldix:dos-uno}) gives $u \in \VarHCBVNF$ and $u' \in \HCBVNF$, so that
    we conclude $t \in \VarHCBVNF$. If $t = x$ we trivially conclude $t \in
    \VarHCBVNF$.
    
    \item The left-to-right implication is straightforward. For the
    right-to-left implication we reason by induction on terms. Suppose $t
    \not\rewrite{\callbyvalue}$ and $\neg\pabs{t}$ and $\neg\pvar{t}$. Then
    necessarily $\papp{t}$. If $t = \termsubs{x}{u'}{u}$, then $u, u'
    \not\rewrite{\callbyvalue}$ (otherwise $t$ would be
    $\callbyvalue$-reducible), $\neg\pabs{u'}$ and $\neg\pvar{u'}$ (otherwise
    $t$ would be $\sVal$-reducible), $\neg\pabs{u}$ and $\neg\pvar{u}$. The \ih
    (\ref{ldix:dos-dos}) gives $u \in \HCBVNF$ and $u' \in \HCBVNF$, so that we
    conclude $t \in \HCBVNF$. If $t = \termapp{u}{u'}$, then $u, u'
    \not\rewrite{\callbyvalue}$ (otherwise $t$ would be
    $\callbyvalue$-reducible), $\neg\pabs{u}$ (otherwise $t$ would be
    $\dBeta$-reducible). Moreover, $\neg\papp{u}$ or $\neg\pvar{u}$ holds. The
    \ih (\ref{ldix:dos-tres}) gives $u' \in \CBVNF$ and ($u \in \VarHCBVNF$ or
    $u \in \HCBVNF$). In both cases we conclude $t = \termapp{u}{u'} \in
    \HCBVNF$. 
    
    \item The left-to-right implication is straightforward. For the
    right-to-left implication we reason by induction on terms. Suppose $t
    \not\rewrite{\callbyvalue}$. If $t = \termabs{x}{u}$, then $t \in \CBVNF$
    is straightforward. If $t = \termsubs{x}{u'}{u}$, then $u, u'
    \not\rewrite{\callbyvalue}$ (otherwise $t$ would be
    $\callbyvalue$-reducible), $\neg \pabs{u'}$ and $\neg\pvar{u'}$ (otherwise
    $t$ would be $\sVal$-reducible). The \ih (\ref{ldix:dos-tres}) and
    (\ref{ldix:dos-dos}) give $u \in \CBVNF$ and $u' \in \HCBVNF$, so that we
    conclude $t \in \CBVNF$. If $t = \termapp{u}{u'}$, then $u
    \not\rewrite{\callbyvalue}$ and $u' \not \rewrite{\callbyvalue}$. The \ih
    (\ref{ldix:dos-tres}) on $u'$ gives $u' \in \CBVNF$. Moreover,
    $\neg\pabs{u}$, otherwise $t$ would be $\callbyvalue$-reducible. Then
    $\neg\pvar{u}$ or $\neg\papp{u}$. The \ih (\ref{ldix:dos-uno}) or
    (\ref{ldix:dos-dos}) gives $u \in \VarHCBVNF$ or $u \in \HCBVNF$. Thus
    $\termapp{u}{u'} \in \HCBVNF\subseteq \CBVNF$ as required.
  \end{enumerate}
  
  \item $t \not\loredv$ implies $t \in \CBVNF$ follows a straightforward
  adaptation of the previous right-to-left implication.
\end{enumerate}
\end{proof}

}

(Head) CBN ignores reduction inside arguments of applications, while (Open) CBV
ignores reduction inside abstractions, then CBN (resp. CBV) normal forms are
measured by the following \emphdef{$\callbyname$-size} (resp.
\emphdef{$\callbyvalue$-size}) function:
\begin{center}
$
\begin{array}{llllll}
\nsize{x} \coloneq 0 &
\nsize{\termabs{x}{t}}\coloneq \nsize{t} + 1 & 
\nsize{\termapp{t}{u}} \coloneq \nsize{t} + 1 &
\nsize{\termsubs{x}{u}{t}} \coloneq \nsize{t}
\\
\valsize{x}\coloneq 0 & 
\valsize{\termabs{x}{t}} \coloneq 0 &
\valsize{\termapp{t}{u}} \coloneq \valsize{t} + \valsize{u} + 1 &
\valsize{\termsubs{x}{u}{t}} \coloneq \valsize{t} + \valsize{u}
\end{array}
$
\end{center}


\section{Tight Call-by-Name}
\label{s:tight-cbname}

We now introduce a tight type system $\SysTightCBN$ for CBN which captures
independent exact measures for $\rewrite{\callbyname}$-reduction sequences.
This result is not surprising, since it extends the one in~\cite{AccattoliGK18}
from the pure $\lambda$-calculus to our CBN calculus. However, we revisit the
op.cit. approach by appropriately splitting reduction into multiplicative and
exponential steps, a reformulation  necessary to establish a
precise correspondence with the tight type system for the $\BangRev$-calculus.
In particular, our \emph{counting} mechanism slightly differs
from~\cite{AccattoliGK18} (details below).

In system $\SysTightCBN$ there are two base types: 
$\typeabs$ types terms whose normal form is an abstraction, and 
$\typeneutral$ types terms whose normal form is CBN neutral. The grammar of
types is given by:
\begin{center}
\begin{tabular}{rrcll}
\emphdef{(Tight Types)} & $\typetight$   & $\Coloneq$ & $\typeneutral \mid \typeabs$ \\
\emphdef{(Types)}       & $\sigma, \tau$ & $\Coloneq$ & $\typetight \mid \M \mid \functtype{\M}{\sigma}$ \\
\emphdef{(Multitypes)}  & $\M, \N$       & $\Coloneq$ & $\intertype{\sigma_i}{i \in I}$  where $I$ is a finite set
\end{tabular}
\end{center}
Multitypes are multisets of types. The \emphdef{empty multitype} is denoted by
$\emul$, $\sqcup$  denotes multitype union, and $\sqsubseteq$ multitype
inclusion. Also, $|\M|$ denotes the size of the multitype, thus if $\M =
\intertype{\sigma_i}{i \in I}$ then $|\M| = \#(I)$.
Notice that the grammar for types slightly differs from~\cite{AccattoliGK18},
in particular  types are now allowed to be just multitypes. The main reason to
adopt this change is that this (unique) grammar is used for our three
formalisms CBN, CBV, and CBPV, changing only the definition of tight types for
each case\footnote{In the intersection type literature CBN and CBV do always
adopt different grammars.}.

\emphdef{Typing contexts} (or just \emphdef{contexts}), written $\Gamma,
\Delta$, are functions from variables to multitypes, assigning the empty
multitype to all but a finite set of variables. The domain of $\Gamma$ is given
by $\dom{\Gamma} \eqdef \set{x \mid \Gamma(x) \neq \emul}$. The \emphdef{union
of contexts}, written $\ctxtsum{\Gamma}{\Delta}{}$, is defined by
$(\ctxtsum{\Gamma}{\Delta}{})(x) \eqdef \Gamma(x) \sqcup \Delta(x)$. An example
is $\ctxtsum{(\assign{x}{\intertype{\sigma}{}},
\assign{y}{\intertype{\tau}{}})}{(\assign{x}{\intertype{\sigma}{}},
\assign{z}{\intertype{\tau}{}})}{} = (\assign{x}{\intertype{\sigma, \sigma}{}},
\assign{y}{\intertype{\tau}{}}, \assign{z}{\intertype{\tau}{}})$. This notion
is extended to several contexts as expected, so that $\ctxtsum{}{}{i \in I}
\Gamma_i$ denotes a finite union of contexts (particularly the empty context
when $I = \emptyset$). We write $\ctxtres{\Gamma}{x}{}$ for the context
$(\ctxtres{\Gamma}{x}{})(x) = \emul$ and $(\ctxtres{\Gamma}{x}{})(y) =
\Gamma(y)$ if $y \neq x$.

\emphdef{Type judgements} have the form
$\sequT{\Gamma}{\assign{t}{\sigma}}{\cmult}{\cexp}{\csize}$, where $\Gamma$ is
a typing context, $t$ is a term, $\sigma$ is a type, and the \emphdef{counters}
$(\cmult, \cexp, \csize)$ are expected to provide the following information:
$\cmult$ (resp.  $\cexp$) indicates the number of multiplicative $\mStep$-steps
(resp.  exponential $\eStep$-steps) to normal form, while $\csize$ indicates
the $\callbyname$-size of this normal form. It is worth noticing that the
$\lambda$-calculus hides both multiplicative and
exponential steps in one single $\beta$-reduction rule~\cite{AccattoliGK18}, so
that only two counters suffice, one for the number of $\beta$-reduction steps,
and another for the size of normal forms. Here we want to discriminate between
multiplicative/exponential steps, in the sense that the execution of an ES
generates an exponential step, but not a multiplicative one. This becomes
possible due to the CBN/CBV alternative specifications with ES that we have
adopted, and that is why we need three independent counters, in contrast to the
$\lambda$-calculus.

The type system $\SysTightCBN$ for CBN is given in
Fig.~\ref{fig:typingSchemesNNPersistent} (\emph{persistent} rules)
and~\ref{fig:typingSchemesNNConsuming} (\emph{consuming} rules). A constructor
is consuming (resp. persistent) if it is consumed (resp. not consumed) during
$\callbyname$-reduction. For instance, in
$\termapp{\termapp{\Kterm}{\id}}{\Omega}$ the two abstractions of $\Kterm$ are
consuming, while the abstraction of $\id$ is persistent, and all the other
constructors are also consuming, except those of $\Omega$ that turns out to be
an untyped subterm. The persistent rules
(Fig.~\ref{fig:typingSchemesNNPersistent}) are those typing persistent
constructors, so that none of them increases the first two counters, but only
possibly the third one, which contributes to the size of the normal form. The
consuming rules (Fig.~\ref{fig:typingSchemesNNConsuming}), in contrast, type
consuming constructors, so that they may increase the first two counters,
contributing to the length of the normalisation sequence. Notice in particular
that there are only two rules contributing to the multiplicative/exponential
counting:
\begin{inparaenum}[(1)]
  \item rule $\ruleNDArrowE$ types a \emph{consuming} application, meaning that
  its left-hand side subterm reduces to an abstraction, then causing a
  multiplicative step (rule $\dBeta$)~\footnote{In both~\cite{AccattoliGK18}
  and~\cite{KesnerV20}, it is the consuming abstraction which contributes to
  the multiplicative steps.} followed later by an exponential step; while
  \item rule $\ruleNDESubs$ types a \emph{consuming} substitution causing an
  exponential step. In both cases, it is the constructor
  application/substitution which is considered to be consumed, without any
  further hypothesis on the form of the subterm that appears inside this
  consuming constructor (recall that any term can be substituted in CBN).
  This phenomenon facilitates in particular the identification of exponential
  steps in CBN, in contrast to CBV and CBPV.
\end{inparaenum}

\begin{figure}[ht]
\centering $
\begin{array}{c}
\Rule{\sequT{\Gamma}{\assign{t}{\typeneutral}}{\cmult}{\cexp}{\csize}
      }
     {\sequT{\Gamma}{\assign{\termapp{t}{u}}{\typeneutral}}{\cmult}{\cexp}{\csize+1}}
     {\ruleNTArrowE}
       \qquad
\Rule{\sequT{\Gamma}{\assign{t}{\typetight}}{\cmult}{\cexp}{\csize}
      \quad
      \ptight{\Gamma(x)}
     }
     {\sequT{\ctxtres{\Gamma}{x}{}}{\assign{\termabs{x}{t}}{\typeabs}}{\cmult}{\cexp}{\csize+1}}
     {\ruleNTArrowI}
\end{array}$
\caption{System $\SysTightCBN$ for the Call-by-Name Calculus: Persistent Typing Rules.}
\label{fig:typingSchemesNNPersistent}
\end{figure}
\begin{figure}[ht]
\centering $
\kern-2em
\begin{array}{c}
\Rule{\vphantom{\Gamma}}
     {\sequT{\assign{x}{\multiset{\sigma}}}{\assign{x}{\sigma}}{0}{0}{0}}
     {\ruleNDAxiom}
\qquad
\Rule{\sequT{\Gamma}{\assign{t}{\tau}}{\cmult}{\cexp}{\csize}}
     {\sequT{\ctxtres{\Gamma}{x}{}}{\assign{\termabs{x}{t}}{\functtype{\Gamma(x)}{\tau}}}{\cmult}{\cexp}{\csize}}
     {\ruleNDArrowI}
\\
\\
\Rule{\sequT{\Gamma}{\assign{t}{\functtype{\intertype{\sigma_i}{i \in I}}{\tau}}}{\cmult}{\cexp}{\csize}
      \quad
      \many{\sequT{\Delta_i}{\assign{u}{\sigma_i}}{\cmult_i}{\cexp_i}{\csize_i}}{i \in I}
     }
     {\sequT{\ctxtsum{\Gamma}{\Delta}{}}{\assign{\termapp{t}{u}}{\tau}}{1+\cmult+_{\iI}{\cmult_i}}{1+\cexp+_{\iI}{\cexp_i}}{\csize+_{\iI}{\csize_i}}}
     {\ruleNDArrowE}
\\
\\
\Rule{\sequT{\Gamma;\assign{x}{\intertype{\sigma_i}{i \in I}}}{\assign{t}{\tau}}{\cmult}{\cexp}{\csize}
      \enspace
      \many{\sequT{\Delta_i}{\assign{u}{\sigma_i}}{\cmult_i}{\cexp_i}{\csize_i}}{i \in I}
     }
     {\sequT{\ctxtsum{(\ctxtres{\Gamma}{x}{})}{\Delta_i}{i \in I}}{\assign{\termsubs{x}{u}{t}}{\tau}}{\cmult+_{\iI}{\cmult_i}}{1+\cexp+_{\iI}{\cexp_i}}{\csize+_{\iI}{\csize_i}}}
     {\ruleNDESubs}
\end{array} $
\caption{System $\SysTightCBN$ for the Call-by-Name Calculus: Consuming Typing Rules.}
\label{fig:typingSchemesNNConsuming}
\end{figure}

This dichotomy between consuming/persistent constructors has been first used
in~\cite{KesnerV20} for the $\lambda$ and $\lambda\mu$-calculi, and adapted
here for our distant versions of CBN/CBV as well as for the
$\BangRev$-calculus. We write
$\derivable{}{\sequT{\Gamma}{\assign{t}{\sigma}}{\cmult}{\cexp}{\csize}}{\SysTightCBN}$
if there is a \emphdef{(tree) type derivation} of the judgement
$\sequT{\Gamma}{\assign{t}{\sigma}}{\cmult}{\cexp}{\csize}$ in system
$\SysTightCBN$. The term $t$ is typable in system $\SysTightCBN$, or
$\SysTightCBN$-typable, iff there is a context $\Gamma$, a type $\sigma$ and
counters $(\cmult,\cexp,\csize)$ such that
$\derivable{}{\sequT{\Gamma}{\assign{t}{\sigma}}{\cmult}{\cexp}{\csize}}{\SysTightCBN}$.
We use the capital Greek letters $\Phi, \Psi, \ldots$ to name type derivations,
by writing for example
$\derivable{\Phi}{\sequT{\Gamma}{\assign{t}{\sigma}}{\cmult}{\cexp}{\csize}}{\SysTightCBN}$.
As (local) counters of judgements in a given derivation $\Phi$ contribute to
the global counters of the derivation itself, there is an alternative way to
define counters associated to a derivation $\Phi$: the first counter
$\cmult_{\Phi}$ is given by the number of rules $\ruleNDArrowE$ in $\Phi$, the
second counter $\cexp_{\Phi}$ is the number of rules $\ruleNDESubs$ in $\Phi$
and finally the third counter $\csize_{\Phi}$ is the number of rules
$\ruleNTArrowE$ and $\ruleNTArrowI$ in $\Phi$. We prefer however to
systematically write counters in judgements to easy the understanding of the
examples and proofs.

A \emphdef{multitype} $\intertype{\sigma_i}{i \in I}$ is \emphdef{tight},
written $\ptight{\intertype{\sigma_i}{i \in I}}$, if $\sigma_i \in \typetight$
for all $i \in I$. A \emphdef{context} $\Gamma$ is said to be \emphdef{tight}
if it assigns tight multitypes to all variables. A \emphdef{type derivation}
$\derivable{\Phi}{\sequT{\Gamma}{\assign{t}{\sigma}}{\cmult}{\cexp}{\csize}}{\SysTight}$
is \emphdef{tight} if $\Gamma$ is tight and $\sigma \in \typetight$.

The proofs of soundness and completeness related to our CBN type system are
respectively based on subject reduction and expansion properties, and they are
very similar to those in~\cite{AccattoliGK18}. The most important point to be
mentioned is that system $\SysTightCBN$ is now counting \emph{separately} the
multiplicative and exponential steps of $\rewrite{\callbyname}$-reductions to
normal-form. 

\parrafo{Soundness}
The soundness property is based on a series of auxiliary results.

\begin{lemma}[Tight Spreading]
Let
$\derivable{\Phi}{\sequT{\Gamma}{\assign{t}{\sigma}}{\cmult}{\cexp}{\csize}}{\SysTightCBN}$
such that $\Gamma$ is tight.
If $t \in\HCBNNF$, then $\sigma \in \typetight$.
\label{l:name:tight-spreading}
\end{lemma}

\begin{proof}
By induction on $\Phi$.
\begin{itemize}
  \item $\ruleNTArrowE$. Then $\sigma \in \typetight$ by definition.

  \item $\ruleNTArrowI$. This case does not apply since $t \notin \HCBNNF$.

  \item $\ruleNDAxiom$. This case is immediate since $\Gamma = \multiset{\sigma}$
  tight implies $\sigma \in \typetight$.
  
  \item $\ruleNDArrowE$. Then $t = \termapp{r}{u}$. If $t \in \HCBNNF$, then $r
  \in \HCBNNF$. We have a derivation
  $\derivable{\Phi_{r}}{\sequT{\Gamma'}{\assign{r}{\multiset{\functtype{\M}{\sigma}}}}{\cmult'}{\cexp'}{\csize'}}{\SysTightCBN}$,
  with $\Gamma' \ctxleq \Gamma$, so that $\Gamma'$ is also tight. The \ih gives 
  $\multiset{\functtype{\M}{\sigma}} \in \typetight$, and this is a
  contradiction. Hence, this case does not apply.

  \item $\ruleNDArrowI$ and $\ruleNDESubs$. These cases do not apply since $t
  \notin \HCBNNF$. 
\end{itemize}  
\end{proof}


\begin{lemma}
Let
$\derivable{\Phi}{\sequT{\Gamma}{\assign{t}{\sigma}}{\cmult}{\cexp}{\csize}}{\SysTightCBN}$
tight. Then, $\cmult = \cexp = 0$ iff $t \in  \CBNNF$.
\label{l:name:czero-normal}
\end{lemma}

\begin{proof}
$\left.\Rightarrow\right)$ We prove simultaneously the following statements by
induction on $\Phi$.
\begin{enumerate}
  \item\label{l:name:czero-normal:vr}
  $\derivable{\Phi}{\sequT{\Gamma}{\assign{t}{\sigma}}{0}{0}{\csize}}{\SysTightCBN}$
  tight and $\pvar{t}$ or $\papp{t}$ implies $t \in \HCBNNF$.
  
  \item\label{l:name:czero-normal:no}
  $\derivable{\Phi}{\sequT{\Gamma}{\assign{t}{\sigma}}{0}{0}{\csize}}{\SysTightCBN}$
  tight and $\pabs{t}$ implies $t \in \CBNNF$.
\end{enumerate}

\begin{itemize}
  \item $\ruleNDAxiom$. Then $t = x$ and $\pvar{t}$. By definition $x \in
  \HCBVNF$ trivially holds.
  
  \item $\ruleNTArrowE$. Then $t = \termapp{r}{u}$, and
  $\sequT{\Gamma}{\assign{r}{\typeneutral}}{0}{0}{\csize-1}$ is also tight.
  We are in the case $\papp{t}$ so we need to show $t \in \HCBNNF$. The \ih
  (\ref{l:name:czero-normal:vr}) gives $r \in \HCBNNF$ and we thus conclude
  $\termapp{r}{u} \in \HCBNNF$.
  
  \item $\ruleNTArrowI$. Then $t = \termabs{x}{r}$, and the derivation typing
  $r$ is also tight and has the two first counters equal to $0$. The \ih gives
  $r \in \CBNNF$, then we conclude $t \in \CBNNF$. 
  
  \item $\ruleNDArrowE$, $\ruleNDESubs$. The hypothesis about the counters is
  not verified so these cases do not apply. 
  
  \item $\ruleNDArrowI$. This case does not apply since the derivation cannot
  be tight. 
\end{itemize}

$\left.\Leftarrow\right)$ By induction on $t$.
\begin{itemize}
  \item $t = x$. Then, $\Phi$ necessarily ends with rule $\ruleNDAxiom$ and the
  statement trivially holds.
  
  \item $t = \termapp{r}{u}$. By definition $t \in \CBNNF$ gives $r \in
  \HCBNNF$. There are two cases to consider:
  \begin{enumerate}
    \item if $\Phi$ ends with rule $\ruleNTArrowE$, then $\sigma =
    \typeneutral$, and
    $\sequT{\Gamma}{\assign{r}{\typeneutral}}{\cmult}{\cexp}{\csize}$. The \ih
    gives $\cmult = \cexp = 0$ and we are done.
    
    \item if $\Phi$ ends with rule $\ruleNDArrowE$, then
    $\derivable{\Phi_{r}}{\sequT{\Gamma'}{\assign{r}{\multiset{\functtype{\M}{\tau}}}}{\cmult'}{\cexp'}{\csize'}}{\SysTightCBN}$
    where $\Gamma'$ is tight. Moreover, Lem.~\ref{l:name:tight-spreading} gives
    $\multiset{\functtype{\M}{\tau}} \in \typetight$ which is a contradiction.
    Hence, this case does not apply.
  \end{enumerate}
  
  \item $t = \termabs{x}{r}$. Then $r \in \CBNNF$ and there are two cases to
  consider for $\Phi$:
  \begin{enumerate}
    \item if $\Phi$ ends with rule $\ruleNTArrowI$, then the subderivation of
    $\Phi$ is tight, and thus property holds by the \ih 
    
    \item if $\Phi$ ends with rule $\ruleNDArrowI$, then $\Phi$ cannot be tight
    so this case does not apply.
  \end{enumerate}
  
  \item $t = \termsubs{x}{u}{r}$. By definition $t \notin \CBNNF$ so this cases
  does not apply. 
\end{itemize}
\end{proof}


\begin{lemma}
Let
$\derivable{\Phi}{\sequT{\Gamma}{\assign{t}{\sigma}}{0}{0}{\csize}}{\SysTightCBN}$.
If $\Phi$ is tight and $t \in \CBNNF$, then $\csize = \nsize{t}$.
\label{l:name:tight-size}
\end{lemma}

\begin{proof}
By induction on $\Phi$.
\begin{itemize}
  \item $\ruleNDAxiom$. Then $t = x$ and $\csize = 0 = \nsize{t}$ as expected.
  
  \item $\ruleNTArrowE$. Then $t = \termapp{r}{u}$, $\sigma = \typeneutral$,
  and $\sequT{\Gamma}{\assign{r}{\typeneutral}}{0}{0}{\csize-1}$. The \ih gives
  $\csize - 1 = \nsize{r}$ and hence $\csize = \nsize{r} + 1 = \nsize{t}$.
  
  \item $\ruleNTArrowI$. Then $t = \termabs{x}{r}$, and the subderivation
  typing $r$ is tight and has counter $\csize - 1$. The \ih gives $\csize - 1 =
  \nsize{r}$ and hence $\csize = \nsize{r} + 1 = \nsize{t}$.
  
  \item $\ruleNDArrowE$. Then first two counters are not $0$ so this case does
  not apply. 

  \item $\ruleNDArrowI$. Then $\Phi$ is not tight so this case does not apply. 
  
  \item $\ruleNDESubs$. Then the second counter is not $0$ so this cases does
  not apply.
\end{itemize}
\end{proof}


\begin{lemma}[Substitution]
Let
$\derivable{\Phi_{t}}{\sequT{\Gamma; x:\multiset{\sigma_i}_{\iI}}{\assign{t}{\tau}}{\cmult}{\cexp}{\csize}}{\SysTightCBN}$
and 
$(\derivable{\Phi_{u}}{\sequT{\Delta_i}{\assign{u}{\sigma_i}}{\cmult_i}{\cexp_i}{\csize_i}}{\SysTightCBN})_{\iI}$.
Then,
$\derivable{\Phi_{\substitute{x}{u}{t}}}{\sequT{\ctxtsum{\Gamma}{\Delta_i}{\iI}}{\assign{\substitute{x}{u}{t}}{\tau}}{\cmult+_{\iI}\cmult_i}{\cexp+_{\iI}\cexp}{\csize+_{\iI}\csize_i}}{\SysTightCBN}$.
\label{l:name:substitution-tight}
\end{lemma}

\begin{proof}
By induction on $\Phi_t$.
\begin{itemize}
  \item $\ruleNDAxiom$. There are two cases.
  \begin{itemize}
    \item  $t = x$. Then  $\substitute{x}{u}{t} = u$ and $|I| = 1$. More
    precisely,
    $\derivable{\Phi_{t}}{\sequT{\assign{x}{\multiset{\sigma_1}}}{\assign{x}{\sigma_1}}{0}{0}{0}}{\SysTightCBN}$.
    We conclude by taking $\Phi_{\substitute{x}{u}{t}} = \Phi^1_{u}$ since the
    counters are as expected.
  
    \item $t = y \neq x$. Then, $\substitute{x}{u}{t} = y$ and $|I| = 0$. More
    precisely, we have
    $\derivable{\Phi_{t}}{\sequT{\assign{y}{\multiset{\tau}}}{\assign{y}{\tau}}{0}{0}{0}}{\SysTightCBN}$.
    Hence we conclude by taking $\Phi_{\substitute{x}{u}{t}} = \Phi_{t}$ and
    the counters are as expected.
  \end{itemize}
  
  \item $\ruleNTArrowE$. Then $t = \termapp{r}{s}$, $\substitute{x}{s}{t} =
  \termapp{\substitute{x}{v}{r}}{\substitute{x}{v}{s}}$ and $\tau =
  \typeneutral$. We have a premise
  $\derivable{\Phi_{r}}{\sequT{\Gamma; x:\multiset{\sigma_i}_{\iI}}{\assign{r}{\typeneutral}}{\cmult}{\cexp}{\csize-1}}{\SysTightCBN}$.
  By the \ih we get a type derivation
  $\derivable{\Phi_{\substitute{x}{u}{r}}}{\sequT{\ctxtsum{\Gamma}{\Delta_i}{\iI}}{\assign{\substitute{x}{u}{r}}{\typeneutral}}{\cmult+_{\iI}\cmult_i}{\cexp+_{\iI}\cexp}{\csize-1+_{\iI}\csize_i}}{\SysTightCBN}$.
  We conclude by applying $\ruleNTArrowE$ to this premise, thus obtaining
  $\derivable{\Phi_{\substitute{x}{u}{t}}}{\sequT{\ctxtsum{\Gamma}{\Delta_i}{\iI}}{\assign{\substitute{x}{u}{t}}{\typeneutral}}{\cmult+_{\iI}\cmult_i}{\cexp+_{\iI}\cexp}{\csize+_{\iI}\csize_i}}{\SysTightCBN}$ as required.
  
  \item All the other cases are straightforward by induction.
\end{itemize}
\end{proof}


\begin{lemma}[Exact Subject Reduction]
Let
$\derivable{\Phi}{\sequT{\Gamma}{\assign{t}{\sigma}}{\cmult}{\cexp}{\csize}}{\SysTightCBN}$
be a tight derivation. If $t \loredname t'$, then there is
$\derivable{\Phi'}{\sequT{\Gamma}{\assign{t'}{\sigma}}{\cmult'}{\cexp'}{\csize}}{\SysTightCBN}$
such that
\begin{enumerate}
  \item\label{l:name:subject-reduction-tight:cmult} $\cmult' = \cmult - 1$ and
  $\cexp' = \cexp$ if $t \loredname t'$ is an $\mStep$-step.
  
  \item\label{l:name:subject-reduction-tight:cexo} $\cexp' = \cexp - 1$ and
  $\cmult' = \cmult$ if $t \loredname t'$ is an $\eStep$-step.
\end{enumerate}
\label{l:name:subject-reduction-tight}
\end{lemma}

\begin{proof}
We actually prove the following stronger statement that allows us to reason
inductively:

Let $t \loredname t'$ and
$\derivable{\Phi}{\sequT{\Gamma}{\assign{t}{\sigma}}{\cmult}{\cexp}{\csize}}{\SysTightCBN}$
such that $\Gamma$ is tight, and either $\sigma \in \typetight$ or
$\neg\pabs{t}$. Then, there exists
$\derivable{\Phi'}{\sequT{\Gamma}{\assign{t'}{\sigma}}{\cmult'}{\cexp'}{\csize}
}{\SysTightCBN}$
such that
\begin{enumerate}
  \item $\cmult' = \cmult - 1$ and $\cexp' = \cexp$ if $t \loredname t'$ is an
  $\mStep$-step.
  
  \item $\cexp' = \cexp - 1$ and $\cmult' = \cmult$ if $t \loredname t'$ is an
  $\eStep$-step.
\end{enumerate}

We proceed by induction on $t \loredname t'$.
\begin{itemize}
  \item $t = \termapp{\ctxtapp{\ctxt{L}}{\termabs{x}{r}}}{u}
  \loredname \ctxtapp{\ctxt{L}}{\termsubs{x}{u}{r}} = t'$ is an $\mStep$-step.
  We proceed by induction on $\ctxt{L}$. We only show here the case $\ctxt{L} =
  \Box$ as the inductive case is straightforward.

  We first remark that $\Phi$ cannot end with rule $\ruleBTArrowE$ since
  $\termabs{x}{u}$ cannot be typed with $\typeneutral$, then $\Phi$ ends with
  $\ruleNDArrowE$ and  has the following form: \[
\Rule{
  \Rule{
    \derivable{\Phi_{r}}{\sequT{\Gamma';\assign{x}{\multiset{\sigma_i}_{\iI}}}{\assign{r}{\sigma}}{\cmult_1}{\cexp_1}{\csize_1}}{\SysTightCBN}
  }{
    \sequT{\Gamma'}{\assign{\termabs{x}{r}}{\functtype{\multiset{\sigma_i}_{\iI}}{\sigma}}}{\cmult_1}{\cexp_1}{\csize_1}
  }{\ruleNDArrowI}
  \quad
  (\derivable{\Phi^i_{u}}{\sequT{\Delta_i}{\assign{u}{\sigma_i}}{\cmult^i_2}{\cexp^i_2}{\csize^i_2}}{\SysTightCBN})_{\iI}
}{
  \sequT{\ctxtsum{\Gamma'}{\Delta_i}{\iI}}{\assign{\termapp{(\termabs{x}{r})}{u}}{\sigma}}{1+\cmult_1+_{\iI}{\cmult^i_2}}{1+\cexp_1+_{\iI}{\cexp^i_2}}{\csize_1+_{\iI}{\csize^i_2}}
}{\ruleNDArrowE}
  \] with $\cmult = 1 + \cmult_1 +_{\iI}{\cmult^i_2}$, $\cexp = 1 + \cexp_1
  +_{\iI}{\cexp^i_2}$ and $\csize = \csize_1 +_{\iI}{\csize^i_2}$. We conclude
  by $\ruleNDESubs$ with the following type derivation \[
\Rule{
  \derivable{\Phi_{r}}{\sequT{\Gamma';\assign{x}{\multiset{\sigma_i}_{\iI}}}{\assign{r}{\sigma}}{\cmult_1}{\cexp_1}{\csize_1}}{\SysTightCBN}
  \quad
  (\derivable{\Phi^i_{u}}{\sequT{\Delta_i}{\assign{u}{\sigma_i}}{\cmult^i_2}{\cexp^i_2}{\csize^i_2}}{\SysTightCBN})_{\iI}
}{
  \derivable{\Phi'}{\sequT{\ctxtsum{\Gamma'}{\Delta_i}{\iI}}{\assign{\termsubs{x}{u}{r}}{\sigma}}{\cmult_1+_{\iI}{\cmult^i_2}}{1+\cexp_1+_{\iI}{\cexp^i_2}}{\csize_1+_{\iI}{\csize^i_2}}}{\SysTightCBN}
}{\ruleNDESubs}
  \] taking $\cmult' = \cmult_1 +_{\iI}{\cmult^i_2} = \cmult - 1$ and $\cexp' =
  1 + \cexp_1+_{\iI}{\cexp^i_2} = \cexp$.
  
  \item $t = \termsubs{x}{u}{r} \loredname \substitute{x}{u}{r} = t'$ is an
  $\eStep$-step. Then, rule $\ruleNDESubs$ is necessarily used \[
\Rule{
  \derivable{\Phi_{r}}{\sequT{\Gamma';\assign{x}{\multiset{\sigma_i}_{\iI}}}{\assign{r}{\sigma}}{\cmult_1}{\cexp_1}{\csize_1}}{\SysTightCBN}
  \quad
 (\derivable{\Phi^i_{u}}{\sequT{\Delta_i}{\assign{u}{\sigma_i}}{\cmult^i_2}{\cexp^i_2}{\csize^i_2}}{\SysTightCBN})_{\iI}
}{
  \derivable{\Phi}{\sequT{\ctxtsum{\Gamma'}{\Delta_i}{\iI}}{\assign{\termsubs{x}{u}{r}}{\sigma}}{\cmult_1+_{\iI}\cmult^i_2}{1+\cexp_1+_{\iI}\cexp^i_2}{\csize_1+_{\iI}\csize^i_2}}{\SysTightCBN}
}{\ruleNDESubs}
  \] with $\cmult = \cmult_1 +_{\iI}\cmult^i_2$, $\cexp = 1 + \cexp_1
  +_{\iI}\cexp^i_2$ and $\csize = \csize_1 +_{\iI}\csize^i_2$. We conclude by
  applying Lem.~\ref{l:name:substitution-tight} with $\Phi_{r}$ and $\Phi_{u}$,
  thus obtaining
  $\derivable{\Phi'}{\sequT{\ctxtsum{\Gamma'}{\Delta_i}{\iI}}{\assign{\substitute{x}{u}{t}}{\sigma}}{\cmult_1+_{\iI}\cmult^i_2}{\cexp_1+_{\iI}\cexp^i_2}{\csize_1+_{\iI}\csize^i_2}}{\SysTightCBN}$.
  Note that the counters are as expected.
    
  \item $t = \termapp{r}{u} \loredname \termapp{r'}{u} = t'$, where $r
  \loredname r'$ and $\neg\pabs{r}$. Then, $\Phi$ ends with rule
  $\ruleNTArrowE$ or $\ruleNDArrowE$, and in either case the subterm $r$ has an
  associated typing derivation $\Phi_{r}$ which verifies the hypothesis. Hence,
  by the \ih there exists a $\Phi_{r'}$ with proper counters. Thus, we can
  derive $\Phi'$ applying the same rule as in $\Phi$ to conclude.
  
  \item $t = \termabs{x}{u} \loredname \termabs{x}{u'} = t'$, where $r
  \loredname r'$. By the hypothesis we have necessarily $\sigma \in \typetight$
  since $t$ is an abstraction, then $t$ is type with rule $\ruleNTArrowI$. The
  statement then immediately follows from the \ih
\end{itemize}
\end{proof}


\begin{theorem}[Soundness]
If
$\derivable{\Phi}{\sequT{\Gamma}{\assign{t}{\sigma}}{\cmult}{\cexp}{\csize}}{\SysTightCBN}$
is tight, then there exists $p$ such that $p \in \CBNNF$ and
$t \rewriten{\callbyname}^{(\cmult,\cexp)} p$  with $\cmult$ $\mStep$-steps, $\cexp$
$\eStep$-steps, and $\nsize{p} = \csize$.
\label{t:name:correctness-tight}
\end{theorem}

\begin{proof}
We prove the statement for $\loredname$ and then conclude for the general
notion of reduction $\rewrite{\callbyname}$ by the observation that all
reduction sequences to normal form have the same number of multiplicative and
exponential steps. Let
$\derivable{\Phi}{\sequT{\Gamma}{\assign{t}{\sigma}}{\cmult}{\cexp}{\csize}}{\SysTight}$.
We proceed by induction on $\cmult + \cexp$:
\begin{itemize}
  \item If $\cmult + \cexp = 0$, then $\cmult = \cexp = 0$ and
  Lem.~\ref{l:name:czero-normal} gives $t \in \CBNNF$. Moreover, by
  Lem.~\ref{l:name:tight-size} we get both $\nsize{t} = \csize$. Thus, we
  conclude with $p = t$.

  \item If $\cmult + \cexp > 0$, then $t \notin \CBNNF$ holds by
  Lem.~\ref{l:name:czero-normal} and thus there exists $t'$ such that $t
  \lorednamen^{(1,0)} t'$ or $t \lorednamen^{(0,1)} t'$ by
  Prop.~\ref{l:cbn-cbv-normal-forms}. By
  Lem.~\ref{l:name:subject-reduction-tight} there exists a type derivation
  $\derivable{\Phi'}{\sequT{\Gamma}{\assign{t'}{\sigma}}{\cmult'}{\cexp'}{\csize}}{\SysTightCBN}$
  such that $1 + \cmult' + \cexp' = \cmult + \cexp$. By the \ih there exists $p
  \in \CBNNF$ such that $t' \lorednamen^{(\cmult',\cexp')} p$ with $\csize =
  \nsize{p}$. Then $t \lorednamen^{(1,0)} t' \lorednamen^{(\cmult',\cexp')} p$
  (resp. $t \lorednamen^{(0,1)} t' \lorednamen^{(\cmult',\cexp')} p$) which
  means $t \lorednamen^{(\cmult,\cexp)} p$, as expected.
\end{itemize}
\end{proof}


\parrafo{Completeness}
Similarly, the completeness result relays on a series of intermediate lemmas:

\begin{lemma}
If $t \in \CBNNF$, then there is a tight derivation
$\derivable{\Phi}{\sequT{\Gamma}{\assign{t}{\sigma}}{0}{0}{\nsize{t}}}{\SysTightCBN}$.
\label{l:name:normal-forms-tight}
\end{lemma}

\begin{proof}
By simultaneous induction on the following claims:
\begin{enumerate}
  \item\label{l:name:normal-forms-tight:ne} If $t \in \HCBNNF$, then there
  exists a tight derivation
  $\derivable{\Phi}{\sequT{\Gamma}{\assign{t}{\typeneutral}}{0}{0}{\nsize{t}}}{\SysTightCBN}$.
  
  \item\label{l:name:normal-forms-tight:no} If $t \in \CBNNF$, then there
  exists a tight derivation
  $\derivable{\Phi}{\sequT{\Gamma}{\assign{t}{\typetight}}{0}{0}{\nsize{t}}}{\SysTightCBN}$.
\end{enumerate}
\begin{itemize}
  \item $t = x$. Then, $t \in \HCBVNF \subseteq \CBVNF$ and we conclude both
  (\ref{l:name:normal-forms-tight:ne}) and (\ref{l:name:normal-forms-tight:no})
  by $\ruleNDAxiom$ since $\nsize{x} = 0$.
  
  \item $t = \termapp{r}{u}$. Then, $t \in \HCBNNF$ or $t \in \CBNNF$ implies
  $r \in \HCBNNF$. By the \ih (\ref{l:name:normal-forms-tight:ne}) on $r$ there
  exists a tight derivation
  $\derivable{\Phi_{r}}{\sequT{\Gamma}{\assign{r}{\typeneutral}}{0}{0}{\nsize{r}}}{\SysTightCBN}$.
  We conclude (\ref{l:name:normal-forms-tight:ne}) by $\ruleNTArrowE$,
  obtaining $\nsize{t} = \nsize{r} + 1$ as expected. Then
  (\ref{l:name:normal-forms-tight:no}) also holds. 
  
  \item $t = \termabs{x}{r}$. Then, $t \in \CBNNF$ implies $r \in \CBNNF$. By
  the \ih (\ref{l:name:normal-forms-tight:no}) on $r$ there exists a tight
  derivation
  $\derivable{\Phi_{r}}{\sequT{\Gamma'}{\assign{r}{\typetight}}{0}{0}{\nsize{r}}}{\SysTightCBN}$.
  We conclude (\ref{l:name:normal-forms-tight:no}) by using $\ruleNTArrowI$,
  thus obtaining $\nsize{t} = \nsize{r} + 1$ as expected.
  
  \item $t = \termsubs{x}{u}{r}\notin \CBNNF$. So this case does not apply. 
\end{itemize}
\end{proof}


\begin{lemma}[Anti-Substitution]
Let
$\derivable{\Phi_{\substitute{x}{u}{t}}}{\sequT{\Gamma'}{\assign{\substitute{x}{u}{t}}{\tau}}{\cmult'}{\cexp'}{\csize'}}{\SysTightCBN}$.
Then, there exist
$\derivable{\Phi_{t}}{\sequT{\Gamma;\assign{x}{\multiset{\sigma_i}_{\iI}}}{\assign{t}{\tau}}{\cmult}{\cexp}{\csize}}{\SysTightCBN}$
and
$(\derivable{\Phi^i_{u}}{\sequT{\Delta_i}{\assign{u}{\sigma_i}}{\cmult_i}{\cexp_i}{\csize_i}}{\SysTightCBN})_{\iI}$
such that $\Gamma' = \ctxtsum{\Gamma}{\Delta_i}{\iI}$, $\cmult' =
\cmult +_{\iI}{\cmult_i}$, $\cexp' = \cexp +_{\iI}{\cexp_i}$ and $\csize' =
\csize +_{\iI}{\csize_i}$.
\label{l:name:anti-substitution-tight}
\end{lemma}

\begin{proof}
By induction on $t$.
\begin{itemize}
  \item $t = x$. Then, $\substitute{x}{u}{t} = u$ and we take
  $\derivable{\Phi_{x}}{\sequT{\assign{x}{\multiset{\sigma_1}}}{\assign{x}{\sigma_1}}{0}{0}{0}}{\SysTightCBN}$
  and
  $\derivable{\Phi^{1}_{u}}{\sequT{\Gamma'}{\assign{u}{\sigma_1}}{\cmult'}{\cexp'}{\csize'}}{\SysTightCBN}$
  and the property holds.
  
  \item $t = y \neq x$. Then, $\substitute{x}{v}{t} = y$ and we take
  $\derivable{\Phi_{y}}{\sequT{\assign{y}{\multiset{\tau}}}{\assign{y}{\tau}}{0}{0}{0}}{\SysTightCBN}$
  so that $|I| = 0$ and the property holds.
  
  \item In all the other cases the property holds by the \ih
\end{itemize}
\end{proof}


\begin{lemma}[Exact Subject Expansion]
Let
$\derivable{\Phi'}{\sequT{\Gamma}{\assign{t'}{\sigma}}{\cmult'}{\cexp'}{\csize}}{\SysTightCBN}$
be a tight derivation. If $t \loredname t'$, then there is
$\derivable{\Phi}{\sequT{\Gamma}{\assign{t}{\sigma}}{\cmult}{\cexp}{\csize}}{\SysTightCBN}$
such that
\begin{enumerate}
  \item $\cmult' = \cmult - 1$ and $\cexp' = \cexp$ if $t \loredname t'$ is an
  $\mStep$-step.
  
  \item $\cexp' = \cexp - 1$ and $\cmult' = \cmult$ if $t \loredname t'$ is an
  $\eStep$-step.
\end{enumerate}
\label{l:name:subject-expansion-tight}
\end{lemma}

\begin{proof}
We actually prove the following stronger statement that allows us to reason
inductively:

Let $t \loredname  t'$ and
$\derivable{\Phi'}{\sequT{\Gamma}{\assign{t'}{\sigma}}{\cmult'}{\cexp'}{\csize}}{\SysTightCBN}$
such that $\Gamma$ is tight, and either $\sigma \in \typetight$ or
$\neg\pabs{t}$. Then, there exists
$\derivable{\Phi}{\sequT{\Gamma}{\assign{t}{\sigma}}{\cmult}{\cexp}{\csize}}{\SysTightCBN}$
such that
\begin{enumerate}
  \item $\cmult' = \cmult - 1$ and $\cexp' = \cexp$ if $t \loredname t'$ is an
  $\mStep$-step.
  
  \item $\cexp' = \cexp - 1$ and $\cmult' = \cmult$ if $t \loredname t'$ is an
  $\eStep$-step.
\end{enumerate}

We proceed by induction on $t \loredname t'$.
\begin{itemize}
  \item $t = \termapp{\ctxtapp{\ctxt{L}}{\termabs{x}{r}}}{u}
  \loredname \ctxtapp{\ctxt{L}}{\termsubs{x}{u}{r}} = t'$ is an $\mStep$-step.
  We proceed by induction on $\ctxt{L}$. We only show here the case $\ctxt{L} =
  \Box$ as the inductive case is straightforward.

  Then, $t' = \termsubs{x}{u}{r}$ and $\Phi'$ has the following form: \[
\Rule{
  \derivable{\Phi_{r}}{\sequT{\Gamma';\assign{x}{\multiset{\sigma_i}_{\iI}}}{\assign{r}{\sigma}}{\cmult_1}{\cexp_1}{\csize_1}}{\SysTightCBN}
  \quad
  (\derivable{\Phi^i_{u}}{\sequT{\Delta_i}{\assign{u}{\sigma_i}}{\cmult^i_2}{\cexp^i_2}{\csize^i_2}}{\SysTightCBN})_{\iI}
}{
  \sequT{\ctxtsum{\Gamma'}{\Delta_i}{\iI}}{\assign{\termsubs{x}{u}{r}}{\sigma}}{\cmult_1+_{\iI}{\cmult^i_2}}{1+\cexp_1+_{\iI}{\cexp^i_2}}{\csize_1+_{\iI}{\csize^i_2}}
}{\ruleNDESubs}
  \] with $\cmult' = \cmult_1 +_{\iI}{\cmult^i_2}$, $\cexp' = 1 + \cexp_1
  +_{\iI}{\cexp^i_2}$ and $\csize = \csize_1 +_{\iI}{\csize^i_2}$. We
  conclude with the following type derivation \[
\Rule{
  \Rule{
    \derivable{\Phi_{r}}{\sequT{\Gamma';\assign{x}{\multiset{\sigma_i}_{\iI}}}{\assign{r}{\sigma}}{\cmult_1}{\cexp_1}{\csize_1}}{\SysTightCBN}
  }{
    \sequT{\Gamma'}{\assign{\termabs{x}{r}}{\functtype{\multiset{\sigma_i}_{\iI}}{\sigma}}}{\cmult_1}{\cexp_1}{\csize_1}
  }{\ruleNDArrowI}
    \quad
    (\derivable{\Phi^i_{u}}{\sequT{\Delta_i}{\assign{u}{\sigma_i}}{\cmult^i_2}{\cexp^i_2}{\csize^i_2}}{\SysTightCBN})_{\iI}
  }{
    \derivable{\Phi}{\sequT{\ctxtsum{\ctxtres{\Gamma'}{x}{}}{\Delta}{}}{\assign{\termapp{(\termabs{x}{r})}{u}}{\sigma}}{1+\cmult_1+_{\iI}{\cmult^i_2}}{1+\cexp_1+_{\iI}{\cexp^i_2}}{\csize_1+_{\iI}{\csize^i_2}}}{\SysTightCBN}
}{\ruleNDArrowE}
  \] taking $\cmult = 1 + \cmult_1 +_{\iI}{\cmult^i_2} = \cmult' + 1$ and
  $\cexp = 1 + \cexp_1 +_{\iI}{\cexp^i_2} = \cexp'$.
  
  \item $t = \termsubs{x}{u}{r} \loredname \substitute{x}{u}{r} = t'$ is an
  $\eStep$-step. By Lem.~\ref{l:name:anti-substitution-tight} there exist type
  derivations
  $\derivable{\Phi_{r}}{\sequT{\Gamma';\assign{x}{\multiset{\sigma_i}_{\iI}}}{\assign{r}{\sigma}}{\cmult_1}{\cexp_1}{\csize_1}}{\SysTightCBN}$
  and
  $(\derivable{\Phi^i_{u}}{\sequT{\Delta_i}{\assign{u}{\sigma_i}}{\cmult^i_2}{\cexp^i_2}{\csize^i_2}}{\SysTightCBN})_{\iI}$
  such that $\Gamma = \ctxtsum{\Gamma'}{\Delta_i}{\iI}$, $\cmult' = \cmult_1
  +_{\iI}{\cmult^i_2}$, $\cexp' = \cexp_1 +_{\iI}{\cexp^i_2}$ and $\csize =
  \csize_1 +_{\iI}{\csize^i_2}$. We conclude by $\ruleNDESubs$ with the
  following type derivation \[
\Rule{
  \derivable{\Phi_{r}}{\sequT{\Gamma';\assign{x}{\multiset{\sigma_i}_{\iI}}}{\assign{r}{\sigma}}{\cmult_1}{\cexp_1}{\csize_1}}{\SysTightCBN}
  \quad
  (\derivable{\Phi^i_{u}}{\sequT{\Delta_i}{\assign{u}{\sigma_i}}{\cmult^i_2}{\cexp^i_2}{\csize^i_2}}{\SysTightCBN})_{\iI}
}{
  \derivable{\Phi}{\sequT{\ctxtsum{\Gamma'}{\Delta_i}{\iI}}{\assign{\termsubs{x}{u}{r}}{\sigma}}{\cmult_1+_{\iI}{\cmult^i_2}}{1+\cexp_1+_{\iI}{\cexp^i_2}}{\csize_1+_{\iI}{\csize^i_2}}}{\SysTightCBN}
}{\ruleNDESubs}
  \] taking $\cmult = \cmult_1 +_{\iI}{\cmult^i_2} = \cmult'$ and $\cexp = 1 +
  \cexp_1 +_{\iI}{\cexp^i_2} = \cexp' + 1$.
  
  \item $t = \termapp{r}{u} \loredname \termapp{r'}{u} = t'$, where $r
  \loredname r'$ and $\neg\pabs{r}$. Then, $\Phi'$ ends with rule
  $\ruleNTArrowE$ or $\ruleNDArrowE$, and in either case the subterm $r'$ has
  an associated typing derivation $\Phi_{r'}$ which verifies the hypothesis
  since $\neg\pabs{r}$. Hence, by \ih there exists a $\Phi_{r}$ with proper
  counters such that we can derive $\Phi$ applying the same rule as in $\Phi'$
  to conclude.
  
  \item $t = \termabs{x}{u} \loredname \termabs{x}{u'} = t'$, where $r
  \loredname r'$. By the hypothesis we have necessarily $\sigma \in \typetight$
  since $t'$ is an abstraction, then $t'$ is type with rule $\ruleNTArrowI$.
  The statement then immediately follows from the \ih
\end{itemize}
\end{proof}


\begin{theorem}[Completeness]
If $t \rewriten{\callbyname}^{(\cmult,\cexp)} p$  with $p \in \CBNNF$, then
there exists a tight type derivation
$\derivable{\Phi}{\sequT{\Gamma}{\assign{t}{\sigma}}{\cmult}{\cexp}{\valsize{p}}}{\SysTightCBN}$.
\label{t:name:completeness-tight}
\end{theorem}

\begin{proof}
As for soundness (Thm.~\ref{t:name:correctness-tight}) we prove the
statement for $\loredname$ and then conclude for the general notion of
reduction $\rewrite{\callbyname}$. Let $t \lorednamen^{(\cmult,\cexp)} p$. We
proceed by induction on $\cmult+\cexp$:
\begin{itemize}
  \item If $\cmult+\cexp = 0$, then $\cmult = \cexp = 0$ and thus $t = p$,
  which implies $t \in \CBNNF$. Lem.~\ref{l:name:normal-forms-tight} allows to
  conclude.
  
  \item If $\cmult+\cexp > 0$, then there is $t'$ such that $t
  \lorednamen^{(1,0)} t' \lorednamen^{(\cmult-1,\cexp)} p$ or $t
  \lorednamen^{(0,1)} t' \lorednamen^{(\cmult,\cexp-1)} p$. By the \ih there is
  a tight derivation
  $\derivable{\Phi'}{\sequT{\Gamma}{\assign{t'}{\sigma}}{\cmult'}{\cexp'}{\valsize{p}}}{\SysTightCBN}$
  such $\cmult' + \cexp' = \cmult + \cexp -1$.
  Lem.~\ref{l:name:subject-expansion-tight} gives a tight derivation
  $\derivable{\Phi}{\sequT{\Gamma}{\assign{t}{\sigma}}{\cmult''}{\cexp''}{\valsize{p}}}{\SysTightCBN}$
  such $\cmult'' + \cexp'' = \cmult' + \cexp' + 1$. We then have $\cmult'' +
  \cexp'' = \cmult + \cexp$. The fact that $\cmult'' = \cmult$ and $\cexp'' =
  \cexp$ holds by a simple case analysis.
\end{itemize}
\end{proof}


Notice that the previous theorem is stated by using the general notion of
reduction $\rewrite{\callbyname}$, but the proofs are done using the
deterministic reduction $\loredname$. The equivalence holds because any two
different reduction paths to normal form have the same number of multiplicative
and exponential steps, as already remarked.

\begin{example}
\label{ex:t0cbn}
Consider $t_0 =
\termapp{\termapp{\Kterm}{(\termapp{z}{\id})}}{(\termapp{\id}{\id})}$ that
$\rewrite{\callbyname}$-reduces in 2 $\mStep$-steps and 2 $\eStep$-steps to the term
$\termapp{z}{\id} \in \CBNNF$, whose $\callbyname$-size is 1: \[
t_0 = \termapp{\underline{\termapp{\Kterm}{(\termapp{z}{\id})}}}{(\termapp{\id}{\id})}
\rewrite{\dBeta}
\underline{\termapp{\termsubs{x}{\termapp{z}{\id}}{(\termabs{y}{x})}}{(\termapp{\id}{\id})}}
\rewrite{\dBeta}
\underline{\termsubs{x}{\termapp{z}{\id}}{\termsubs{y}{\termapp{\id}{\id}}{x}}}
\rewrite{\sTerm}
\underline{\termsubs{y}{\termapp{\id}{\id}}{(\termapp{z}{\id})}}
\rewrite{\sTerm}
\termapp{z}{\id}
\] System $\SysTightCBN$ admits a tight type derivation for $t_0$ with the
expected final counter $(2,2,1)$:
{\small
\begin{center}
$
\Rule{
  \Rule{
    \Rule{
      \Rule{
        \Rule{}{
          \sequT{\assign{x}{\multiset{\typeneutral}}}{\assign{x}{\typeneutral}}{0}{0}{0}
        }{\ruleNDAxiom}
      }{
        \sequT{\assign{x}{\multiset{\typeneutral}}}{\assign{\termabs{y}{x}}{\functtype{\emul}{\typeneutral}}}{0}{0}{0}
      }{\ruleNDArrowI}
    }{
      \sequT{}{\assign{\Kterm}{\functtype{\multiset{\typeneutral}}{\functtype{\emul}{\typeneutral}}}}{0}{0}{0}
    }{\ruleNDArrowI}
    \Rule{
      \Rule{}{
        \sequT{\assign{z}{\multiset{\typeneutral}}}{\assign{z}{\typeneutral}}{0}{0}{0}
      }{\ruleNDAxiom}
    }{
      \sequT{\assign{z}{\multiset{\typeneutral}}}{\assign{\termapp{z}{\id}}{\typeneutral}}{0}{0}{1}
    }{\ruleNTArrowE}
  }{
    \sequT{\assign{z}{\multiset{\typeneutral}}}{\assign{\termapp{\Kterm}{(\termapp{z}{\id})}}{\functtype{\emul}{\typeneutral}}}{1}{1}{1}
  }{\ruleNDArrowE}
}{
  \sequT{\assign{z}{\multiset{\typeneutral}}}{\assign{\termapp{\termapp{\Kterm}{(\termapp{z}{\id})}}{(\termapp{\id}{\id})}}{\typeneutral}}{2}{2}{1}
}{\ruleNDArrowE}
$\end{center}
}
\end{example}


\section{Tight Call-by-Value}
\label{s:tight-value}
  
In this section we introduce a tight system $\SysTightCBV$ for CBV which
captures exact measures for $\rewrite{\callbyvalue}$-reduction sequences. Here
we do not only distinguish length of reduction sequences from size of normal
forms, but also discriminate multiplicative from exponential steps. Moreover,
we establish a precise relation between the counters in CBV and their
counterparts in the $\BangRev$-calculus (\cf Sec.~\ref{s:tight-translations}).
This constitutes one of the main contributions of the present work.

The grammar of types of system $\SysTightCBV$ is given by:
\begin{center}
\begin{tabular}{rrcll}
\emphdef{(Tight Types)} & $\typetight$   & $\Coloneq$ & $\typeneutral \mid \typevalue \mid \typevar$ \\
\emphdef{(Types)}       & $\sigma, \tau$ & $\Coloneq$ & $\typetight \mid \M \mid \functtype{\M}{\sigma}$ \\
\emphdef{(Multitypes)}  & $\M, \N$       & $\Coloneq$ & $\intertype{\sigma_i}{i \in I}$  where $I$ is a finite set
\end{tabular}
\end{center}
Indeed, apart from the constant $\typeneutral$ also used in CBN, we now
introduce constants $\typevalue$ and $\typevar$, typing respectively, terms
reducing to values; and terms reducing to persistent variables that are not
acting as values (\ie they are applied to some argument to produce neutral
terms). Notice that we remove the base type $\typeabs$ from CBN, since
abstractions are in particular values, so they will be typed with
$\typevalue$. The notions of \emphdef{tightness} for multitypes, contexts and
derivations are exactly the same we used for CBN.

A type system for CBV must type variables without knowing yet the future role
they are going to play (head of neutral term or value). This is one of the
main difficulties behind the definition of such a system. Besides that, the
system must also be able to count multiplicative and exponential steps
independently. The resulting type system $\SysTightCBV$ for CBV is given by
the typing rules in Fig.~\ref{fig:typingSchemesCBVPersistent}
and~\ref{fig:typingSchemesCBVConsuming}, distinguishing between
persistent and consuming rules resp. As a matter of notation, given a
tight type $\typetight_0$ we write $\overline{\typetight_0}$ to denote
a tight type different from $\typetight_0$. Thus for instance,
$\overline{\typevar} \in \{\typevalue, \typeneutral\}$.

\begin{figure}[ht]
\centering $
\kern-1em
\begin{array}{c}
\Rule{}
     {\sequT{\assign{x}{\multiset{\typevar}}}{\assign{x}{\typevar}}{0}{0}{0}}
     {\ruleVTAxiomVar}
\quad
\Rule{}
     {\sequT{}{\assign{x}{\typevalue}}{0}{0}{0}}
     {\ruleVTAxiom}
\quad
\Rule{}
     {\sequT{}{\assign{\termabs{x}{t}}{\typevalue}}{0}{0}{0}}
     {\ruleVTArrowI}
\\
\\
\Rule{\sequT{\Gamma}{\assign{t}{\overline{\typevalue}}}{\cmult}{\cexp}{\csize}
      \enspace
      \sequT{\Delta}{\assign{u}{\overline{\typevar}}}{\cmult'}{\cexp'}{\csize'}
     }
     {\sequT{\ctxtsum{\Gamma}{\Delta}{}}{\assign{\termapp{t}{u}}{\typeneutral}}{\cmult+\cmult'}{\cexp+\cexp'}{\csize+\csize'+1}}
     {\ruleVTArrowE}
\\
\\
\Rule{\sequT{\Gamma}{\assign{t}{\tau}}{\cmult}{\cexp}{\csize}
      \enspace
      \sequT{\Delta}{\assign{u}{\typeneutral}}{\cmult'}{\cexp'}{\csize'}
      \enspace
      \ptight{\Gamma(x)}
     }
     {\sequT{\ctxtsum{(\ctxtres{\Gamma}{x}{})}{\Delta}{}}{\assign{\termsubs{x}{u}{t}}{\tau}}{\cmult+\cmult'}{\cexp+\cexp'}{\csize+\csize'}}
     {\ruleVTESubs}
\end{array}$
\caption{System $\SysTightCBV$ for the Call-by-Value Calculus: Persistent Typing Rules.}
\label{fig:typingSchemesCBVPersistent}
\end{figure}

\begin{figure}[ht]
\centering $
\begin{array}{c}
\Rule{\vphantom{\Gamma}}
     {\sequT{\assign{x}{\M}}{\assign{x}{\M}}{0}{1}{0}}
     {\ruleVDAxiom}
\qquad
\Rule{\sequT{\Gamma}{\assign{t}{\multiset{\functtype{\M}{\tau}}}}{\cmult}{\cexp}{\csize}
      \quad
      \sequT{\Delta}{\assign{u}{\M}}{\cmult'}{\cexp'}{\csize'}
     }
     {\sequT{\ctxtsum{\Gamma}{\Delta}{}}{\assign{\termapp{t}{u}}{\tau}}{\cmult+\cmult'+1}{\cexp+\cexp'-1}{\csize+\csize'}}
     {\ruleVDArrowE}
\\
\\
\Rule{\sequT{\Gamma}{\assign{t}{\multiset{\functtype{\M}{\tau}}}}{\cmult}{\cexp}{\csize}
      \quad
      \sequT{\Delta}{\assign{u}{\typeneutral}}{\cmult'}{\cexp'}{\csize'}
      \quad
      \ptight{\M}
     }
     {\sequT{\ctxtsum{\Gamma}{\Delta}{}}{\assign{\termapp{t}{u}}{\tau}}{\cmult+\cmult'+1}{\cexp+\cexp'-1}{\csize+\csize'}}
     {\ruleVDApp}
\\
\\
\Rule{\many{\sequT{\Gamma_i}{\assign{t}{\tau_i}}{\cmult_i}{\cexp_i}{\csize_i}}{\iI}}
     {\sequT{\ctxtsum{}{\ctxtres{\Gamma_i}{x}{}}{i \in I}}{\assign{\termabs{x}{t}}{\intertype{\functtype{\Gamma_i(x)}{\tau_i}}{i \in I}}}{+_{\iI}{\cmult_i}}{1+_{\iI}{\cexp_i}}{+_{\iI}{\csize_i}}}
     {\ruleVDArrowI}
\\
\\
\Rule{\sequT{\Gamma}{\assign{t}{\sigma}}{\cmult}{\cexp}{\csize}
      \quad
      \sequT{\Delta}{\assign{u}{\Gamma(x)}}{\cmult'}{\cexp'}{\csize'}
     }
     {\sequT{\ctxtsum{(\ctxtres{\Gamma}{x}{})}{\Delta}{}}{\assign{\termsubs{x}{u}{t}}{\sigma}}{\cmult+\cmult'}{\cexp+\cexp'}{\csize+\csize'}}
     {\ruleVDESubs}
\end{array}
$
\caption{System $\SysTightCBV$ for the Call-by-Value Calculus: Consuming Typing Rules.}
\label{fig:typingSchemesCBVConsuming}
\end{figure}
Some rules deserve a comment. A difficult property to be statically captured by
the counters is that an exponential step can only be generated by the meeting
of a substitution constructor with an appropriate value argument. This remark
leads to the introduction of different rules for typing variables, depending on
the role they play (to be a value or not). Indeed, there are three axioms for
variables: $\ruleVTAxiomVar$ typing variables that will persist as the head of
a neutral term; $\ruleVTAxiom$ typing variable values that may be substituted
by other values, \ie they are \emph{placeholders} for future persistent values;
and $\ruleVDAxiom$ typing variable values that, in particular,  may be consumed
as arguments. Consuming values are always typed with multitypes. Rule
$\ruleVTArrowE$ types neutral applications, \ie the left premise has type
$\typevar$ or $\typeneutral$.
Rule $\ruleVDArrowI$ increases the second counter, just like $\ruleVDAxiom$,
typing a value that is consumed as an argument. Rules $\ruleVDArrowE$ and
$\ruleVDApp$ increment the first counter because the (consuming) application
will be used to perform a $\dBeta$-step, while they decrement the second
counter to compensate for the left-hand-side value that is not being consumed
after all. In other words, consuming values are systematically typed by
incrementing their exponential counter by one (rules $\ruleVDAxiom$ and
$\ruleVDArrowI$), but they can finally act as computations instead as values,
notably when they are placed in a head position, so that their exponential
counter needs to be adjusted correctly (\cf Ex.~\ref{ex:t0cbv}). The decrement
of the exponential counters in rules $\ruleVDArrowE$ and $\ruleVDApp$ can also
be understood by means of the subtle translation from CBV to the
$\BangRev$-calculus that we introduce in Sec.~\ref{s:cbn-cbv-embeddings}. Rule
$\ruleVDApp$ is particularly useful to type $\dBeta$-redexes whose reduction
does not create an exponential redex, because the argument of the substitution
created by the $\dBeta$-step does not reduce to a value.

In spite of the decrements in rules $\ruleVDArrowE$ and $\ruleVDApp$, the
counters are positive.

\begin{lemma}
If
$\derivable{\Phi}{\sequT{\Gamma}{\assign{t}{\sigma}}{\cmult}{\cexp}{\csize}}{\SysTightCBV}$
then $\cexp \geq 0$. Moreover, if $\sigma$ is a multitype then $\cexp > 0$.
\label{l:value:cexp-positive}
\end{lemma}

\begin{proof}
By induction on $\Phi$ analysing the last rule applied.
\begin{itemize}
  \item $\ruleVTAxiomVar$, $\ruleVTAxiom$, $\ruleVTArrowI$, $\ruleVDAxiom$. The
  result is immediate.
  
  \item $\ruleVTArrowE$. Then $t = \termapp{r}{u}$, $\sigma = \typeneutral$,
  $\cexp = \cexp_1 + \cexp_2$ and we have type derivations $\Phi_{r}$ and
  $\Phi_{u}$ with counters $\cexp_1$ and $\cexp_2$ resp. both assigning
  non-multiset types. By \ih we get $\cexp_1 \geq 0$ and $\cexp_2 \geq 0$, and
  hence we conclude $\cexp \geq 0$.
  
  \item $\ruleVTESubs$. Then $t = \termsubs{x}{u}{r}$, $\cexp = \cexp_1 +
  \cexp_2$ and we have type derivations $\Phi_{r}$ and $\Phi_{u}$ with
  counters $\cexp_1$ and $\cexp_2$ resp. We conclude directly from the \ih
  
  \item $\ruleVDArrowE$. Then $t = \termapp{r}{u}$, $\cexp = \cexp_1 + \cexp_2
  - 1$ and we have type derivations $\Phi_{r}$ and $\Phi_{u}$ with counters
  $\cexp_1$ and $\cexp_2$ resp. both assigning multiset types. By \ih we get
  $\cexp_1 > 0$ and $\cexp_2 > 0$, and hence we conclude $\cexp > 0$ as
  expected.
  
  \item $\ruleVDApp$. Then $t = \termapp{r}{u}$, $\sigma \in \typetight$,
  $\cexp = \cexp_1 + \cexp_2 - 1$ and we have a derivation $\Phi_{r}$ with
  counter $\cexp_1$ and a multiset type, and another derivation $\Phi_{u}$
  with counter $\cexp_2$ and type $\typeneutral$. By \ih we get $\cexp_1 > 0$
  and $\cexp_2 \geq 0$, so we can safely conclude $\cexp \geq 0$ as expected.
  
  \item $\ruleVDArrowI$. Then $t = \termabs{x}{r}$, $\sigma =
  \intertype{\functtype{\M_i}{\tau_i}}{i \in I}$, $\cexp = 1
  +_{i \in I}{\cexp_i}$ for some set of indices $I$, and we have a type
  derivation $\Phi^{i}_{r}$ with counter $\cexp_i$ for each $i \in I$. By \ih
  we get $\cexp_i \geq 0$ for every $i \in I$, and hence we conclude $\cexp >
  0$ as expected.
  
  \item $\ruleVDESubs$. Then $t = \termsubs{x}{u}{r}$, $\cexp = \cexp_1 +
  \cexp_2$ and we have type derivations $\Phi_{r}$ and $\Phi_{u}$ with
  counters $\cexp_1$ and $\cexp_2$ resp. We conclude directly from the \ih
\end{itemize}
\end{proof}


\parrafo{Soundness}
The soundness property is based on a series of auxiliary results that enables
to reason about tight type derivations:

\begin{lemma}[Tight Spreading]
Let
$\derivable{\Phi}{\sequT{\Gamma}{\assign{t}{\sigma}}{\cmult}{\cexp}{\csize}}{\SysTightCBV}$
such that $\Gamma$ is tight.
\begin{enumerate}
  \item\label{l:value:tight-spreading:tight} If $t \in\HCBVNF$ or $\cmult
  = \cexp = 0$, then $\sigma \in \typetight$.
  
  \item\label{l:value:tight-spreading:var} If $\sigma =
  \multiset{\functtype{\M}{\tau}}$, then $t \notin \VarHCBVNF$. 
\end{enumerate}
\label{l:value:tight-spreading}
\end{lemma}

\begin{proof}
We show simultaneously the two properties by induction on $\Phi$.
\begin{enumerate}
  \item We consider every possible rule.
  \begin{itemize}
    \item $\ruleVTAxiomVar$, $\ruleVTAxiom$, $\ruleVTArrowI$. Then $\sigma \in
    \typetight$ by definition.

    \item $\ruleVTESubs$. Then $t = \termsubs{x}{u}{r}$. The hypothesis implies
    that the type derivation $\Phi_{r}$ typing $r$ with type $\sigma$ has a
    tight context, and the two first counters of $\Phi_{r}$ are equal to $0$ or
    $r \in \HCBVNF$. Then the \ih (\ref{l:value:tight-spreading:tight})
    applied to $\Phi_{r}$ gives $\sigma \in \typetight$.

    \item $\ruleVDAxiom$. Then $t = x \notin \HCBVNF$ and $\cexp > 0$. Thus,
    this case does not apply.
    
    \item $\ruleVTArrowE$, $\ruleVDArrowE$. Then $t = \termapp{r}{u}$. The
    proof concludes with $\cmult > 0$ so let us assume $t \in \HCBVNF$, which
    means in particular that $r \in \VarHCBVNF$ or $r \in \HCBVNF$. Then, we
    have a derivation
    $\derivable{\Phi_{r}}{\sequT{\Gamma'}{\assign{r}{\multiset{\functtype{\M}{\sigma}}}}{\cmult'}{\cexp'}{\csize'}}{\SysTightCBV}$,
    with $\Gamma' \ctxleq \Gamma$, so that $\Gamma'$ is also tight. The \ih
    (\ref{l:value:tight-spreading:var}) gives $r \notin  \VarHCBVNF$, so that
    $r \in \HCBVNF$. Then we can apply the \ih
    (\ref{l:value:tight-spreading:tight}) to $\Phi_{r}$ which gives
    $\multiset{\functtype{\M}{\sigma}} \in \typetight$, and this is a
    contradiction. Hence, this case does not apply.
    
    \item $\ruleVDApp$. This case is similar to the previous one.
    
    \item $\ruleVDArrowI$. Then $t = \termabs{x}{u} \notin \HCBVNF$.
    Moreover, by Lem.~\ref{l:value:cexp-positive}, we have $\cexp > 0$. Thus,
    this case does not apply.
    
    \item $\ruleVDESubs$. Then $t = \termsubs{x}{u}{r}$. The hypothesis implies
    that the type derivation $\Phi_{u}$ typing $u$ has a tight context, and the
    two first counters of $\Phi_{u}$ are equal to $0$ or $u \in \HCBVNF$.
    Moreover, we know by construction that the type of $u$ is a multiset. But
    the \ih (\ref{l:value:tight-spreading:tight}) applied to $\Phi_{u}$ gives a
    tight type for $u$ which leads to a contradiction.
  \end{itemize}
  
  \item Let $\sigma = \multiset{\functtype{\M}{\tau}}$. Suppose $t \in
  \VarHCBVNF$. We reason by cases. 
  \begin{itemize}
    \item $t = x$. Then $\Phi$ is necessarily $\ruleVDAxiom$ because of the
    type $\sigma$, thus $\Gamma = \assign{x}{\multiset{\functtype{\M}{\tau}}}$
    contradicts the fact that $\Gamma$ is tight.
    
    \item $t = \termsubs{x}{u}{r}$, with $r \in \VarHCBVNF$ and $u \in
    \HCBVNF$. Then $\Phi$ necessarily ends with rule $\ruleVDESubs$ or
    $\ruleVTESubs$. In the first case, $\Phi$ contains a subderivation typing
    $u$ with a context $\Delta \ctxleq \Gamma$ and a multiset type $\N$. Thus
    $\Delta$ is also tight so that the \ih
    (\ref{l:value:tight-spreading:tight}) gives $\N \in \typetight$ which is a
    contradiction because $\N$ is a multiset. In the second case $\Phi$
    contains a subderivation typing $r$ with a tight context and the same type
    $\multiset{\functtype{\M}{\tau}}$, so that the \ih
    (\ref{l:value:tight-spreading:var}) gives $r \notin \VarHCBVNF$, which
    implies $t \notin \VarHCBVNF$, contradiction.
  \end{itemize}
\end{enumerate}
\end{proof}


\begin{lemma}
Let
$\derivable{\Phi}{\sequT{\Gamma}{\assign{t}{\sigma}}{\cmult}{\cexp}{\csize}}{\SysTightCBV}$
tight. Then, $\cmult = \cexp = 0$ iff $t \in  \CBVNF$.
\label{l:value:czero-normal}
\end{lemma}

\begin{proof}
$\left.\Rightarrow\right)$ We prove simultaneously the following statements by
induction on $\Phi$.
\begin{enumerate}
  \item\label{l:value:czero-normal:vr}
  $\derivable{\Phi}{\sequT{\Gamma}{\assign{t}{\sigma}}{0}{0}{\csize}}{\SysTightCBV}$
  tight and $\pvar{t}$ implies $t \in \VarHCBVNF$.
  
  \item\label{l:value:czero-normal:ne}
  $\derivable{\Phi}{\sequT{\Gamma}{\assign{t}{\sigma}}{0}{0}{\csize}}{\SysTightCBV}$
  tight and $\papp{t}$ implies $t \in \HCBVNF$.
  
  \item\label{l:value:czero-normal:no}
  $\derivable{\Phi}{\sequT{\Gamma}{\assign{t}{\sigma}}{0}{0}{\csize}}{\SysTightCBV}$
  tight and $\pabs{t}$ implies $t \in \CBVNF$.
\end{enumerate}

\begin{itemize}
  \item $\ruleVTAxiom$, $\ruleVTAxiomVar$. Then $t = x$ and $\pvar{t}$. By
  definition $x \in \VarHCBVNF$ trivially holds.
  
  \item $\ruleVTArrowE$. Then $t = \termapp{r}{u}$, $\sigma = \typeneutral$,
  $\Gamma = \ctxtsum{\Gamma'}{\Delta}{}$, $\csize = \csize' + \csize'' + 1$,
  $\sequT{\Gamma'}{\assign{r}{\overline{\typevalue}}}{0}{0}{\csize'}$, and
  $\sequT{\Delta}{\assign{u}{\overline{\typevar}}}{0}{0}{\csize''}$.
  We are in the case $\papp{t}$ so we need to show $t \in \HCBVNF$. Moreover,
  $\neg\pabs{r}$ holds since $r$ is typed with $\overline{\typevalue}$
  (\ie $\typevar$ or $\typeneutral$), and $\pabs{r}$ implies $\sigma$ is
  $\typevalue$ or a multiset type. By hypothesis $\Gamma'$ is tight, hence \ih
  (\ref{l:value:czero-normal:vr}) or (\ref{l:value:czero-normal:ne}) gives $r
  \in \VarHCBVNF$ or $r \in \HCBVNF$ respectively. Similarly, $\Delta$ is tight
  and one of the \ih gives $u \in \CBVNF$, since $\VarHCBVNF \subseteq
  \CBVNF$ and $\HCBVNF \subseteq \CBVNF$. We thus conclude $\termapp{r}{u} \in
  \HCBVNF$.
  
  \item $\ruleVTArrowI$. Then $t = \termabs{x}{r} \in \CBVNF$ by definition.
  
  \item $\ruleVTESubs$. Then $t = \termsubs{x}{u}{r}$, $\Gamma =
  \ctxtsum{(\ctxtres{\Gamma'}{x}{})}{\Delta}{}$, $\csize = \csize' + \csize''$,
  $\sequT{\Gamma'}{\assign{r}{\sigma}}{0}{0}{\csize'}$,
  $\sequT{\Delta}{\assign{u}{\typeneutral}}{0}{0}{\csize''}$ and
  $\ptight{\Gamma'(x)}$. Moreover, note that $\pabs{u}$ implies its type is
  $\typevalue$ or a multiset, while $\pvar{u}$ implies its type is $\typevar$,
  $\typevalue$ or a multiset. Then, $\neg\pabs{u}$ and $\neg\pvar{u}$ hold
  since it is typed with $\typeneutral$. Since $\ctxtres{\Gamma'}{x}{}$ is
  tight and $\ptight{\Gamma'(x)}$, then $\Gamma'$ and $\Delta$ are both tight
  as well. The \ih (\ref{l:value:czero-normal:ne}) gives $u \in \HCBVNF$. On
  the other hand, some of the \ih applies and gives $r \in \VarHCBVNF$ or $r
  \in \HCBVNF$ or $r \in\CBVNF$. We thus conclude $t \in \VarHCBVNF$ or $t \in
  \HCBVNF$ or $t \in\CBVNF$, respectively.
  
  \item $\ruleVDAxiom$. This case does not apply since it concludes with $\cexp
  > 0$.
  
  \item $\ruleVDArrowE$, $\ruleVDApp$. These cases do not apply since they
  conclude with $\cmult > 0$.
  
  \item $\ruleVDArrowI$. This case does not apply since, by
  Lem.~\ref{l:value:cexp-positive}, it concludes with $\cexp > 0$.
    
  \item $\ruleVDESubs$. Then $t = \termsubs{x}{u}{r}$, $\csize = \csize' +
  \csize''$ and $\Gamma = \ctxtsum{(\ctxtres{\Gamma'}{x}{})}{\Delta}{}$ tight
  such that, in particular,
  $\derivable{\Phi_{u}}{\sequT{\Delta}{\assign{u}{\Gamma'(x)}}{0}{0}{\csize''}}{\SysTightCBV}$
  with $\Delta$ tight. By Lem.~\ref{l:value:tight-spreading}
  (\ref{l:value:tight-spreading:tight}) on $\Phi_{u}$, $\Gamma'(x) \in
  \typetight$ which leads to a contradiction since $\Gamma'(x)$ is a multiset
  by definition. Hence, this case does not apply.
\end{itemize}

$\left.\Leftarrow\right)$ By induction on $t$.
\begin{itemize}
  \item $t = x$. Then, $\Phi$ necessarily ends with rule $\ruleVTAxiom$, or
  $\ruleVTAxiomVar$ and the statement trivially holds.
  
  \item $t = \termapp{r}{u}$. By definition $t \in \CBVNF$ gives $r \in
  \VarHCBVNF$ or $r \in \HCBVNF$, and $u \in \CBVNF$. There are three cases to
  consider:
  \begin{enumerate}
    \item if $\Phi$ ends with rule $\ruleVTArrowE$, then $\sigma =
    \typeneutral$, $\Gamma = \ctxtsum{\Gamma'}{\Delta}{}$, $\cmult = \cmult' +
    \cmult''$, $\cexp = \cexp' + \cexp''$, $\csize = \csize' + \csize'' + 1$,
    $\sequT{\Gamma'}{\assign{r}{\overline{\typevalue}}}{\cmult'}{\cexp'}{\csize'}$,
    $\sequT{\Delta}{\assign{u}{\overline{\typevar}}}{\cmult''}{\cexp''}{\csize''}$.
    Moreover, $\Gamma'$ and $\Delta$ are both tight. Then, the \ih gives
    $\cmult' = \cexp' = 0$ and $\cmult'' = \cexp'' = 0$. Hence, $\cmult = \cexp
    = 0$.
    
    \item if $\Phi$ ends with rule $\ruleVDArrowE$, then $\Gamma =
    \ctxtsum{\Gamma'}{\Delta}{}$, $\cmult = \cmult' + \cmult'' + 1$, $\cexp =
    \cexp' + \cexp'' - 1$, $\csize = \csize' + \csize''$ and
    $\derivable{\Phi_{r}}{\sequT{\Gamma'}{\assign{r}{\multiset{\functtype{\M}{\tau}}}}{\cmult'}{\cexp'}{\csize'}}{\SysTightCBV}$
    where $\Gamma'$ is tight. We know by Lem.~\ref{l:value:tight-spreading}
    (\ref{l:value:tight-spreading:var}) that $r \notin \VarHCBVNF$. Then, it is
    necessarily the case $r \in \HCBVNF$ and Lem.~\ref{l:value:tight-spreading}
    (\ref{l:value:tight-spreading:tight}) gives
    $\multiset{\functtype{\M}{\tau}} \in \typetight$ which is a contradiction.
    Hence, this case does not apply.
    
    \item if $\Phi$ ends with rule $\ruleVDApp$, this case is similar to the
    previous one.
  \end{enumerate}
  
  \item $t = \termabs{x}{r}$. There are two cases to consider for $\Phi$:
  \begin{enumerate}
    \item if $\Phi$ ends with rule $\ruleVTArrowI$, then $\cmult = \cexp = 0$
    by definition.
    
    \item if $\Phi$ ends with rule $\ruleVDArrowI$, then $\sigma$ should be a
    multiset of functional types, which contradicts the hypothesis of $\Phi$
    being tight. Hence, this case does not apply.
  \end{enumerate}
  
  \item $t = \termsubs{x}{u}{r}$. By definition $t \in \CBVNF$ implies $r \in
  \CBVNF$ and $u \in \HCBVNF \subseteq \CBVNF$. Then, there are two cases to
  consider for $\Phi$:
  \begin{enumerate}
    \item if $\Phi$ ends with rule $\ruleVTESubs$, then $\Gamma =
    \ctxtsum{(\ctxtres{\Gamma'}{x}{})}{\Delta}{}$, $\cmult = \cmult' +
    \cmult''$, $\cexp = \cexp' + \cexp''$, $\csize = \csize' + \csize''$,
    $\sequT{\Gamma'}{\assign{r}{\sigma}}{\cmult'}{\cexp'}{\csize'}$,
    $\sequT{\Delta}{\assign{u}{\typeneutral}}{\cmult''}{\cexp''}{\csize''}$
    and $\ptight{\Gamma'(x)}$. Since $\Gamma$ is tight and
    $\ptight{\Gamma'(x)}$, then $\Gamma'$ and $\Delta$ are both tight as well.
    Then, the \ih with $r \in \CBVNF$ and $u \in \CBVNF$ gives $\cmult' =
    \cexp' = 0$ and $\cmult'' = \cexp'' = 0$. Hence, $\cmult = \cexp = 0$.
    
    \item if $\Phi$ ends with rule $\ruleVDESubs$, then $\cmult = \cmult' +
    \cmult''$, $\cexp = \cexp' + \cexp''$, $\csize = \csize' + \csize''$ and
    $\Gamma = \ctxtsum{(\ctxtres{\Gamma'}{x}{})}{\Delta}{}$ tight such that, in
    particular,
    $\derivable{\Phi_{u}}{\sequT{\Delta}{\assign{u}{\Gamma'(x)}}{\cmult}{\cexp}{\csize''}}{\SysTightCBV}$
    with $\Delta$ tight. Moreover, $t \in \CBVNF$ implies $u \in \HCBVNF$.
    By Lem.~\ref{l:value:tight-spreading} (\ref{l:value:tight-spreading:tight})
    $\Gamma'(x) \in \typetight$ which is a contradiction since $\Gamma'(x)$ is
    a multiset type by definition. Hence, this case does not apply.
  \end{enumerate}
\end{itemize}
\end{proof}


\begin{lemma}
Let
$\derivable{\Phi}{\sequT{\Gamma}{\assign{t}{\sigma}}{0}{0}{\csize}}{\SysTightCBV}$.
If $\Phi$ is tight, then $\csize = \valsize{t}$.
\label{l:value:tight-size}
\end{lemma}

\begin{proof}
By induction on $\Phi$.
\begin{itemize}
  \item $\ruleVTAxiomVar$, $\ruleVTAxiom$. Then $t = x$ and $\csize = 0 =
  \valsize{t}$ as expected.
  
  \item $\ruleVTArrowE$. Then $t = \termapp{r}{u}$, $\sigma = \typeneutral$,
  $\Gamma = \ctxtsum{\Gamma'}{\Delta}{}$, $\csize = \csize' + \csize'' + 1$,
  $\sequT{\Gamma'}{\assign{r}{\overline{\typevalue}}}{0}{0}{\csize'}$,
  $\sequT{\Delta}{\assign{u}{\overline{\typevar}}}{0}{0}{\csize''}$. Then,
  $\Gamma'$ and $\Delta$ are both tight, hence the \ih gives $\csize' =
  \valsize{r}$ and $\csize'' = \valsize{u}$. Hence, $\csize = \valsize{r} +
  \valsize{u} + 1 = \valsize{t}$.
  
  \item $\ruleVTArrowI$. Then $t = \termabs{x}{r}$, and $\csize = 0 =
  \valsize{t}$.
  
  \item $\ruleVTESubs$. Then $t = \termsubs{x}{u}{r}$, $\Gamma =
  \ctxtsum{(\ctxtres{\Gamma'}{x}{})}{\Delta}{}$, $\csize = \csize' + \csize''$,
  $\sequT{\Gamma'}{\assign{r}{\sigma}}{0}{0}{\csize'}$,
  $\sequT{\Delta}{\assign{u}{\typeneutral}}{0}{0}{\csize''}$ and
  $\ptight{\Gamma'(x)}$. Since $\Gamma$ is tight and $\ptight{\Gamma'(x)}$,
  then $\Gamma'$ and $\Delta$ are both tight as well. Thus, \ih gives $\csize'
  = \valsize{r}$ and $\csize'' = \valsize{u}$. Then, $\csize = \valsize{r} +
  \valsize{u} = \valsize{t}$. 
  
  \item $\ruleVDAxiom$. This case does not apply since it concludes with $\cexp
  > 0$.
  
  \item $\ruleVDArrowE$, $\ruleVDApp$. These cases do not apply since they
  conclude with $\cmult > 0$.
  
  \item $\ruleVDArrowI$. This case does not apply since, by
  Lem.~\ref{l:value:cexp-positive}, it concludes with $\cexp > 0$.
  
  \item $\ruleVDESubs$. Then $t = \termsubs{x}{u}{r}$, $\csize = \csize' +
  \csize''$ and $\Gamma = \ctxtsum{(\ctxtres{\Gamma'}{x}{})}{\Delta}{}$ tight
  such that, in particular,
  $\derivable{\Phi_{u}}{\sequT{\Delta}{\assign{u}{\Gamma'(x)}}{0}{0}{\csize''}}{\SysTightCBV}$
  with $\Delta$ tight. By Lem.~\ref{l:value:tight-spreading}
  (\ref{l:value:tight-spreading:tight}) on $\Phi_{u}$, $\Gamma'(x) \in
  \typetight$ which leads to a contradiction, since $\Gamma'(x)$ is a multiset
  type by definition. Hence, this case does not apply.
\end{itemize}
\end{proof}


A type derivation for a value with a multiset types might be split in
multiple derivations. 

\begin{lemma}[Split Type for Exact Value]
Let
$\derivable{\Phi}{\sequT{\Gamma}{\assign{v}{\M}}{\cmult}{\cexp}{\csize}}{\SysTightCBV}$
such that $\M = \ctxtsum{\!}{\M_i}{i \in I}$. Then, there exist
$\many{\derivable{\Phi_{v}^{i}}{\sequT{\Gamma_i}{\assign{v}{\M_i}}{\cmult_i}{\cexp_i}{\csize_i}}{\SysTightCBV}}{i \in I}$
such that $\Gamma = \ctxtsum{\!}{\Gamma_i}{i \in I}$, $\cmult =
+_{i \in I}{\cmult_i}$, $\cexp = 1 +_{i \in I}{\cexp_i} - |I|$, and
$\csize = +_{i \in I}{\csize_i}$.
\label{l:value:split-value-tight}
\end{lemma}

\begin{proof}
By case analysis on the form of $v$.
\begin{itemize}
  \item $v = x$. By $\ruleVDAxiom$
  $\derivable{\Phi}{\sequT{\assign{x}{\M}}{\assign{x}{\M}}{0}{1}{0}}{\SysTightCBV}$.
  From $\M = \ctxtsum{\!}{\M_i}{i \in I}$ we can build 
  $\many{\derivable{\Phi^{i}_{v}}{\sequT{\assign{x}{\M_i}}{\assign{x}{\M_i}}{0}{1}{0}}{\SysTightCBV}}{i \in I}$
  by using $\ruleVDAxiom$ several times, and conclude since
  $+_{i \in I}{\cexp_i} = |I|$ and $\cexp = 1$. Note that the persistent rules
  $\ruleVTAxiomVar$ and $\ruleVTAxiom$ do not apply since they do not conclude
  with a multiset type. 
  
  \item $v = \termabs{x}{r}$. By $\ruleVDArrowI$, we have $\Gamma =
  \ctxtsum{\!}{\ctxtres{\Gamma_j}{x}{}}{j \in J}$, $\M =
  \intertype{\functtype{\Gamma_j(x)}{\tau_j}}{j \in J}$, $\cmult =
  +_{j \in J}{\cmult_j}$, $\cexp = 1 +_{j \in J}{\cexp_j}$, $\csize =
  +_{j \in J}{\csize_i}$ and
  $\many{\derivable{\Phi^{j}_{r}}{\sequT{\Gamma_j}{\assign{r}{\tau_j}}{\cmult_j}{\cexp_j}{\csize_j}}{\SysTightCBV}}{j \in J}$
  for some $J$. Let $\M_i =
  \intertype{\functtype{\Gamma_k(x)}{\tau_k}}{k \in K_i}$ for each $i \in I$,
  where $J =  \biguplus_{i \in I}{K_i}$. Then, by using $\ruleVDArrowI$ we
  build \[
\many{\derivable{\Phi^{i}_{v}}{\sequT{\ctxtsum{\!}{\ctxtres{\Gamma_k}{x}{}}{k \in K_i}}{\assign{\termabs{x}{r}}{\M_i}}{+_{k \in K_i}{\cmult_k}}{1+_{k \in K_i}{\cexp_k}}{+_{k \in K_i}{\csize_k}}}{\SysTightCBV}}{i \in I}
  \] Finally, we conclude since:
  \begin{itemize}
    \item $\Gamma =
    \ctxtsum{\!}{\ctxtres{\Gamma_j}{x}{}}{j \in J} =
    \ctxtsum{\!}{(\ctxtsum{\!}{\ctxtres{\Gamma_k}{x}{}}{k \in K_i})}{i \in I}$
    
    \item $\cmult = +_{j \in J}{\cmult_j} =
    +_{i \in I}{(+_{k \in K_i}{\cmult_k})}$
    
    \item $\cexp = 1 +_{j \in J}{\cexp_j} = 1
    +_{i \in I}{(1 +_{k \in K_i}{\cexp_k})} - |I|$
    
    \item $\csize = +_{j \in J}{\csize_j} =
    +_{i \in I}{(+_{k \in K_i}{\csize_k})}$
  \end{itemize}
  As before, the persistent rule $\ruleVTArrowI$ does not apply since it does
  not conclude with a multiset type. 
\end{itemize}
\end{proof}


These auxiliary lemmas are used in particular to show how to compute the
counters when substituting typed values. This is the key property to obtain
subject reduction.

\begin{lemma}[Substitution]
Let
$\derivable{\Phi_{t}}{\sequT{\Gamma}{\assign{t}{\tau}}{\cmult}{\cexp}{\csize}}{\SysTightCBV}$
and
$\derivable{\Phi_{v}}{\sequT{\Delta}{\assign{v}{\Gamma(x)}}{\cmult'}{\cexp'}{\csize'}}{\SysTightCBV}$.
Then,
$\derivable{\Phi_{\substitute{x}{v}{t}}}{\sequT{\ctxtsum{\ctxtres{\Gamma}{x}{}}{\Delta}{}}{\assign{\substitute{x}{v}{t}}{\tau}}{\cmult+\cmult'}{\cexp+\cexp'-1}{\csize+\csize'}}{\SysTightCBV}$.
\label{l:value:substitution-tight}
\end{lemma}

\begin{proof}
By induction on $t$.
\begin{itemize}
  \item $t = x$. Then, $\substitute{x}{v}{t} = v$ and there are three
  possibilities for $\Phi_{t}$:
  \begin{itemize}
    \item $\ruleVTAxiomVar$. Then $\tau = \typevar$ and $\Gamma(x) =
    \multiset{\typevar}$. There is only one way of typing a value with
    $\multiset{\typevar}$, namely $\ruleVDAxiom$, hence $v$ is a variable and
    $\Delta = \assign{v}{\multiset{\typevar}}$. We conclude by taking
    $\derivable{\Phi_{\substitute{x}{v}{t}}}{\sequT{\Delta}{\assign{v}{\typevar}}{0}{0}{0}}{\SysTightCBV}$
    by rule $\ruleVTAxiomVar$. Note that the counters are as expected.
    
    \item $\ruleVTAxiom$. Then $\tau = \typevalue$ and $\Gamma = \emptyset$.
    There are only two ways of typing a value with type $\emul$, namely
    $\ruleVDAxiom$, if $v$ is a variable, and $\ruleVDArrowI$ with $I =
    \emptyset$, if $v$ is an abstraction. In both cases the counters of
    $\Phi_{v}$ are $(0,1,0)$. We conclude by taking
    $\derivable{\Phi_{\substitute{x}{v}{t}}}{\sequT{}{\assign{v}{\typevalue}}{0}{0}{0}}{\SysTightCBV}$
    by $\ruleVTAxiom$, if $v$ is a variable, or $\ruleVTArrowI$, if $v$ is an
    abstraction.
    
    \item $\ruleVDAxiom$. Then we have
    $\derivable{\Phi_{t}}{\sequT{\assign{x}{\tau}}{\assign{x}{\tau}}{0}{1}{0}}{\SysTightCBV}$,
    $\tau$ being a multiset type. We conclude by taking
    $\Phi_{\substitute{x}{v}{t}} = \Phi_{v}$ since the counters are as
    expected.
  \end{itemize}
  
  \item $t = y \neq x$. Then, $\substitute{x}{v}{t} = y$, and there are three
  possibilities for $\Phi_{t}$: $\ruleVTAxiomVar$, $\ruleVTAxiom$ and
  $\ruleVDAxiom$, with $\Gamma(x) = \emul$ in all cases. As observed above
  the counters of $\Phi_{v}$ are $(0,1,0)$ whenever $\Gamma(x) = \emul$, hence
  we conclude by taking $\Phi_{\substitute{x}{v}{t}} = \Phi_{t}$.
  
\item All the other cases are straightforward by the \ih and
Lem.~\ref{l:value:split-value-tight}.

\end{itemize}
\end{proof}


\begin{lemma}[Exact Subject Reduction]
Let
$\derivable{\Phi}{\sequT{\Gamma}{\assign{t}{\sigma}}{\cmult}{\cexp}{\csize}}{\SysTightCBV}$
be a tight derivation. If $t \loredv t'$, then there is
$\derivable{\Phi'}{\sequT{\Gamma}{\assign{t'}{\sigma}}{\cmult'}{\cexp'}{\csize}}{\SysTightCBV}$
such that
\begin{inparaenum}[(1)]
  \item\label{l:value:subject-reduction-tight:mult} $\cmult' = \cmult - 1$ and
  $\cexp' = \cexp$ if $t \loredv t'$ is an $\mStep$-step;
  \item\label{l:value:subject-reduction-tight:exp} $\cexp' = \cexp - 1$ and
  $\cmult' = \cmult$ if $t \loredv t'$ is an $\eStep$-step.
\end{inparaenum}
\label{l:value:subject-reduction-tight}
\end{lemma}

\begin{proof}
We actually prove the following stronger statement that allows us to reason
inductively:

Let $t \loredv t'$ and
$\derivable{\Phi}{\sequT{\Gamma}{\assign{t}{\sigma}}{\cmult}{\cexp}{\csize}}{\SysTightCBV}$
such that $\Gamma$ is tight, and either $\sigma \in \typetight$ or
$\neg\pabs{t}$. Then, there exists
$\derivable{\Phi'}{\sequT{\Gamma}{\assign{t'}{\sigma}}{\cmult'}{\cexp'}{\csize}}{\SysTightCBV}$
such that
\begin{enumerate}
  \item $\cmult' = \cmult - 1$ and $\cexp' = \cexp$ if $t \loredv t'$ is an
  $\mStep$-step.
  
  \item $\cexp' = \cexp - 1$ and $\cmult' = \cmult$ if $t \loredv t'$ is an
  $\eStep$-step.
\end{enumerate}

We proceed by induction on $t \loredv t'$.
\begin{itemize}
  \item $t = \termapp{\ctxtapp{\ctxt{L}}{\termabs{x}{r}}}{u} \loredv
  \ctxtapp{\ctxt{L}}{\termsubs{x}{u}{r}} = t'$ is an $\mStep$-step. We proceed
  by induction on $\ctxt{L}$. We only show here the case $\ctxt{L} = \Box$ as
  the inductive case is straightforward.

  We first remark that $\Phi$ cannot end with rule $\ruleVTArrowE$ since
  $\termabs{x}{u}$ cannot be typed neither with $\typevar$ nor with
  $\typeneutral$, then there are two cases depending on the last rule of
  $\Phi$.
  \begin{enumerate}
    \item $\ruleVDArrowE$. Then $\Phi$ has the following form: \[
\Rule{
  \Rule{
    \derivable{\Phi_{r}}{\sequT{\Gamma'}{\assign{r}{\sigma}}{\cmult_1}{\cexp_1}{\csize_1}}{\SysTightCBV}
  }{
    \sequT{\ctxtres{\Gamma'}{x}{}}{\assign{\termabs{x}{r}}{\multiset{\functtype{\Gamma'(x)}{\sigma}}}}{\cmult_1}{1+\cexp_1}{\csize_1}
  }{\ruleVDArrowI}
  \derivable{\Phi_{u}}{\sequT{\Delta}{\assign{u}{\Gamma'(x)}}{\cmult_2}{\cexp_2}{\csize_2}}{\SysTightCBV}
}{
  \sequT{\ctxtsum{\ctxtres{\Gamma'}{x}{}}{\Delta}{}}{\assign{\termapp{(\termabs{x}{r})}{u}}{\sigma}}{\cmult_1+\cmult_2+1}{\cexp_1+\cexp_2}{\csize_1+\csize_2}
}{\ruleVDArrowE}
    \] with $\cmult = \cmult_1 + \cmult_2 + 1$, $\cexp = \cexp_1 + \cexp_2$
    and $\csize = \csize_1 + \csize_2$. We conclude by $\ruleVDESubs$ with
    the following type derivation \[
\Rule{
  \derivable{\Phi_{r}}{\sequT{\Gamma'}{\assign{r}{\sigma}}{\cmult_1}{\cexp_1}{\csize_1}}{\SysTightCBV}
  \quad
  \derivable{\Phi_{u}}{\sequT{\Delta}{\assign{u}{\Gamma'(x)}}{\cmult_2}{\cexp_2}{\csize_2}}{\SysTightCBV}
}{
  \derivable{\Phi'}{\sequT{\ctxtsum{\ctxtres{\Gamma'}{x}{}}{\Delta}{}}{\assign{\termsubs{x}{u}{r}}{\sigma}}{\cmult_1+\cmult_2}{\cexp_1+\cexp_2}{\csize_1+\csize_2}}{\SysTightCBV}
}{\ruleVDESubs}
    \] taking $\cmult' = \cmult_1 + \cmult_2 = \cmult - 1$ and $\cexp' =
    \cexp_1 + \cexp_2 = \cexp$.
    
    \item $\ruleVDApp$. Then $\Phi$ has the following form: \[
\Rule{
  \Rule{
    \derivable{\Phi_{r}}{\sequT{\Gamma'}{\assign{r}{\sigma}}{\cmult_1}{\cexp_1}{\csize_1}}{\SysTightCBV}
  }{
    \sequT{\ctxtres{\Gamma'}{x}{}}{\assign{\termabs{x}{r}}{\multiset{\functtype{\Gamma'(x)}{\sigma}}}}{\cmult_1}{1+\cexp_1}{\csize_1}
  }{\ruleVDArrowI}
  \derivable{\Phi_{u}}{\sequT{\Delta}{\assign{u}{\typeneutral}}{\cmult_2}{\cexp_2}{\csize_2}}{\SysTightCBV}
  \quad
  \ptight{\Gamma'(x)}
}{
  \sequT{\ctxtsum{\ctxtres{\Gamma'}{x}{}}{\Delta}{}}{\assign{\termapp{(\termabs{x}{r})}{u}}{\sigma}}{\cmult_1+\cmult_2+1}{\cexp_1+\cexp_2}{\csize_1+\csize_2}
}{\ruleVDApp}
    \] with $\cmult = \cmult_1 + \cmult_2 + 1$, $\cexp = \cexp_1 + \cexp_2$
    and $\csize = \csize_1 + \csize_2$. We conclude by $\ruleVTESubs$ with
    the following type derivation \[
\Rule{
  \derivable{\Phi_{r}}{\sequT{\Gamma'}{\assign{r}{\sigma}}{\cmult_1}{\cexp_1}{\csize_1}}{\SysTightCBV}
  \quad
  \derivable{\Phi_{u}}{\sequT{\Delta}{\assign{u}{\typeneutral}}{\cmult_2}{\cexp_2}{\csize_2}}{\SysTightCBV}
  \quad
  \ptight{\Gamma'(x)}
}{
  \derivable{\Phi'}{\sequT{\ctxtsum{\ctxtres{\Gamma'}{x}{}}{\Delta}{}}{\assign{\termsubs{x}{u}{r}}{\sigma}}{\cmult_1+\cmult_2}{\cexp_1+\cexp_2}{\csize_1+\csize_2}}{\SysTightCBV}
}{\ruleVTESubs}
    \] taking $\cmult' = \cmult_1 + \cmult_2 = \cmult - 1$ and $\cexp' =
    \cexp_1 + \cexp_2 = \cexp$.
  \end{enumerate}
  
  \item $t = \termsubs{x}{\ctxtapp{\ctxt{L}}{v}}{r} \loredv
  \ctxtapp{\ctxt{L}}{\substitute{x}{v}{r}} = t'$ is an $\eStep$-step. We
  proceed by induction on $\ctxt{L}$. We only show here the case $\ctxt{L} =
  \Box$ as the inductive case is straightforward.

  Then, $t = \termsubs{x}{v}{r}$. Note that $\ruleVTESubs$ does not apply since
  it requires the value $v$ to be typed with $\typeneutral$, which is not
  possible. Thus, rule $\ruleVDESubs$ is necessarily used and we have \[
\Rule{
  \derivable{\Phi_{r}}{\sequT{\Gamma'}{\assign{r}{\sigma}}{\cmult_1}{\cexp_1}{\csize_1}}{\SysTightCBV}
  \quad
  \derivable{\Phi_{v}}{\sequT{\Delta}{\assign{v}{\Gamma'(x)}}{\cmult_2}{\cexp_2}{\csize_2}}{\SysTightCBV}
}{
  \derivable{\Phi}{\sequT{\ctxtsum{\ctxtres{\Gamma'}{x}{}}{\Delta}{}}{\assign{\termsubs{x}{v}{r}}{\sigma}}{\cmult_1+\cmult_2}{\cexp_1+\cexp_2}{\csize_1+\csize_2}}{\SysTightCBV}
}{\ruleVDESubs}
  \] with $\cmult = \cmult_1 + \cmult_2$, $\cexp = \cexp_1 + \cexp_2$ and
  $\csize = \csize_1 + \csize_2$. We conclude by applying
  Lem.~\ref{l:value:substitution-tight} with $\Phi_{r}$ and $\Phi_{v}$, thus
  obtaining
  $\derivable{\Phi'}{\sequT{\ctxtsum{\ctxtres{\Gamma'}{x}{}}{\Delta}{}}{\assign{\substitute{x}{v}{t}}{\sigma}}{\cmult_1+\cmult_2}{\cexp_1+\cexp_2-1}{\csize_1+\csize_2}}{\SysTightCBV}$.
  Note that the counters are as expected.
  
  \item $t = \termapp{r}{u} \loredv \termapp{r'}{u} = t'$, where $r \loredv r'$
  and $\neg\pabs{r}$. Then, $\Phi$ ends with rule $\ruleVTArrowE$,
  $\ruleVDArrowE$ or $\ruleVDApp$, and in either case the subterm $r$ has an
  associated typing derivation $\Phi_{r}$ which verifies the hypothesis since
  $\neg\pabs{r}$. Hence, by \ih there exists a $\Phi_{r'}$ with proper counters
  such that we can derive $\Phi'$ applying the same rule as in $\Phi$ to
  conclude.

  \item $t = \termapp{r}{u} \loredv \termapp{r}{u'} = t'$, where $u \loredv u'$
  and $r \in \HCBVNF \cup \VarHCBVNF$. There are two further cases to consider:
  \begin{enumerate}
    \item If $\Phi$ ends with rule $\ruleVTArrowE$, then the subterm $u$ has an
    associated typing derivation which verifies the hypothesis (context and
    type are tight). Then the property holds by \ih
    
    \item If $\Phi$ ends with rule $\ruleVDArrowE$ or $\ruleVDApp$, then the
    subterm $r$ has an associated typing derivation with a tight context and a
    functional type $\multiset{\functtype{\M}{\sigma}}$.
    Lem.~\ref{l:value:tight-spreading} (\ref{l:value:tight-spreading:var})
    gives $r \notin \VarHCBVNF$ so that $r \in \HCBVNF$. Then,
    Lem.~\ref{l:value:tight-spreading} (\ref{l:value:tight-spreading:tight})
    gives $\multiset{\functtype{\M}{\sigma}} \in \typetight$ which leads to a
    contradiction. So this case is not possible.
  \end{enumerate}
  
  \item $t = \termsubs{x}{u}{r} \loredv \termsubs{x}{u}{r'} = t'$, where $r
  \loredv r'$ and $u \in \HCBVNF$. There are two further cases to consider:
  \begin{enumerate}
    \item If $\Phi$ ends with rule $\ruleVTESubs$ then the subterm $r$ has an
    associated typing derivation which verifies the hypothesis (context tight
    and $\neg\pabs{r}$). Then the property holds by \ih

    \item If $\Phi$ ends with rule $\ruleVDESubs$ then the subterm $u$ has an
    associated typing derivation with a tight context and a multiset type $\M$.
    Lem.~\ref{l:value:tight-spreading} (\ref{l:value:tight-spreading:tight})
    gives $\M \in \typetight$ which leads to a contradiction. So this case is
    not possible.
  \end{enumerate}
  
  \item $t = \termsubs{x}{u}{r} \loredv \termsubs{x}{u'}{r} = t'$, where $u
  \loredv u'$ and $\neg\pval{u}$. There are two further cases to consider:
  \begin{enumerate}
    \item If $\Phi$ ends with rule $\ruleVTESubs$ then the subterm $u$ has an
    associated typing derivation which verifies the hypothesis (context and
    type are tight). Then the property holds by the \ih
  
    \item If $\Phi$ ends with rule $\ruleVDESubs$ then the subterm $u$ has an
    associated typing derivation $\Phi_{u}$ which verifies the hypothesis since
    $\neg\pval{u}$ implies $\neg\pabs{u}$. Hence, by \ih there exists a
    $\Phi_{u'}$ with proper counters such that we can derive $\Phi'$ applying
    $\ruleVDESubs$ as in $\Phi$ to conclude.
  \end{enumerate}
\end{itemize}
\end{proof}


An interesting remark is that types of subterms in tight derivations may change
during CBV reduction, unlike other approaches~\cite{AccattoliGK18,LeberlePhD}.
To illustrate this phenomenon, consider the following reduction
$\termsubs{x}{\id}{x} \loredv \id$. The terms on the left and right hand side
can be resp. tightly typed by the following derivations (we omit the counters):
\begin{center}
   \begin{prooftree}
    \begin{prooftree}
    \justifies{\sequ{}{\assign{x}{\typevalue}}}
    \using{\ruleVTAxiom}
   \end{prooftree} \sep
  \begin{prooftree}
    \justifies{\sequ{}{\assign{\id}{\emul}}}
      \using{\ruleVDArrowI}
  \end{prooftree}
  \justifies{\sequ{}{\assign{\termsubs{x}{\id}{x}}{\typevalue}}}
  \using{\ruleVDESubs}
   \end{prooftree}
   \sep \sep \sep \begin{prooftree}
     \phantom{\begin{prooftree}
    \justifies{\sequ{}{\assign{\id}{\emul}}}
      \using{\ruleVDArrowI}
  \end{prooftree}}
     \justifies{\sequ{}{\assign{\id}{\typevalue}}}
     \using{\ruleVTArrowI}
     \end{prooftree} 
\end{center}
Notice that the identity function $\id$ is typed differently in each
derivation: the substitution introduced by rule $\ruleVDESubs$ is a consuming
constructor, which disappears when $\sVal$-reduction consumes its value
argument, a phenomenon that is captured by means of a multitype for the
argument of the substitution. Indeed, $\id$ on the left-hand side derivation is
typed with the multitype $\emul$. This value $\id$ substitutes a variable $x$
typed with $\typevalue$, which is just a placeholder for a persistent value.
Thus, once the substitution is performed, the identity $\id$ becomes persistent
on the right-hand side, a phenomenon which is naturally captured by the tight
type $\typevalue$. This observation also applies to the tight typing system of
CBPV in
Sec.~\ref{s:tight}.
  
\begin{theorem}[Soundness]
If
$\derivable{\Phi}{\sequT{\Gamma}{\assign{t}{\sigma}}{\cmult}{\cexp}{\csize}}{\SysTightCBV}$
is tight, then there exists $p$ such that $p \in \CBVNF$ and
$t \rewriten{\callbyvalue}^{(\cmult,\cexp)} p$ with $\cmult$ $\mStep$-steps,
$\cexp$ $\eStep$-steps, and $\valsize{p} = \csize$.
\label{t:value:correctness-tight}
\end{theorem}

\begin{proof}
We prove the statement for $\loredv$ and then conclude for the general
notion of reduction $\rewrite{\callbyvalue}$ by the observation that all
reduction sequences to normal form have the same number of multiplicative and
exponential steps. Let
$\derivable{\Phi}{\sequT{\Gamma}{\assign{t}{\sigma}}{\cmult}{\cexp}{\csize}}{\SysTight}$.
We proceed by induction on $\cmult + \cexp$:
\begin{itemize}
  \item If $\cmult + \cexp = 0$, then $\cmult = \cexp = 0$ and
  Lem.~\ref{l:value:czero-normal} gives $t \in \CBVNF$. Moreover, by
  Lem.~\ref{l:value:tight-size} we get $\valsize{t} = \csize$. Thus, we
  conclude with $p = t$.

  \item If $\cmult + \cexp > 0$, then $t \notin \CBVNF$ holds by
  Lem.~\ref{l:value:czero-normal} and thus there exists $t'$ such that $t
  \loredvn^{(1,0)} t'$ or $t \loredvn^{(0,1)} t'$ by
  Prop.~\ref{l:cbn-cbv-normal-forms}. By
  Lem.~\ref{l:value:subject-reduction-tight} there exists a type derivation
  $\derivable{\Phi'}{\sequT{\Gamma}{\assign{t'}{\sigma}}{\cmult'}{\cexp'}{\csize}}{\SysTight}$
  such that $1 + \cmult' + \cexp' = \cmult + \cexp$. By the \ih there exists $p
  \in \CBVNF$ such that $t' \loredvn^{(\cmult',\cexp')} p$ with $\csize =
  \valsize{p}$. Then $t \loredvn^{(1,0)} t' \loredvn^{(\cmult',\cexp')} p$ or
  $t \loredvn^{(0,1)} t' \loredvn^{(\cmult',\cexp')} p$ which means $t
  \loredvn^{(\cmult,\cexp)} p$, as expected.
\end{itemize}
\end{proof}


As in CBN, the previous theorem is stated by using the general notion of
reduction $\rewrite{\callbyvalue}$, but the proofs (notably
Lem.~\ref{l:value:subject-reduction-tight}) are done using the deterministic
reduction $\loredv$. Similar comment applies to the forthcoming
Thm.~\ref{t:value:completeness-tight}.

\parrafo{Completeness}
The completeness result is also based on intermediate lemmas, most notably the
so-called \emph{tight typing of normal forms} and \emph{subject expansion}
properties:

\begin{lemma}[Tight Typing of Normal Forms]
If $t \in \CBVNF$, then there is a tight derivation
$\derivable{\Phi}{\sequT{\Gamma}{\assign{t}{\sigma}}{0}{0}{\valsize{t}}}{\SysTightCBV}$.
\label{l:value:normal-forms-tight}
\end{lemma}

\begin{proof}
By simultaneous induction on the following claims:
\begin{enumerate}
  \item\label{l:value:normal-forms-tight:ax} If $t \in \VarHCBVNF$, then there
  exists a tight derivation
  $\derivable{\Phi}{\sequT{\Gamma}{\assign{t}{\typevar}}{0}{0}{\valsize{t}}}{\SysTightCBV}$.
  
  \item\label{l:value:normal-forms-tight:vr} If $t \in \VarHCBVNF$, then there
  exists a tight derivation
  $\derivable{\Phi}{\sequT{\Gamma}{\assign{t}{\typevalue}}{0}{0}{\valsize{t}}}{\SysTightCBV}$.
  
  \item\label{l:value:normal-forms-tight:ne} If $t \in \HCBVNF$, then there
  exists a tight derivation
  $\derivable{\Phi}{\sequT{\Gamma}{\assign{t}{\typeneutral}}{0}{0}{\valsize{t}}}{\SysTightCBV}$.
  
  \item\label{l:value:normal-forms-tight:no} If $t \in \CBVNF$, then there
  exists a tight derivation
  $\derivable{\Phi}{\sequT{\Gamma}{\assign{t}{\typetight}}{0}{0}{\valsize{t}}}{\SysTightCBV}$.
\end{enumerate}
\begin{itemize}
  \item $t = x$. Then, $t \in \VarHCBVNF$ and we conclude
  (\ref{l:value:normal-forms-tight:ax}) by $\ruleVTAxiomVar$ or
  (\ref{l:value:normal-forms-tight:vr}) by $\ruleVTAxiom$, since $\valsize{x} =
  0$ by definition.
  
  \item $t = \termapp{r}{u}$. Then, $t \in \HCBVNF$ with $r \in \VarHCBVNF \cup
  \HCBVNF$ and $u \in \CBVNF$. By \ih (\ref{l:value:normal-forms-tight:no}) on
  $u$ we have
  $\derivable{\Phi_{u}}{\sequT{\Delta}{\assign{u}{\typetight}}{0}{0}{\valsize{u}}}{\SysTightCBV}$
  with $\Delta$ tight. There are two options for $r$, namely $r \in \VarHCBVNF$
  or $r \in \HCBVNF$. By \ih (\ref{l:value:normal-forms-tight:ax}) or
  (\ref{l:value:normal-forms-tight:ne}) resp. there exists a tight derivation
  $\derivable{\Phi_{r}}{\sequT{\Gamma'}{\assign{r}{\overline{\typevalue}}}{0}{0}{\valsize{t}}}{\SysTightCBV}$.
  We conclude (\ref{l:value:normal-forms-tight:ne}) by $\ruleVTArrowE$ with
  $\Phi_{r}$ and $\Phi_{u}$, obtaining $\Gamma = \ctxtsum{\Gamma'}{\Delta}{}$
  and $\valsize{t} = \valsize{r} + \valsize{u} + 1$ as expected.
  
  \item $t = \termabs{x}{r}$. Then, we conclude
  (\ref{l:value:normal-forms-tight:no}) by $\ruleVTArrowI$ since $\valsize{t} =
  0$ by definition.
  
  \item $t = \termsubs{x}{u}{r}$. Then, there are three cases to consider: $t
  \in \VarHCBVNF$, $t \in \HCBVNF$ or $t \in \CBVNF$. In either of them we have
  $u \in \HCBVNF$. Thus, by \ih (\ref{l:value:normal-forms-tight:ne}) there
  exists a tight derivation
  $\derivable{\Phi_{u}}{\sequT{\Delta}{\assign{u}{\typeneutral}}{0}{0}{\valsize{u}}}{\SysTightCBV}$.
  Moreover, in each case we get $r$ in the same set as $t$ by definition. By
  applying the respective \ih on $r$ we get a tight derivation
  $\derivable{\Phi_{r}}{\sequT{\Gamma'}{\assign{r}{\sigma}}{0}{0}{\valsize{r}}}{\SysTightCBV}$
  with the appropiate type for each case. Finally, since $\Gamma'$ tight
  implies $\Gamma'(x)$ tight, we conclude each respective claim by
  $\ruleVTESubs$ with $\Gamma = \ctxtsum{\ctxtres{\Gamma'}{x}{}}{\Delta}{}$
  tight and, hence,
  $\derivable{\Phi}{\sequT{\Gamma}{\assign{\termsubs{x}{u}{r}}{\sigma}}{0}{0}{\valsize{r} + \valsize{u}}}{\SysTightCBV}$
  tight as well. Note that $\valsize{t} = \valsize{r} + \valsize{u}$.
\end{itemize}
\end{proof}


Analogously to the soundness case, we may obtain a single type
derivation for a value with multiset type from several ones, properly combining
their counters, and hence leading to an anti-subtitution result.

\begin{lemma}[Merge type for value]
Let
$\many{\derivable{\Phi_{v}^{i}}{\sequT{\Gamma_i}{\assign{v}{\M_i}}{\cmult_i}{\cexp_i}{\csize_i}}{\SysTightCBV}}{i \in I}$.
Then, there exists
$\derivable{\Phi_{v}}{\sequT{\Gamma}{\assign{v}{\M}}{\cmult}{\cexp}{\csize}}{\SysTightCBV}$.
such that $\Gamma = \ctxtsum{\!}{\Gamma_i}{i \in I}$, $\M =
\ctxtsum{\!}{\M_i}{i \in I}$, $\cmult = +_{i \in I}{\cmult_i}$, $\cexp = 1
+_{i \in I}{\cexp_i} - |I|$, and $\csize = +_{i \in I}{\csize_i}$.
\label{l:value:merge-tight}
\end{lemma}

\begin{proof}
By case analysis on the form of $v$.
\begin{itemize}
  \item $v = x$. By $\ruleVDAxiom$, we have
  $\many{\derivable{\Phi^i_{v}}{\sequT{\assign{x}{\M_i}}{\assign{x}{\M_i}}{0}{1}{0}}{\SysTightCBV}}{i \in I}$.
  Then, we conclude by $\ruleVDAxiom$ with
  $\derivable{\Phi_v}{\sequT{\assign{x}{\M}}{\assign{x}{\M}}{0}{1}{0}}{\SysTightCBV}$,
  where $\M = \ctxtsum{\!}{\M_i}{i \in I}$ and hence $\Gamma =
  \ctxtsum{\!}{\Gamma_i}{i \in I}$ as expected. Moreover,
  $+_{i \in I}{\cmult_i} = +_{i \in I}{\csize_i} = 0$ and
  $+_{i \in I}{\cexp_i} = |I|$, hence $\cexp = 1 +_{i \in I}{\cexp_i} - |I|$.
  Note that the persistent rules $\ruleVTAxiomVar$ and $\ruleVTAxiom$ do not
  apply since they do not conclude with a multiset type.
  
  \item $v = \termabs{x}{r}$. By $\ruleVDArrowI$, for each $i \in I$ we have
  $\Gamma_i = \ctxtsum{\!}{\ctxtres{\Gamma_k}{x}{}}{k \in K_i}$, $\M_i =
  \intertype{\functtype{\Gamma_k(x)}{\tau_k}}{k \in K_i}$ with $\cmult_i =
  +_{k \in K_i}{\cmult_k}$, $\cexp_i = 1 +_{k \in K_i}{\cexp_k}$, $\csize_i
  = +_{k \in K_i}{\csize_k}$ and type derivations
  $\many{\derivable{\Phi^k_{r}}{\sequT{\Gamma_k}{\assign{r}{\tau_k}}{\cmult_k}{\cexp_k}{\csize_k}}{\SysTightCBV}}{k \in K_i}$
  for each $i \in I$. Let $J = \biguplus_{i \in I}{K_i}$ and $\M =
  \intertype{\functtype{\Gamma_j(x)}{\tau_j}}{j \in J} =
  \intertype{\functtype{\Gamma_k(x)}{\tau_k}}{k \in +_{i \in I}{K_i}}$ We
  conclude since, by $\ruleVDArrowI$, we can build the derivation
  $\derivable{\Phi_{v}}{\sequT{\Gamma}{\assign{v}{\M}}{+_{j \in J}{\cmult_j}}{1+_{j \in J}{\cexp_j}}{+_{j \in J}{\csize_j}}}{\SysTightCBV}$
  where:
  \begin{itemize}
    \item $\Gamma = \ctxtsum{\!}{\ctxtres{\Gamma_j}{x}{}}{j \in J} =
    \ctxtsum{\!}{(\ctxtsum{\!}{\ctxtres{\Gamma_k}{x}{}}{k \in K_i})}{i \in I} =
    \ctxtsum{\!}{\Gamma_i}{i \in I}$
    
    \item $\cmult = +_{j \in J}{\cmult_j} =
    +_{i \in I}{(+_{k \in K_i}{\cmult_k})} = +_{i \in I}{\cmult_i}$
    
    \item $\cexp = 1 +_{j \in J}{\cexp_j} = 1
    +_{i \in I}{(1 +_{k \in K_i}{\cexp_k})} - |I| = 1 +_{i \in I}{\cexp_i} -
    |I|$
    
    \item $\csize = +_{j \in J}{\csize_j} =
    +_{i \in I}{(+_{k \in K_i}{\csize_k})} = +_{i \in I}{\csize_i}$
  \end{itemize}
  As before, the persistent rule $\ruleVTArrowI$ does not apply since it does
  not conclude with a multiset type.
\end{itemize}
\end{proof}


Dually to the substitution lemma, the following property states how to extract
type derivations (and thus counters) from substituted term.

\begin{lemma}[Anti-Substitution]
Let
$\derivable{\Phi_{\substitute{x}{v}{t}}}{\sequT{\Gamma'}{\assign{\substitute{x}{v}{t}}{\tau}}{\cmult''}{\cexp''}{\csize''}}{\SysTightCBV}$.
Then, there exist
$\derivable{\Phi_{t}}{\sequT{\Gamma}{\assign{t}{\tau}}{\cmult}{\cexp}{\csize}}{\SysTightCBV}$
and
$\derivable{\Phi_{v}}{\sequT{\Delta}{\assign{v}{\Gamma(x)}}{\cmult'}{\cexp'}{\csize'}}{\SysTightCBV}$
such that $\Gamma' = \ctxtsum{\ctxtres{\Gamma}{x}{}}{\Delta}{}$, $\cmult'' =
\cmult + \cmult'$, $\cexp'' = \cexp + \cexp' - 1$ and $\csize'' = \csize +
\csize'$.
\label{l:value:anti-substitution-tight}
\end{lemma}

\begin{proof}
By induction on $t$.
\begin{itemize}
  \item $t = x$. Then, $\substitute{x}{v}{t} = v$ and we analyze the shape of
  $\tau$:
  \begin{itemize}
    \item $\tau \in \typetight$. Then, there are three further possibilities:
    \begin{itemize}
      \item $\tau = \typeneutral$. This case does not apply since there is no
      way to derive type $\typeneutral$ for a value.
      
      \item $\tau = \typevalue$. Then, $\Phi_{\substitute{x}{v}{t}}$ can only
      be $\ruleVTAxiom$ or $\ruleVTArrowI$. Thus, $\Gamma'$ is empty and
      $\cmult'' = \cexp'' = \csize'' = 0$. We conclude with
      $\derivable{\Phi_{t}}{\sequT{}{\assign{x}{\typevalue}}{0}{0}{0}}{\SysTightCBV}$
      by $\ruleVTAxiom$, and
      $\derivable{\Phi_{v}}{\sequT{}{\assign{v}{\emul}}{0}{1}{0}}{\SysTightCBV}$
      by either $\ruleVDAxiom$ or $\ruleVDArrowI$ with $I = \emptyset$.
      
      \item $\tau = \typevar$. Then, $v$ is necessarily a variable $z$,
      since abstractions cannot be assigned $\typevar$, and
      $\Phi_{\substitute{x}{v}{t}}$ can only be $\ruleVTAxiomVar$. Thus, $\Gamma' =
      \assign{z}{\multiset{\typevar}}$ and $\cmult'' = \cexp'' = \csize'' =
      0$. We conclude with
      $\derivable{\Phi_{t}}{\sequT{\assign{x}{\multiset{\typevar}}}{\assign{x}{\typevar}}{0}{0}{0}}{\SysTightCBV}$
      by $\ruleVTAxiomVar$ and 
      $\derivable{\Phi_{v}}{\sequT{\assign{z}{\multiset{\typevar}}}{\assign{z}{\multiset{\typevar}}}{0}{1}{0}}{\SysTightCBV}$
      by $\ruleVDAxiom$.
    \end{itemize}
    
    \item $\tau = \M$. Then, we take $\Phi_{v} = \Phi_{\substitute{x}{v}{t}}$
    and
    $\derivable{\Phi_{t}}{\sequT{\assign{x}{\tau}}{\assign{x}{\tau}}{0}{1}{0}}{\SysTightCBV}$
    by $\ruleVDAxiom$ to conclude. Note that the counters are as expected.
    
    \item $\tau = \functtype{\M}{\sigma}$. This case does not apply since there
    is no way to derive an arrow type for a value.
  \end{itemize}
  
  \item $t = y \neq x$. Then, $\substitute{x}{v}{t} = y$ and
  $\Phi_{\substitute{x}{v}{t}}$ is either
  $\sequT{\assign{y}{\multiset{\typevar}}}{\assign{y}{\typevar}}{0}{0}{0}$
  by $\ruleVTAxiomVar$, or
  $\sequT{}{\assign{y}{\typevalue}}{0}{0}{0}$ by $\ruleVTAxiom$, or
  $\sequT{}{\assign{y}{\emul}}{0}{1}{0}$ by $\ruleVDAxiom$. In either case we
  have $\Gamma'(x) = \emul$, thus we take $\Phi_{t} =
  \Phi_{\substitute{x}{v}{t}}$ and conclude with
  $\derivable{\Phi_{v}}{\sequT{}{\assign{v}{\emul}}{0}{1}{0}}{\SysTightCBV}$ by
  either $\ruleVDAxiom$ or $\ruleVDArrowI$ with $I = \emptyset$. Note that the
  counters are as expected.
  
\item All the other cases are straightforward by the \ih and
Lem.~\ref{l:value:merge-tight}.

\end{itemize}
\end{proof}


\begin{lemma}[Exact Subject Expansion]
Let
$\derivable{\Phi'}{\sequT{\Gamma}{\assign{t'}{\sigma}}{\cmult'}{\cexp'}{\csize}}{\SysTightCBV}$
be a tight derivation. If $t \loredv t'$, then there is
$\derivable{\Phi}{\sequT{\Gamma}{\assign{t}{\sigma}}{\cmult}{\cexp}{\csize}}{\SysTightCBV}$
such that
\begin{inparaenum}[(1)]
  \item\label{l:value:subject-expansion-tight:mult} $\cmult' = \cmult - 1$ and
  $\cexp' = \cexp$ if $t \loredv t'$ is an $\mStep$-step;
  \item\label{l:value:subject-expansion-tight:exp} $\cexp' = \cexp - 1$ and
  $\cmult' = \cmult$ if $t \loredv t'$ is an $\eStep$-step.
\end{inparaenum}
\label{l:value:subject-expansion-tight}
\end{lemma}

\begin{proof}
We actually prove the following stronger statement that allows us to reason
inductively:

Let $t \loredv t'$ and
$\derivable{\Phi'}{\sequT{\Gamma}{\assign{t'}{\sigma}}{\cmult'}{\cexp'}{\csize}}{\SysTightCBV}$
such that $\Gamma$ is tight, and either $\sigma \in \typetight$ or
$\neg\pabs{t}$. Then, there exists
$\derivable{\Phi}{\sequT{\Gamma}{\assign{t}{\sigma}}{\cmult}{\cexp}{\csize}}{\SysTightCBV}$
such that
\begin{enumerate}
  \item $\cmult' = \cmult - 1$ and $\cexp' = \cexp$ if $t \loredv t'$ is an
  $\mStep$-step.
  
  \item $\cexp' = \cexp - 1$ and $\cmult' = \cmult$ if $t \loredv t'$ is an
  $\eStep$-step.
\end{enumerate}

We proceed by induction on $t \loredv t'$.
\begin{itemize}
  \item $t = \termapp{\ctxtapp{\ctxt{L}}{\termabs{x}{r}}}{u}
  \loredv \ctxtapp{\ctxt{L}}{\termsubs{x}{u}{r}} = t'$ is an $\mStep$-step. We
  proceed by induction on $\ctxt{L}$.  We only show here the case $\ctxt{L} =
  \Box$ as the inductive case is straightforward.

  Then, $t' = \termsubs{x}{u}{r}$ and there are two possibilities for $\Phi'$.
  \begin{enumerate}
    \item $\ruleVTESubs$. Then $\Phi'$ has the following form: \[
\Rule{
  \derivable{\Phi_{r}}{\sequT{\Gamma'}{\assign{r}{\sigma}}{\cmult_1}{\cexp_1}{\csize_1}}{\SysTightCBV}
  \quad
  \derivable{\Phi_{u}}{\sequT{\Delta}{\assign{u}{\typeneutral}}{\cmult_2}{\cexp_2}{\csize_2}}{\SysTightCBV}
  \quad
  \ptight{\Gamma'(x)}
}{
  \sequT{\ctxtsum{\ctxtres{\Gamma'}{x}{}}{\Delta}{}}{\assign{\termsubs{x}{u}{r}}{\sigma}}{\cmult_1+\cmult_2}{\cexp_1+\cexp_2}{\csize_1+\csize_2}
}{\ruleVTESubs}
    \] with $\cmult' = \cmult_1 + \cmult_2$, $\cexp' = \cexp_1 + \cexp_2$
    and $\csize = \csize_1 + \csize_2$. We conclude by $\ruleVDApp$ and
    $\ruleVDArrowI$ with the following type derivation \[
\Rule{
  \Rule{
    \derivable{\Phi_{r}}{\sequT{\Gamma'}{\assign{r}{\sigma}}{\cmult_1}{\cexp_1}{\csize_1}}{\SysTightCBV}
  }{
    \sequT{\ctxtres{\Gamma'}{x}{}}{\assign{\termabs{x}{r}}{\multiset{\functtype{\Gamma'(x)}{\sigma}}}}{\cmult_1}{1+\cexp_1}{\csize_1}
  }{\ruleVDArrowI}
  \derivable{\Phi_{u}}{\sequT{\Delta}{\assign{u}{\typeneutral}}{\cmult_2}{\cexp_2}{\csize_2}}{\SysTightCBV}
  \quad
  \ptight{\Gamma'(x)}
}{
  \derivable{\Phi}{\sequT{\ctxtsum{\ctxtres{\Gamma'}{x}{}}{\Delta}{}}{\assign{\termapp{(\termabs{x}{r})}{u}}{\sigma}}{\cmult_1+\cmult_2+1}{\cexp_1+\cexp_2}{\csize_1+\csize_2}}{\SysTightCBV}
}{\ruleVDApp}
    \] taking $\cmult = \cmult_1 + \cmult_2 + 1 = \cmult' + 1$ and $\cexp =
    \cexp_1 + \cexp_2 = \cexp'$.
    
    \item $\ruleVDESubs$. Then $\Phi'$ has the following form: \[
\Rule{
  \derivable{\Phi_{r}}{\sequT{\Gamma'}{\assign{r}{\sigma}}{\cmult_1}{\cexp_1}{\csize_1}}{\SysTightCBV}
  \quad
  \derivable{\Phi_{u}}{\sequT{\Delta}{\assign{u}{\Gamma'(x)}}{\cmult_2}{\cexp_2}{\csize_2}}{\SysTightCBV}
}{
  \sequT{\ctxtsum{\ctxtres{\Gamma'}{x}{}}{\Delta}{}}{\assign{\termsubs{x}{u}{r}}{\sigma}}{\cmult_1+\cmult_2}{\cexp_1+\cexp_2}{\csize_1+\csize_2}
}{\ruleVDESubs}
    \] with $\cmult' = \cmult_1 + \cmult_2$, $\cexp' = \cexp_1 + \cexp_2$
    and $\csize = \csize_1 + \csize_2$. We conclude by $\ruleVDArrowE$ and
    $\ruleVDArrowI$ with the following type derivation \[
\Rule{
  \Rule{
    \derivable{\Phi_{r}}{\sequT{\Gamma'}{\assign{r}{\sigma}}{\cmult_1}{\cexp_1}{\csize_1}}{\SysTightCBV}
  }{
    \sequT{\ctxtres{\Gamma'}{x}{}}{\assign{\termabs{x}{r}}{\multiset{\functtype{\Gamma'(x)}{\sigma}}}}{\cmult_1}{1+\cexp_1}{\csize_1}
  }{\ruleVDArrowI}
  \derivable{\Phi_{u}}{\sequT{\Delta}{\assign{u}{\Gamma'(x)}}{\cmult_2}{\cexp_2}{\csize_2}}{\SysTightCBV}
}{
  \derivable{\Phi}{\sequT{\ctxtsum{\ctxtres{\Gamma'}{x}{}}{\Delta}{}}{\assign{\termapp{(\termabs{x}{r})}{u}}{\sigma}}{\cmult_1+\cmult_2+1}{\cexp_1+\cexp_2}{\csize_1+\csize_2}}{\SysTightCBV}
}{\ruleVDArrowE}
    \] taking $\cmult = \cmult_1 + \cmult_2 + 1 = \cmult' + 1$ and $\cexp =
    \cexp_1 + \cexp_2 = \cexp'$.
  \end{enumerate}
    
  \item $t = \termsubs{x}{\ctxtapp{\ctxt{L}}{v}}{r} \loredv
  \ctxtapp{\ctxt{L}}{\substitute{x}{v}{r}} = t'$ is an $\eStep$-step. We
  proceed by induction on $\ctxt{L}$. We only show here the case $\ctxt{L} =
  \Box$ as the inductive case is straightforward.

  Then, $t' = \substitute{x}{v}{r}$. Lem.~\ref{l:value:anti-substitution-tight}
  with $\Phi'$, there exist type derivations
  $\derivable{\Phi_{r}}{\sequT{\Gamma'}{\assign{r}{\sigma}}{\cmult_1}{\cexp_1}{\csize_1}}{\SysTightCBV}$
  and
  $\derivable{\Phi_{v}}{\sequT{\Delta}{\assign{v}{\Gamma'(x)}}{\cmult_2}{\cexp_2}{\csize_2}}{\SysTightCBV}$
  such that $\Gamma = \ctxtsum{\ctxtres{\Gamma'}{x}{}}{\Delta}{}$, $\cmult' =
  \cmult_1 + \cmult_2$, $\cexp' = \cexp_1 + \cexp_2 - 1$ and $\csize =
  \csize_1 + \csize_2$. We conclude by $\ruleVDESubs$ with the following type
  derivation \[
\Rule{
  \derivable{\Phi_{r}}{\sequT{\Gamma'}{\assign{r}{\sigma}}{\cmult_1}{\cexp_1}{\csize_1}}{\SysTightCBV}
  \quad
  \derivable{\Phi_{v}}{\sequT{\Delta}{\assign{v}{\Gamma'(x)}}{\cmult_2}{\cexp_2}{\csize_2}}{\SysTightCBV}
}{
  \derivable{\Phi}{\sequT{\ctxtsum{\ctxtres{\Gamma'}{x}{}}{\Delta}{}}{\assign{\termsubs{x}{v}{r}}{\sigma}}{\cmult_1+\cmult_2}{\cexp_1+\cexp_2}{\csize_1+\csize_2}}{\SysTightCBV}
}{\ruleVDESubs}
  \] taking $\cmult = \cmult_1 + \cmult_2 = \cmult'$ and $\cexp =
  \cexp_1 + \cexp_2 = \cexp' + 1$.
  
  \item $t = \termapp{r}{u} \loredv \termapp{r'}{u} = t'$, where $r \loredv r'$
  and $\neg\pabs{r}$. Then, $\Phi'$ ends with rule $\ruleVTArrowE$,
  $\ruleVDArrowE$ or $\ruleVDApp$, and in either case the subterm $r'$ has an
  associated typing derivation $\Phi_{r'}$ which verifies the hypothesis since
  $\neg\pabs{r}$. Hence, by \ih there exists a $\Phi_{r}$ with proper counters
  such that we can derive $\Phi$ applying the same rule as in $\Phi'$ to
  conclude.
  
  \item $t = \termapp{r}{u} \loredv \termapp{r}{u'} = t'$, where $u \loredv u'$
  and $r \in \HCBVNF \cup \VarHCBVNF$. There are two further cases to consider:
  \begin{enumerate}
    \item If $\Phi'$ ends with rule $\ruleVTArrowE$, then the subterm $u'$ has
    an associated typing derivation which verifies the hypothesis (context and
    type are tight). Then the property holds by \ih
    
    \item If $\Phi'$ ends with rule $\ruleVDArrowE$ or $\ruleVDApp$, then the
    subterm $r$ has an associated typing derivation with a tight context and a
    functional type $\multiset{\functtype{\M}{\sigma}}$.
    Lem.~\ref{l:value:tight-spreading} (\ref{l:value:tight-spreading:var})
    gives $r \notin \VarHCBVNF$ so that $r \in \HCBVNF$. Then,
    Lem.~\ref{l:value:tight-spreading} (\ref{l:value:tight-spreading:tight})
    gives $\multiset{\functtype{\M}{\sigma}} \in \typetight$ which leads to a
    contradiction. So this case is not possible.
  \end{enumerate}
  
  \item $t = \termsubs{x}{u}{r} \loredv \termsubs{x}{u}{r'} = t'$, where $r
  \loredv r'$ and $u \in \HCBVNF$. There are two further cases to consider:
  \begin{enumerate}
    \item If $\Phi'$ ends with rule $\ruleVTESubs$ then the subterm $r'$ has an
    associated typing derivation which verifies the hypothesis (context tight
    and $\neg\pabs{r}$). Then the property holds by \ih

    \item If $\Phi'$ ends with rule $\ruleVDESubs$ then the subterm $u$ has an
    associated typing derivation with a tight context and a multiset type $\M$.
    Lem.~\ref{l:value:tight-spreading} (\ref{l:value:tight-spreading:tight})
    gives $\M \in \typetight$ which leads to a contradiction. So this case is
    not possible.
  \end{enumerate}
  
  \item $t = \termsubs{x}{u}{r} \loredv \termsubs{x}{u'}{r} = t'$, where $u
  \loredv u'$ and $\neg\pval{u}$. There are two further cases to consider:
  \begin{enumerate}
    \item  If $\Phi'$ ends with rule $\ruleVTESubs$ then the subterm $u'$ has
    an associated typing derivation which verifies the hypothesis (context and
    type are tight). Then the property holds by the \ih.
    
    \item If $\Phi'$ ends with rule $\ruleVDESubs$ then the subterm
    $u'$ has an associated typing derivation $\Phi_{u'}$ which verifies the
    hypothesis since $\neg\pval{u}$ implies $\neg\pabs{u}$. Hence, by \ih there
    exists a $\Phi_{u}$ with proper counters such that we can derive $\Phi$
    applying $\ruleVDESubs$ as in $\Phi'$ to conclude.
  \end{enumerate}
\end{itemize}
\end{proof}


Notice that tight derivations properly counts, separately, $\mStep$-steps and
$\eStep$-steps. As a consequence, completeness follows.

\begin{theorem}[Completeness]
If $t \rewriten{\callbyvalue}^{(\cmult,\cexp)} p$ with $p \in \CBVNF$, then
there exists a tight type derivation
$\derivable{\Phi}{\sequT{\Gamma}{\assign{t}{\sigma}}{\cmult}{\cexp}{\valsize{p}}}{\SysTightCBV}$.
\label{t:value:completeness-tight}
\end{theorem}

\begin{proof}
As for soundness (Thm.~\ref{t:value:correctness-tight}) prove the
statement for $\loredv$ and then conclude for the general notion of reduction
$\rewrite{\callbyvalue}$. Let $t \loredvn^{(\cmult,\cexp)} p$. We proceed by
induction on $\cmult+\cexp$:
\begin{itemize}
  \item If $\cmult+\cexp = 0$, then $\cmult = \cexp = 0$ and thus $t = p$,
  which implies $t \in \CBVNF$. Lem.~\ref{l:value:normal-forms-tight} allows
  to conclude.

  \item If $\cmult+\cexp > 0$, then there is $t'$ such that $t \loredvn^{(1,0)}
  t' \loredvn^{(\cmult-1,\cexp)} p$ or $t \loredvn^{(0,1)} t'
  \loredvn^{(\cmult,\cexp-1)} p$. By the \ih there is a tight derivation
  $\derivable{\Phi'}{\sequT{\Gamma}{\assign{t'}{\sigma}}{\cmult'}{\cexp'}{\valsize{p}}}{\SysTightCBV}$
  such $\cmult' + \cexp' = \cmult + \cexp -1$.
  Lem.~\ref{l:value:subject-expansion-tight} gives a tight derivation
  $\derivable{\Phi}{\sequT{\Gamma}{\assign{t}{\sigma}}{\cmult''}{\cexp''}{\valsize{p}}}{\SysTightCBV}$
  such $\cmult'' + \cexp'' = \cmult' + \cexp' + 1$. We then have $\cmult'' +
  \cexp'' = \cmult + \cexp$. The fact that $\cmult'' = \cmult$ and $\cexp'' =
  \cexp$ holds by a simple case analysis.
\end{itemize}
\end{proof}


\begin{example}
\label{ex:t0cbv}
Consider $t_0 =
\termapp{\termapp{\Kterm}{(\termapp{z}{\id})}}{(\termapp{\id}{\id})}$ from
Ex.~\ref{ex:t0cbn} but now in the CBV setting. It
$\rewrite{\callbyvalue}$-reduces in 3 $\mStep$-steps and 2 $\eStep$-steps to
$\termsubs{x}{\termapp{z}{\id}}{x} \in \CBVNF$, whose $\callbyvalue$-size is 1,
as follows: \[
\begin{array}{lllllllll}
t_0 & =                 & \termapp{\underline{\termapp{\Kterm}{(\termapp{z}{\id})}}}{(\termapp{\id}{\id})}
    & \rewrite{\dBeta}  & \underline{\termapp{\termsubs{x}{\termapp{z}{\id}}{(\termabs{y}{x})}}{(\termapp{\id}{\id})}}
    & \rewrite{\dBeta}  & \termsubs{x}{\termapp{z}{\id}}{\termsubs{y}{\underline{\termapp{\id}{\id}}}{x}} \\
    & \rewrite{\dBeta}  & \termsubs{x}{\termapp{z}{\id}}{\underline{\termsubs{y}{\termsubs{w}{\id}{w}}{x}}}
    & \rewrite{\sVal}   & \termsubs{x}{\termapp{z}{\id}}{\underline{\termsubs{w}{\id}{x}}}
    & \rewrite{\sVal}   & \termsubs{x}{\termapp{z}{\id}}{x}
\end{array}
\] System $\SysTightCBV$ admits a proper tight type derivation for $t_0 =
\termapp{\termapp{\Kterm}{(\termapp{z}{\id})}}{(\termapp{\id}{\id})}$ from
Ex.~\ref{ex:t0cbv}, with the expected final counter $(3,2,1)$, as shown in
Fig.~\ref{f:t0cbv}.
\end{example}

\begin{sidewaysfigure}
\input{example-tight-derivation-cbv.tex}
\caption{Tight type derivation for $t_0$ in System $\SysTightCBV$.}
\label{f:t0cbv}
\end{sidewaysfigure}


\section{The \BangRev-Calculus}
\label{s:bang}

This section briefly presents the \emphdef{bang calculus at a
distance}~\cite{BucciarelliKRV20}, called $\BangRev$-calculus. It is a
(conservative) extension of the original bang
calculus~\cite{Ehrhard16,EhrhardG16}, it uses ES operators and
\emph{reduction at a distance}~\cite{AccattoliK10}, thus integrating
commutative conversions without jeopardising confluence
(see~\cite{BucciarelliKRV20} for a discussion). Indeed,
T.~Ehrhard~\cite{Ehrhard16} studies the CBPV from a Linear Logic (LL) point of
view by extending the $\lambda$-calculus with two new unary constructors
\emph{bang} ($\termbang{\!}$) and \emph{dereliction} ($\termder{\!}$), playing
the role of the CBPV primitives $\mathtt{thunk}$/$\mathtt{force}$ respectively.
His calculus suffers from the absence of \emph{commutative
conversions}~\cite{Regnier94,CarraroG14}, making some redexes to be
syntactically blocked when open terms are considered. As a consequence, some
normal forms are semantically equivalent to non-terminating programs, a
situation which is clearly unsound. The bang calculus~\cite{EhrhardG16} adds
commutative conversions specified by means of $\sigma$-reduction rules, which
are crucial to unveil hidden (blocked) redexes. This approach, however,
presents a major drawback since the resulting combined reduction relation is
not confluent. The $\BangRev$-calculus~\cite{BucciarelliKRV20} fixes these two
problems at the same time. Indeed, the syntax of the bang calculus is enriched
with explicit substitutions (ES), and $\sigma$-equivalence is integrated in
the primary reduction system by using the \emph{distance}
paradigm~\cite{AccattoliK10}, without any need to unveil hidden redexes by
means of an independent relation.

We consider the following grammar for terms (denoted by $\TermExplicit$) and
contexts:
\begin{center} 
\begin{tabular}{rrcll}
\emphdef{(Terms)}             & $t,u$       & $\Coloneq$  & $x \in \TermVariable \mid \termapp{t}{u} \mid \termabs{x}{t} \mid \termbang{t} \mid \termder{t} \mid \termsubs{x}{u}{t}$ \\
\emphdef{(List contexts)}     & $\ctxt{L}$  & $\Coloneq$  & $\Box \mid \termsubs{x}{t}{\ctxt{L}}$ \\ 
\emphdef{(Surface contexts)}  & $\ctxt{S}$  & $\Coloneq$  & $\Box \mid \termapp{\ctxt{S}}{t} \mid \termapp{t}{\ctxt{S}} \mid \termabs{x}{\ctxt{S}} \mid \termder{\ctxt{S}} \mid \termsubs{x}{u}{\ctxt{S}} \mid \termsubs{x}{\ctxt{S}}{t}$
\end{tabular}
\end{center}
Special terms are $\Delta_{\termbang{}} =
\termabs{x}{\termapp{x}{\termbang{x}}}$, and $\Omega_{\termbang{}} =
\termapp{\Delta_{\termbang{}}}{\termbang{\Delta_{\termbang{}}}}$. Surface
contexts do not allow the symbol $\Box$ to occur inside the bang constructor
$\termbang{}$. This is similar to \emph{weak} contexts in $\lambda$-calculus,
where $\Box$ cannot occur inside $\lambda$-abstractions. As we will see in
Sec.~\ref{s:cbn-cbv-embeddings}, surface reduction in the $\BangRev$-calculus
is perfectly sufficient to capture head reduction in CBN, disallowing reduction
inside arguments, as well as open CBV, disallowing reduction inside
abstractions. Finally, we define the \emphdef{$\bangweak$-size} of terms as
follows:
\begin{center}
$
\begin{array}{lll@{\hspace{1cm}}lll@{\hspace{1cm}}lll}
\wsize{x}                   & \coloneq & 0 &
\wsize{\termder{t}}         & \coloneq & \wsize{t} &
\wsize{\termsubs{x}{u}{t}}  & \coloneq & \wsize{t} + \wsize{u} \\
\wsize{\termbang{t}}        & \coloneq & 0 &
\wsize{\termabs{x}{t}}      & \coloneq & 1 + \wsize{t} & 
\wsize{\termapp{t}{u}}      & \coloneq & 1 + \wsize{t} + \wsize{u}
\end{array}
$
\end{center} 
The \emphdef{$\BangRev$-calculus} is given by the set of terms $\TermExplicit$
and the \emphdef{(surface) reduction relation} $\rewrite{\bangweak}$, which is
defined as the \emph{union} of $\rewrite{\dBeta}$, $\rewrite{\sBang}$
(\ttbf{s}ubstitute bang) and $\rewrite{\dBang}$ (\ttbf{d}istant bang),
defined respectively as the closure by contexts $\ctxt{S}$ of the following
three rewriting rules:
\begin{center} $\begin{array}{lll}
\termapp{\ctxtapp{\ctxt{L}}{\termabs{x}{t}}}{u}    & \rrule{\dBeta} & \ctxtapp{\ctxt{L}}{\termsubs{x}{u}{t}} \\
\termsubs{x}{\ctxtapp{\ctxt{L}}{\termbang{u}}}{t}  & \rrule{\sBang} & \ctxtapp{\ctxt{L}}{\substitute{x}{u}{t}} \\
\termder{(\ctxtapp{\ctxt{L}}{\termbang{t}})}       & \rrule{\dBang} & \ctxtapp{\ctxt{L}}{t}
\end{array}$
\end{center}
The rules are defined \emph{at a distance}, as in CBN/CBV, in the sense that
the list context $\ctxt{L}$ allows the main constructors involved in the
rules to be separated by an arbitrary finite list of substitutions. This new
formulation integrates commutative conversions inside the main (logical)
reduction rules of the calculus, thus inheriting  the benefits enumerated in
Sec.~\ref{s:intro}. Indeed, rule $\sBang$ can be decomposed in two different
rules $\termsubs{x}{\termbang{u}}{t} \rrule{} \substitute{x}{u}{t}$ and
$\termsubs{x}{\ctxtapp{\ctxt{L}}{\termbang{u}}}{t} \rrule{}
\ctxtapp{\ctxt{L}}{\termsubs{x}{\termbang{u}}{t}}$, while $\dBang$ can be
decomposed in $\termder{(\termbang{t})} \rrule{} t$ and
$\termder{(\ctxtapp{\ctxt{L}}{\termbang{t}})} \rrule{}
\ctxtapp{\ctxt{L}}{\termder{\termbang{t}}}$. We write $\rewriten{\bangweak}$
for the reflexive-transitive closure of $\rewrite{\bangweak}$. Given the
translation of the bang calculus into LL proof-nets~\cite{Ehrhard16}, we
refer to $\dBeta$-steps as \emphdef{$\mStep$ultiplicative} and
$\sBang$/$\dBang$-steps as \emphdef{$\eStep$xponential} steps. We write $t
\rewriten{\bangweak}^{(\cmult,\cexp)} u$ if $t \rewriten{\bangweak} u$ using
$\cmult$ $\mStep$-steps and $\cexp$ $\eStep$-steps.

\begin{example}
\label{ex:t0pcbn}
Consider the following reduction sequence from $t'_0 =
\termapp{\termapp{\Kterm}{(\termbang{(\termapp{z}{\termbang{\id}})})}}{(\termbang{(\termapp{\id}{\termbang{\id}})})}$: \[
\begin{array}{lllll}
t'_0 = \termapp{\underline{\termapp{\Kterm}{(\termbang{(\termapp{z}{\termbang{\id}})})}}}{(\termbang{(\termapp{\id}{\termbang{\id}})})}
  & \rewrite{\dBeta} & \underline{\termapp{\termsubs{x}{\termbang{(\termapp{z}{\termbang{\id}})}}{(\termabs{y}{x})}}{(\termbang{(\termapp{\id}{\termbang{\id}})})}}
  & \rewrite{\dBeta} & \underline{\termsubs{x}{\termbang{(\termapp{z}{\termbang{\id}})}}{\termsubs{y}{\termbang{(\termapp{\id}{\termbang{\id}})}}{x}}} \\
  & \rewrite{\sBang} & \underline{\termsubs{y}{\termbang{(\termapp{\id}{\termbang{\id}})}}{(\termapp{z}{\termbang{\id}})}}
  & \rewrite{\sBang} & \termapp{z}{\termbang{\id}}
\end{array}
\] Notice that the second $\dBeta$-step uses action at a distance, where
$\ctxt{L}$ is $\termsubs{x}{\termbang{(\termapp{z}{\termbang{\id}})}}{\Box}$.
\end{example}

The relation $\rewrite{\bangweak}$ enjoys a weak diamond property,
\ie one-step divergence can be closed in one step if the diverging terms are
different. This property has two important consequences.

\begin{theorem}[Confluence~\cite{BucciarelliKRV20}]
\label{t:confluence}
The reduction relation $\rewrite{\bangweak}$ is confluent.
Moreover, any two different $\rewrite{\bangweak}$-reduction paths to normal form have the same length.
\end{theorem}

The second point relies essentially on the fact that reductions are
disallowed under bangs. An important consequence is that we can focus on any
particular \emph{deterministic} strategy for the $\BangRev$-calculus, without
changing the number of steps to $\bangweak$-normal form.

\parrafo{Neutral, Normal, and Clash-Free Terms}
A term is said to be \emphdef{$\bangweak$-normal} if there is no $t'$ such that
$t \rewrite{\bangweak} t'$, in which case we write $t \not\rewrite{\bangweak}$.
However, some ill-formed $\bangweak$-normal terms are not still the ones that
represent a desired result for a computation, they are called \emphdef{clashes}
(meta-variable $c$), and take one of the following forms:
$\termapp{\ctxtapp{\ctxt{L}}{\termbang{t}}}{u}$,
$\termsubs{y}{\ctxtapp{\ctxt{L}}{\termabs{x}{u}}}{t}$, 
$\termder{(\ctxtapp{\ctxt{L}}{\termabs{x}{u}})}$, or
$\termapp{t}{({\ctxtapp{\ctxt{L}}{\termabs{x}{u}}})}$. Remark that in the three
first kind of clashes, replacing $\termabs{x}{\!}$ by $\termbang{\!}$, and
inversely, creates a (root) redex, namely
$\termapp{(\ctxtapp{\ctxt{L}}{\termabs{x}{t}})}{u}$,
$\termsubs{x}{\ctxtapp{\ctxt{L}}{\termbang{t}}}{t}$ and 
$\termder{(\ctxtapp{\ctxt{L}}{\termbang{t}})}$, respectively.

A term is \emphdef{clash free} if it does not reduce to a term containing a
clash, it is \emphdef{surface clash free}, written $\cfz$, if it does not
reduce to a term containing a clash outside the scope of any constructor $!$.
Thus,  $t$ is not $\cfz$ if and only if there exist a surface context
$\ctxt{S}$ and a clash $c$ such that $t \rewriten{\bangweak}
\ctxtapp{\ctxt{S}}{c}$. Surface clash free normal terms can be characterised
as follows:
\begin{center}
\begin{tabular}{rrcll}
\emphdef{(Neutral $\cfz$)}      & $\icfntrl$  & $\Coloneq$ & $x \in \TermVariable \mid \termapp{\icfntrl}{\icfnabs} \mid \termder{(\icfntrl)} \mid \termsubs{x}{\icfntrl}{\icfntrl}$ \\
\emphdef{(Neutral-Abs $\cfz$)}  & $\icfnabs$  & $\Coloneq$ & $\termbang{t} \mid \icfntrl \mid \termsubs{x}{\icfntrl}{\icfnabs}$ \\ 
\emphdef{(Neutral-Bang $\cfz$)} & $\icfnbang$ & $\Coloneq$ & $\icfntrl \mid \termabs{x}{\icfnrml} \mid \termsubs{x}{\icfntrl}{\icfnbang}$ \\
\emphdef{(Normal $\cfz$)}       & $\icfnrml$  & $\Coloneq$ & $\icfnabs \mid \icfnbang$
\end{tabular}
\end{center}

\begin{proposition}[Clash-Free Normal Terms~\cite{BucciarelliKRV20}]
\label{p:clashfree}
Let $t \in \TermExplicit$. Then $t$ is a surface clash free $\bangweak$-normal
term iff $\cfnrml{t}$.
\end{proposition}


\section{A Tight Type System for the \BangRev-Calculus}
\label{s:tight}

The methodology used to define the type system $\SysTight$ for the
$\BangRev$-calculus is based on~\cite{BucciarelliKRV20}, inspired in turn
from~\cite{Carvalho:thesis,BernadetL13,AccattoliGK18}, which defines
non-idempotent intersection type systems to count reduction lengths for
different evaluation strategies in the $\lambda$-calculus. In the case of the
$\BangRev$-calculus, however, Thm.~\ref{t:confluence}
guarantees that all reduction  paths to normal form have the same length, so
that it is not necessary to reason w.r.t. any particular evaluation strategy.

The grammar of types of system $\SysTight$ is given by:
\begin{center}
\begin{tabular}{rrcll}
\emphdef{(Tight Types)} & $\typetight$   & $\Coloneq$ & $\typeneutral \mid \typeabs \mid \typevalue$ \\
\emphdef{(Types)}       & $\sigma, \tau$ & $\Coloneq$ & $\typetight \mid \M \mid \functtype{\M}{\sigma}$ \\
\emphdef{(Multitype)}   & $\M, \N$       & $\Coloneq$ & $\intertype{\sigma_i}{i \in I}$  where $I$ is a finite set
\end{tabular}
\end{center}
The constant $\typeabs$ (resp. $\typevalue$) types terms whose normal form has
the shape $\ctxtapp{\ctxt{L}}{\termabs{x}{t}}$ (resp.
$\ctxtapp{\ctxt{L}}{\termbang{t}}$, \ie \emph{values} in the
$\BangRev$-calculus sense), and the constant $\typeneutral$ types terms whose
normal form is a neutral $\cfz$. As before, the notions of \emphdef{tightness}
for multitypes, contexts and derivations are exactly the same used for CBN.
Typing rules are split in two groups: the \emph{persistent} rules
(Fig.~\ref{fig:typingSchemesTightPersistent}) and the \emph{consuming} ones
(Fig.~\ref{fig:typingSchemesTightConsuming}).

\begin{figure}[ht]
\centering $
\kern-2em
\begin{array}{c}
\Rule{\sequT{\Gamma}{\assign{t}{\typeneutral}}{\cmult}{\cexp}{\csize}
      \quad
      \sequT{\Delta}{\assign{u}{\overline{\typeabs}}}{\cmult'}{\cexp'}{\csize'}
     }
     {\sequT{\ctxtsum{\Gamma}{\Delta}{}}{\assign{\termapp{t}{u}}{\typeneutral}}{\cmult+\cmult'}{\cexp+\cexp'}{\csize+\csize'+1}}
     {\ruleBTArrowE}
\qquad
\Rule{\sequT{\Gamma}{\assign{t}{\typetight}}{\cmult}{\cexp}{\csize}
      \quad
      \ptight{\Gamma(x)}
     }
     {\sequT{\ctxtres{\Gamma}{x}{}}{\assign{\termabs{x}{t}}{\typeabs}}{\cmult}{\cexp}{\csize+1}}
     {\ruleBTArrowI}
\\
\\
\Rule{\vphantom{\sequT{}{\assign{\termbang{t}}{\typevalue}}{0}{0}{0}}}
     {\sequT{}{\assign{\termbang{t}}{\typevalue}}{0}{0}{0}}
     {\ruleBTBang}
\quad
\Rule{\sequT{\Gamma}{\assign{t}{\typeneutral}}{\cmult}{\cexp}{\csize}}
     {\sequT{\Gamma}{\assign{\termder{t}}{\typeneutral}}{\cmult}{\cexp}{\csize}}
     {\ruleBTDer}
\quad
\Rule{\sequT{\Gamma}{\assign{t}{\tau}}{\cmult}{\cexp}{\csize}
      \enspace
      \sequT{\Delta}{\assign{u}{\typeneutral}}{\cmult'}{\cexp'}{\csize'}
      \enspace
      \ptight{\Gamma(x)}
     }
     {\sequT{\ctxtsum{(\ctxtres{\Gamma}{x}{})}{\Delta}{}}{\assign{\termsubs{x}{u}{t}}{\tau}}{\cmult+\cmult'}{\cexp+\cexp'}{\csize+\csize'}}
     {\ruleBTESubs}
\end{array}$
\caption{System $\SysTight$ for the $\BangRev$-Calculus: Persistent Typing Rules.}
\label{fig:typingSchemesTightPersistent}
\end{figure}

\begin{figure}[ht]
\centering $
\begin{array}{c}
\Rule{\vphantom{\sequT{}{\assign{\termbang{t}}{\typevalue}}{0}{0}{0}}}
     {\sequT{\assign{x}{\intertype{\sigma}{}}}{\assign{x}{\sigma}}{0}{0}{0}}
     {\ruleBDAxiom}
\qquad
\Rule{\sequT{\Gamma}{\assign{t}{\functtype{\M}{\tau}}}{\cmult}{\cexp}{\csize}
      \quad
      \sequT{\Delta}{\assign{u}{\M}}{\cmult'}{\cexp'}{\csize'}
     }
     {\sequT{\ctxtsum{\Gamma}{\Delta}{}}{\assign{\termapp{t}{u}}{\tau}}{\cmult+\cmult'+1}{\cexp+\cexp'}{\csize+\csize'}}
     {\ruleBDArrowE}
\\
\\
\Rule{\sequT{\Gamma}{\assign{t}{\functtype{\M}{\tau}}}{\cmult}{\cexp}{\csize}
      \quad
      \sequT{\Delta}{\assign{u}{\typeneutral}}{\cmult'}{\cexp'}{\csize'}
      \quad
      \ptight{\M}
     }
     {\sequT{\ctxtsum{\Gamma}{\Delta}{}}{\assign{\termapp{t}{u}}{\tau}}{\cmult+\cmult'+1}{\cexp+\cexp'}{\csize+\csize'}}
     {\ruleBDApp}
\\
\\
\Rule{\sequT{\Gamma}{\assign{t}{\tau}}{\cmult}{\cexp}{\csize}}
     {\sequT{\ctxtres{\Gamma}{x}{}}{\assign{\termabs{x}{t}}{\functtype{\Gamma(x)}{\tau}}}{\cmult}{\cexp}{\csize}}
     {\ruleBDArrowI}
\qquad
\Rule{(\sequT{\Gamma_i}{\assign{t}{\sigma_i}}{\cmult_i}{\cexp_i}{\csize_i})_{i \in I}}
     {\sequT{\ctxtsum{}{\Gamma_i}{i \in I}}{\assign{\termbang{t}}{\intertype{\sigma_i}{\iI}}}{+_{\iI}{\cmult_i}}{1+_{\iI}{\cexp_i}}{+_{\iI}{\csize_i}}}
     {\ruleBDBang}
\\
\\
\Rule{\sequT{\Gamma}{\assign{t}{\intertype{\sigma}{}}}{\cmult}{\cexp}{\csize}}
     {\sequT{\Gamma}{\assign{\termder{t}}{\sigma}}{\cmult}{\cexp}{\csize}}
     {\ruleBDDer}
\qquad
\Rule{\sequT{\Gamma}{\assign{t}{\sigma}}{\cmult}{\cexp}{\csize}
      \quad
      \sequT{\Delta}{\assign{u}{\Gamma(x)}}{\cmult'}{\cexp'}{\csize'}
     }
     {\sequT{\ctxtsum{(\ctxtres{\Gamma}{x}{})}{\Delta}{}}{\assign{\termsubs{x}{u}{t}}{\sigma}}{\cmult+\cmult'}{\cexp+\cexp'}{\csize+\csize'}}
     {\ruleBDESubs}
\end{array} $
\caption{System $\SysTight$ for the $\BangRev$-Calculus: Consuming Typing Rules.}
\label{fig:typingSchemesTightConsuming}
\end{figure}

As in CBV, the $\sBang$ exponential steps do not only depend on a (consuming)
substitution constructor, but on the (bang) form of its argument. This makes
the exponential counting more subtle. Notice also that rule $\ruleBTDer$ does
not count $\termder{\!}$ constructors, according to the definition of
$\wsize{\_}$ given in Sec.~\ref{s:bang} and in contrast
to~\cite{BucciarelliKRV20}. This is to keep a more intuitive relation with the
CBN/CBV translations (Sec.~\ref{s:cbname-cbvalue}), where $\termder{\!}$ plays
a silent role.

As in~\cite{BucciarelliKRV20}, system $\SysTight$ is quantitatively sound and
complete. More precisely,

\begin{theorem}[Soundness and Completeness]\mbox{}
\begin{enumerate}
  \item \label{t:correctness} If
  $\derivable{\Phi}{\sequT{\Gamma}{\assign{t}{\sigma}}{\cmult}{\cexp}{\csize}}{\SysTight}$
  is tight, then there exists $p$ such that $\cfnrml{p}$ and
  $t \rewriten{\bangweak}^{(\cmult,\cexp)} p$ with $\cmult$ $\mStep$-steps,
  $\cexp$ $\eStep$-steps, and $\wsize{p} = \csize$.

  \item \label{t:completeness-tight} If
  $t \rewriten{\bangweak}^{(\cmult,\cexp)} p$ with $\cfnrml{p}$, then there
  exists a tight type derivation
  $\derivable{\Phi}{\sequT{\Gamma}{\assign{t}{\sigma}}{\cmult}{\cexp}{\wsize{p}}}{\SysTight}$.
\end{enumerate}
\label{t:soundness-completeness-bang}
\end{theorem}

\ifreport{
\begin{proof}
Similar to \cite{BucciarelliKRV20}.
\end{proof}    
}

\begin{example}
\label{ex:traducciones}
Consider $t'_0
= \termapp{\termapp{\Kterm}{(\termbang{(\termapp{z}{\termbang{\id}})})}}{(\termbang{(\termapp{\id}{\termbang{\id}})})}$
from Ex.~\ref{ex:t0pcbn}, which normalises in $2$ $\mStep$-steps and $2$
$\eStep$-steps to $\termapp{z}{\termbang{\id}} \in \icfnrml$ of
$\bangweak$-size $1$. A tight derivation for $t'_0$ with appropriate
final counters $(2,2,1)$ is given below.
{\small
\[
\Rule{
  \Rule{
    \Rule{
      \Rule{
        \Rule{}{
          \sequT{\assign{x}{\multiset{\typeneutral}}}{\assign{x}{\typeneutral}}{0}{0}{0}
        }{\ruleBDAxiom}
      }{
        \sequT{\assign{x}{\multiset{\typeneutral}}}{\assign{\termabs{y}{x}}{\functtype{\emul}{\typeneutral}}}{0}{0}{0}
      }{\ruleBDArrowI}
    }{
      \sequT{}{\assign{\Kterm}{\functtype{\multiset{\typeneutral}}{\functtype{\emul}{\typeneutral}}}}{0}{0}{0}
    }{\ruleBDArrowI}
    \Rule{
      \Rule{
        \Rule{}{
          \sequT{\assign{z}{\multiset{\typeneutral}}}{\assign{z}{\typeneutral}}{0}{0}{0}
        }{\ruleBDAxiom}
        \Rule{}{
          \sequT{}{\assign{\termbang{\id}}{\typevalue}}{0}{0}{0}
        }{\ruleBTBang}
      }{
        \sequT{\assign{z}{\multiset{\typeneutral}}}{\assign{\termapp{z}{\termbang{\id}}}{\typeneutral}}{0}{0}{1}
      }{\ruleBTArrowE}
    }{
      \sequT{\assign{z}{\multiset{\typeneutral}}}{\assign{\termbang{(\termapp{z}{\termbang{\id}})}}{\multiset{\typeneutral}}}{0}{1}{1}
    }{\ruleBDBang}
  }{
    \sequT{\assign{z}{\multiset{\typeneutral}}}{\assign{\termapp{\Kterm}{(\termbang{(\termapp{z}{\termbang{\id}})})}}{\functtype{\emul}{\typeneutral}}}{1}{1}{1}
  }{\ruleBDArrowE}
  \kern-4em
  \Rule{}{
    \sequT{}{\assign{\termbang{(\termapp{\id}{\termbang{\id}})}}{\emul}}{0}{1}{0}
  }{\ruleBDBang}
}{
  \sequT{\assign{z}{\multiset{\typeneutral}}}{\assign{\termapp{\termapp{\Kterm}{(\termbang{(\termapp{z}{\termbang{\id}})})}}{(\termbang{(\termapp{\id}{\termbang{\id}})})}}{\typeneutral}}{2}{2}{1}
}{\ruleBDArrowE}
\]
}

\noindent Notice that the only persistent rules are $\ruleBTBang$ and
$\ruleBTArrowE$, used to type $\termapp{z}{\termbang{\id}}$. Indeed,
$\termapp{z}{\termbang{\id}}$ is the $\bangweak$-normal form of $t'_0$.
\end{example}


\section{Untyped Translations}
\label{s:cbn-cbv-embeddings}

CBN/CBV (untyped) encodings into the  bang
calculus~\cite{GuerrieriM18}, inspired from Girard's encodings, establish two
translations $cbn$ and $cbv$, such that when $t$ reduces to $u$ in CBN (resp.
CBV), $cbn(t)$ reduces to $cbn(u)$ (resp. $cbv(t)$ reduces to $cbv(u)$) in the
bang calculus.
These two encodings are dual: CBN forbids reduction inside arguments, which are
translated to bang terms, while CBV forbids reduction under
$\lambda$-abstractions, also translated to bang terms.

In this paper we use alternative encodings. For CBN, we slightly adapt to
explicit substitutions Girard's translation into LL~\cite{Girard87}. The
resulting encoding preserves normal forms and is sound and complete with
respect to the standard (quantitative) type system in~\cite{Gardner94}. For
CBV, we discard the original encoding in~\cite{GuerrieriM18} for two reasons:
CBV normal forms are not necessarily translated to normal forms in the bang
calculus (see~\cite{GuerrieriM18}), and levels of terms (the level of $t$ is
the number of $\termbang{\!}$ surrounding $t$) are not preserved either
(see~\cite{FaggianG21}). We thus adopt the CBV encoding
in~\cite{BucciarelliKRV20} which preserves normal forms as well as levels.

The CBN and CBV embedding into the $\BangRev$-calculus, written $\cbnterm{\_}$ and
$\cbvterm{\_}$ resp., are inductively defined as: \[
\begin{array}{c@{\qquad}c}
\begin{array}{rcl}
\cbnterm{x}                    & \eqdef & x \\
\cbnterm{(\termabs{x}{t})}     & \eqdef & \termabs{x}{\cbnterm{t}} \\
\cbnterm{(\termapp{t}{u})}     & \eqdef & \termapp{\cbnterm{t}}{\termbang{\cbnterm{u}}} \\
\cbnterm{(\termsubs{x}{u}{t})} & \eqdef & \termsubs{x}{\termbang{\cbnterm{u}}}{\cbnterm{t}}
\end{array}
&
\begin{array}{rcl}
\cbvterm{x} & \eqdef & \termbang{x} \\
\cbvterm{(\termabs{x}{t})} & \eqdef & \termbang{\termabs{x}{\cbvterm{t}}} \\
\cbvterm{(\termapp{t}{u})} & \eqdef & \left \{
\begin{array}{ll}
  \termapp{\ctxtapp{\ctxt{L}}{s}}{\cbvterm{u}} & \text{if $\cbvterm{t} = \ctxtapp{\ctxt{L}}{\termbang{s}}$} \\
\termapp{\termder{(\cbvterm{t})}}{\cbvterm{u}}   & \text{otherwise}
\end{array}
\right. \\
\cbvterm{(\termsubs{x}{u}{t})} & \eqdef & \termsubs{x}{\cbvterm{u}}{\cbvterm{t}}
\end{array}
\end{array} \]

Both translations extend to list contexts $\ctxt{L}$ as expected. Remark that
there are no two consecutive $\termbang{\!}$ constructors in the image of the
translation. The CBN embedding extends Girard's translation to ES, while the
CBV one is  different. Indeed, the translation of an application
$\termapp{t}{u}$ is usually defined as $\termapp{\termder{(\cbvterm{t})}}{\cbvterm{u}}$
(see \eg\cite{EhrhardG16}). This definition does not preserve normal forms, \ie
$\termapp{x}{y}$ is a $\callbyvalue$-normal form but its translated version
$\termapp{\termder{(\termbang{x})}}{\termbang{y}}$ is not a $\bangweak$-normal
form. We restore this fundamental property by using the well-known notion of
superdevelopment~\cite{BezemKO03}, so that $\dBang$-reductions are applied by
the translation on the fly. Moreover, simulation of CBN/CBV in the
$\BangRev$-calculus also holds.

\begin{lemma}[Simulation~\cite{BucciarelliKRV20}]
\label{l:preservation-and-simulation}
Let $t \in \TermLambda$.
\begin{enumerate}
  \item\label{l:preservation-cbn} $t \not\rewrite{\callbyname}$ implies
  $\cbnterm{t}\not\rewrite{\bangweak}$, and $t \rewrite{\callbyname} s$ implies
  $\cbnterm{t} \rewrite{\bangweak} \cbnterm{s}$.
  \item\label{l:preservation-cbv} $t \not\rewrite{\callbyvalue}$ implies
  $\cbvterm{t}\not\rewrite{\bangweak}$, and $t \rewrite{\callbyvalue} s$ implies
  $\cbvterm{t} \rewriten{\bangweak} \cbvterm{s}$.
\end{enumerate}
\end{lemma}  

\begin{example}
\label{ex:embedding}
Consider again $t_0 =
\termapp{\termapp{\Kterm}{(\termapp{z}{\id})}}{(\termapp{\id}{\id})}$.
To illustrate the CBN case
(Lem.~\ref{l:preservation-and-simulation}:\ref{l:preservation-cbn}), notice
that the reduction sequence $t_0 \rewriten{\cbnname} \termapp{z}{\id} = s_0$
given in Ex.~\ref{ex:t0cbn} is translated to the sequence $\cbnterm{t_0} =
\termapp{\termapp{\Kterm}{(\termbang{(\termapp{z}{\termbang{\id}})})}}{(\termbang{(\termapp{\id}{\termbang{\id}})})}
= t'_0 \rewriten{\bangweak} \termapp{z}{\termbang{\id}} = \cbnterm{s_0}$ given in
Ex.~\ref{ex:t0pcbn}.

To illustrate the CBV case
(Lem.~\ref{l:preservation-and-simulation}:\ref{l:preservation-cbv}),
consider the sequence $t_0 \rewriten{\cbvname}
\termsubs{x}{\termapp{z}{\id}}{x} = s_1$ in Ex.~\ref{ex:t0cbv}. Then for
$\id' = \termabs{w}{\termbang{w}}$ we have: \[
\begin{array}{l@{\enskip}l@{\enskip}l@{\enskip}l@{\enskip}l}
\cbvterm{t_0}
  & =                 & \termapp{\termder{(\underline{\termapp{(\termabs{x}{\termbang{\termabs{y}{\termbang{x}}}})}{(\termapp{z}{\termbang{\id'}})}})}}{(\termapp{\id'}{\termbang{\id'}})}
  & \rewrite{\dBeta}  & \termapp{\underline{\termder{(\termsubs{x}{\termapp{z}{\termbang{\id'}}}{(\termbang{\termabs{y}{\termbang{x}}})})}}}{(\termapp{\id'}{\termbang{\id'}})} \\
  & \rewrite{\dBang}  & \underline{\termapp{\termsubs{x}{\termapp{z}{\termbang{\id'}}}{(\termabs{y}{\termbang{x}})}}{(\termapp{\id'}{\termbang{\id'}})}} 
  & \rewrite{\dBeta}  & \termsubs{x}{\termapp{z}{\termbang{\id'}}}{\termsubs{y}{\underline{\termapp{\id'}{\termbang{\id'}}}}{(\termbang{x})}} \\
  & \rewrite{\dBeta}  & \termsubs{x}{\termapp{z}{\termbang{\id'}}}{\underline{\termsubs{y}{\termsubs{w}{\termbang{\id'}}{(\termbang{w})}}{(\termbang{x})}}} 
  & \rewrite{\sBang}  & \termsubs{x}{\termapp{z}{\termbang{\id'}}}{\underline{\termsubs{w}{\termbang{\id'}}{(\termbang{x})}}} \\
  & \rewrite{\sBang}  & \termsubs{x}{\termapp{z}{\termbang{\id'}}}{(\termbang{x})} 
  & =                 & \cbvterm{s_1}
\end{array}
\]
Notice how this sequence requires extra reduction steps with respect to the one
given in Ex.~\ref{ex:t0cbv}. Indeed, the $\eStep$-step $\rewrite{\dBang}$ in
the $\BangRev$-calculus has no counterpart in CBV.
\end{example}


\section{Typed Translations}
\label{s:tight-translations}

\parrafo{Call-by-Name}
We study the correspondence between derivations in CBN and their encodings in
the $\BangRev$-calculus. First we inject the set of types for $\SysTightCBN$
(generated by the base types $\typeneutral$ and $\typeabs$) into the set of
types of $\SysTight$ (generated also by the base type $\typevalue$) by means of
the function: $\cbntype{\typeneutral} \eqdef \typeneutral$, $\cbntype{\typeabs}
\eqdef\typeabs$, $\cbntype{(\functtype{\M}{\sigma})} \eqdef
\functtype{\cbntype\M}{\cbntype\sigma}$ and
$\cbntype{\intertype{\sigma_i}{i \in I}} \eqdef
\intertype{\cbntype{\sigma_i}}{i \in I}$. Then we translate terms, using the
function $\cbnterm{\_}$ from Sec.~\ref{s:cbn-cbv-embeddings}. Translation of
contexts is defined as expected: $\cbntype{\Gamma} =
\set{\assign{x_i}{\cbntype{\M_i}}}_{i \in I}$. Another notion is needed to
restrict $\SysTight$ derivations to those that come from the translation of
some $\SysTightCBN$ derivation. Indeed, a $\SysTight$ derivation $\Phi$ is
\emphdef{\cbnrelevant} if all the contexts and types involved in $\Phi$ are in
the image of the translation $\cbntype{\_}$. We then obtain:

\begin{theorem}
$\derivable{\Phi}{\sequT{\Gamma}{\assign{t}{\sigma}}{\cmult}{\cexp}{\csize}}{\SysTightCBN}$
if and only if
$\derivable{\Phi'}{\sequT{\cbntype\Gamma}{\assign{\cbnterm{t}}{\cbntype\sigma}}{\cmult}{\cexp}{\csize}}{\SysTight}$
is \cbnrelevant.
\label{t:name:translation}
\end{theorem}

\begin{proof}
The left-to-right implication is by induction on $\Phi$ by analysing the last
rule applied. 
\begin{itemize}
  \item $\ruleNTArrowE$. Then we have $t = \termapp{r}{u}$, $\sigma =
  \typeneutral$ and $\csize = \csize' + 1$, and $\Phi$ is \[
  \Rule{
    \sequT{\Gamma}{\assign{r}{\typeneutral}}{\cmult}{\cexp}{\csize'}
  }{
    \sequT{\Gamma}{\assign{\termapp{r}{u}}{\typeneutral}}{\cmult}{\cexp}{\csize'+1}
  }{\ruleNTArrowE}
  \] By applying the \ih on the premise we may construct $\Phi'$ \[
  \Rule{
    \sequT{\cbntype{\Gamma}}{\assign{\cbnterm{r}}{\cbntype{\typeneutral}}}{\cmult}{\cexp}{\csize'}
    \Rule{}{
      \sequT{}{\assign{\termbang{\cbnterm{u}}}{\cbntype{\typebang}}}{0}{0}{0}
    }{\ruleBTBang}
  }{
    \sequT{\cbntype{\Gamma}}{\assign{ \termapp{\cbnterm{r}}{\termbang{\cbnterm{u}}}}{\cbntype{\typeneutral}}}{\cmult}{\cexp}{\csize'+1}
  }{\ruleBTArrowE} \]

  \item $\ruleNTArrowI$. Then we have $t = \termabs{x}{u}$, $\Gamma =
  \ctxtres{\Gamma_1}{x}{}$, $\sigma = \typeabs$, $\csize = \csize' + 1$, and
  $\Phi$ is \[
  \Rule{
    \sequT{\Gamma_1}{\assign{u}{\typetight}}{\cmult}{\cexp}{\csize'}
    \quad
    \ptight{\Gamma_1(x)}
  }{
    \sequT{\ctxtres{\Gamma_1}{x}{}}{\assign{\termabs{x}{u}}{\typeabs}}{\cmult}{\cexp}{\csize'+1}
  }{\ruleNTArrowI}
  \] By applying the \ih on the premise we may construct $\Phi'$ \[
  \Rule{
    \sequT{\cbntype{\Gamma_1}}{\assign{\cbnterm{u}}{\cbntype{\typetight}}}{\cmult}{\cexp}{\csize'}
    \quad
    \ptight{\cbntype{\Gamma_1}(x)}
  }{
    \sequT{\ctxtres{\cbntype{\Gamma_1}}{x}{}}{\assign{\termabs{x}{\cbnterm{u}}}{\cbntype{\typeabs}}}{\cmult}{\cexp}{\csize'+1}
  }{\ruleBTArrowI} \]
  \item $\ruleNDAxiom$. Then $\Gamma = \assign{x}{\multiset{\sigma}}$, $\cmult
  = \cexp = \csize = 0$, and $\Phi$ is \[
  \Rule{\vphantom{\Gamma}}{
    \sequT{\assign{x}{\multiset{\sigma}}}{\assign{x}{\sigma}}{0}{0}{0}
  }{\ruleNDAxiom}
  \] Then we construct $\Phi'$ \[
  \Rule{\vphantom{\Gamma}}{
    \sequT{\assign{x}{\multiset{\cbntype{\sigma}}}}{\assign{\cbnterm{x}}{\cbntype{\sigma}}}{0}{0}{0}
  }{\ruleBDAxiom} \]

  \item $\ruleNDArrowE$. Then $t = \termapp{r}{u}$, $\Gamma =
  \ctxtsum{\Gamma_1}{\Delta_i}{\iI}$, $\cmult = 1 + \cmult' +_{\iI}{\cmult_i}$,
  $\cexp = 1 + \cexp' +_{\iI}{\cexp_i}$, $\csize = \csize' +_{\iI}{\csize_i}$,
  and $\Phi$ is \[
  \Rule{
    \sequT{\Gamma_1}{\assign{r}{\functtype{\intertype{\tau_i}{i \in I}}{\sigma}}}{\cmult'}{\cexp'}{\csize'}
    \quad
    \many{\sequT{\Delta_i}{\assign{u}{\tau_i}}{\cmult_i}{\cexp_i}{\csize_i}}{i \in I}
  }{
    \sequT{\ctxtsum{\Gamma_1}{\Delta_i}{\iI}}{\assign{\termapp{r}{u}}{\sigma}}{1+\cmult'+_{\iI}{\cmult_i}}{1+\cexp'+_{\iI}{\cexp_i}}{\csize'+_{\iI}{\csize_i}}
  }{\ruleNDArrowE}
  \] By applying the \ih on all the premises we may construct $\Phi'$ \[
  \Rule{
    \sequT{\cbntype{\Gamma_1}}{\assign{\cbnterm{r}}{\functtype{\intertype{\cbntype{\tau_i}}{i \in I}}{\cbntype{\sigma}}}}{\cmult'}{\cexp'}{\csize'}
    \Rule{
      \many{\sequT{\cbntype{\Delta_i}}{\assign{\cbnterm{u}}{\cbntype{\tau_i}}}{\cmult_i}{\cexp_i}{\csize_i}}{i \in I}
    }{
      \sequT{+_{\iI}\cbntype{\Delta_i}}{\assign{\termbang{\cbnterm{u}}}{\multiset{\cbntype{\tau_i}}_{\iI}}}{+_{\iI}\cmult_i}{1+_{\iI}\cexp_i}{+_{\iI}\csize_i}
    }{\ruleBDBang}
  }{
    \sequT{\ctxtsum{\cbntype{\Gamma_1}}{\cbntype{\Delta_i}}{\iI}}{\assign{\termapp{\cbnterm{r}}{\termbang{\cbnterm{u}}}}{\cbntype{\sigma}}}{1+\cmult'+_{\iI}{\cmult_i}}{1+\cexp'+_{\iI}{\cexp_i}}{\csize'+_{\iI}{\csize_i}}
  }{\ruleBDArrowE} \]

  \item $\ruleNDArrowI$. Then $t = \termabs{x}{u}$, $\Gamma = \ctxtres{\Gamma_1}{x}{}$, $\sigma =
  \functtype{\Gamma_1(x)}{\tau}$, and $\Phi$ is \[
  \Rule{
    \sequT{\Gamma_1}{\assign{u}{\tau}}{\cmult}{\cexp}{\csize}
  }{
    \sequT{\ctxtres{\Gamma_1}{x}{}}{\assign{\termabs{x}{u}}{\functtype{\Gamma_1(x)}{\tau}}}{\cmult}{\cexp}{\csize}
  }{\ruleNDArrowI}
  \] By applying the \ih on the premise we may construct $\Phi'$ \[
  \Rule{
    \sequT{\cbntype{\Gamma_1}}{\assign{\cbnterm{u}}{\cbntype{\tau}}}{\cmult}{\cexp}{\csize}
  }{
    \sequT{\ctxtres{\cbntype{\Gamma_1}}{x}{}}{\assign{\termabs{x}{\cbnterm{u}}}{\functtype{\cbntype{\Gamma_1}(x)}{\cbntype{\tau}}}}{\cmult}{\cexp}{\csize}
  }{\ruleBDArrowI} \]
  
  \item $\ruleNDESubs$. Then $t = \termsubs{x}{u}{r}$, $\Gamma =
  \ctxtsum{\ctxtres{\Gamma_1}{x}{}}{\Delta_i}{\iI}$, $\cmult = \cmult'
  +_{\iI}{\cmult_i}$, $\cexp = 1 + \cexp' +_{\iI}{\cexp_i}$, $\csize = \csize'
  +_{\iI}{\csize_i}$, and $\Phi$ is \[
  \Rule{
    \sequT{\Gamma_1;\assign{x}{\intertype{\tau_i}{i \in I}}}{\assign{r}{\sigma}}{\cmult'}{\cexp'}{\csize'}
      \quad
      \many{\sequT{\Delta_i}{\assign{u}{\tau_i}}{\cmult_i}{\cexp_i}{\csize_i}}{i \in I}
  }{
    \sequT{\ctxtsum{(\ctxtres{\Gamma_1}{x}{})}{\Delta_i}{i \in I}}{\assign{\termsubs{x}{u}{r}}{\sigma}}{\cmult'+_{\iI}{\cmult_i}}{1+\cexp'+_{\iI}{\cexp_i}}{\csize'+_{\iI}{\csize_i}}
  }{\ruleNDESubs}
  \] By applying the \ih on all the premises we may construct $\Phi'$ \[
  \Rule{
    \sequT{\cbntype{\Gamma_1};\assign{x}{\intertype{\cbntype{\tau_i}}{i \in I}}}{\assign{\cbnterm{r}}{\cbntype{\sigma}}}{\cmult'}{\cexp'}{\csize'}
    \Rule{
      \many{\sequT{\cbntype{\Delta_i}}{\assign{\cbnterm{u}}{\cbntype{\tau_i}}}{\cmult_i}{\cexp_i}{\csize_i}}{i \in I}
    }{
      \sequT{+_{\iI}\cbntype{\Delta_i}}{\assign{\termbang{\cbnterm{u}}}{\multiset{\cbntype{\tau_i}}_{\iI}}}{+_{\iI}\cmult_i}{1+_{\iI}\cexp_i}{+_{\iI}\csize_i}
    }{\ruleBDBang}
  }{
    \sequT{\ctxtsum{(\ctxtres{\cbntype{\Gamma_1}}{x}{})}{\cbntype{\Delta_i}}{i \in I}}{\assign{\termsubs{x}{\termbang{\cbnterm{u}}}{\cbnterm{r}}}{\cbntype{\sigma}}}{\cmult'+_{\iI}{\cmult_i}}{1+\cexp'+_{\iI}{\cexp_i}}{\csize'+_{\iI}{\csize_i}}
  }{\ruleBDESubs} \]
\end{itemize}

In all cases, the resulting derivation is \cbnrelevant\ by the \ih\ and the
fact that the $\ruleBDESubs$ rule (in the last case) is used as required.

\medskip
The right-to-left implication is by induction on $\Phi'$ by analysing the last
rule applied.
\begin{itemize}
  \item $\ruleBTArrowE$. Then $t = \termapp{t_1}{u}$ with $\cbnterm{t} =
  \termapp{\cbnterm{t_1}}{\termbang{\cbnterm{u}}}$, $\Gamma =
  \ctxtsum{\Gamma_1}{\Gamma_2}{}$, $\cbntype{\sigma} = \typeneutral =
  \sigma$ and $\Phi'$ is of the form \[
\Rule{
  \derivable{\Phi'_{t_1}}{\sequT{\cbntype{\Gamma_1}}{\assign{\cbnterm{t_1}}{\typeneutral}}{\cmult_1}{\cexp_1}{\csize_1}}{\SysTight}
  \quad
  \derivable{\Phi'_{u}}{\sequT{\cbntype{\Gamma_2}}{\assign{\termbang{\cbnterm{u}}}{\overline{\typeabs}}}{\cmult_2}{\cexp_2}{\csize_2}}{\SysTight}
}{
  \sequT{\ctxtsum{\cbntype{\Gamma_1}}{\cbntype{\Gamma_2}}{}}{\assign{\cbnterm{t}}{\typeneutral}}{\cmult_1+\cmult_2}{\cexp_1+\cexp_2}{\csize_1+\csize_2+1}
}{\ruleBTArrowE}
  \]
  where $\cmult = \cmult_1 + \cmult_2$, $\cexp = \cexp_1 + \cexp_2$ and $\csize
  = \csize_1 + \csize_2 + 1$. Moreover, since $\cbnterm{u}$ is typed with
  $\overline{\typeabs}$, then the only possible case is $\typevalue$, \ie
  $\cbntype{\Gamma_2} = \emptyset$ (and thus $\Gamma_2 = \emptyset$) and
  $\cmult_2 = \cexp_2 = \csize_2 = 0$. Then, by the \ih on $\Phi'_{t_1}$ we
  have
  $\derivable{\Phi_{t_1}}{\sequT{\Gamma_1}{\assign{t_1}{\typeneutral}}{\cmult_1}{\cexp_1}{\csize_1}}{\SysTightCBN}$.
  Finally, we derive $\Phi$ as follows \[
\Rule{
  \derivable{\Phi_{t_1}}{\sequT{\Gamma_1}{\assign{t_1}{\typeneutral}}{\cmult_1}{\cexp_1}{\csize_1}}{\SysTightCBN}
  }{
  \sequT{\Gamma_1}{\assign{t}{\typeneutral}}{\cmult_1}{\cexp_1}{\csize_1 +1}
}{\ruleNTArrowE}
  \] We conclude since $\Gamma = \Gamma_1$, $\cmult = \cmult_1$, $\cexp =
  \cexp_1$ and $\csize = \csize_1 + 1$.
  
  \item $\ruleBTArrowI$. Then $\cbnterm{t} = \termabs{x}{\cbnterm{u}}$ and $t
  = \termabs{x}{u}$ and since $\Phi'$ is \cbnrelevant, then it has a premise of
  the form
  $\derivable{}{\sequT{\cbntype{\Gamma};\assign{x}{\cbntype{\typetight_0}}}{\assign{\cbnterm{t}}{\cbntype{\typetight_1}}}{\cmult}{\cexp}{\csize-1}}{\SysTight}$.
  The \ih\ gives
  $\derivable{}{\sequT{\Gamma;\assign{x}{\typetight_0}}{\assign{t}{\typetight_1}}{\cmult}{\cexp}{\csize-1}}{\SysTightCBN}$.
  We thus conclude
  $\derivable{}{\sequT{\Gamma}{\assign{t}{\typeabs}}{\cmult}{\cexp}{\csize}}{\SysTightCBN}$
  by rule $\ruleNTArrowI$.
  
  \item $\ruleBTBang$. Then $\cbnterm{t} = \termbang{u}$, which is not
  possible by definition. Hence this case does not apply.

  \item $\ruleBTDer$. Then $\cbnterm{t} = \termder{(u)}$ which is not
  possible by definition. Hence, this case does not apply.
  
  \item $\ruleBTESubs$. Then $t = \termsubs{x}{u}{r}$ with $\cbntype{t} =
  \termsubs{x}{\cbnterm{u}}{\termbang{\cbnterm{r}}}$, $\Gamma =
  \ctxtsum{\Gamma_1}{\Gamma_2}{}$ and $\Phi'$ is of the form \[
\Rule{
  \derivable{\Phi'_{r}}{\sequT{\cbntype{\Gamma_1};\assign{x}{\M'}}{\assign{\cbnterm{r}}{\cbntype{\sigma}}}{\cmult_1}{\cexp_1}{\csize_1}}{\SysTight}
  \quad
  \derivable{\Phi'_{u}}{\sequT{\cbntype{\Gamma_2}}{\assign{\termbang{\cbnterm{u}}}{\cbntype{\typeneutral}}}{\cmult_2}{\cexp_2}{\csize_2}}{\SysTight}
}{
  \sequT{\ctxtsum{\cbntype{\Gamma_1}}{\cbntype{\Gamma_2}}{}}{\assign{\cbntype{t}}{\cbntype{\sigma}}}{\cmult_1+\cmult_2}{\cexp_1+\cexp_2}{\csize_1+\csize_2}
}{\ruleBTESubs}
  \] with $\ptight{\M'}$. This means in particular that
  $\termbang{\cbnterm{u}}$ is typed with $\typeneutral$ which is not possible
  by construction. So this case does not apply.
  
  \item $\ruleBDAxiom$. Then $\cbnterm{t} = x$ and $t = x$. We trivially
  conclude with rule $\ruleNDAxiom$. 
  
  \item $\ruleBDArrowE$. Then $t = \termapp{t_1}{u}$ with $\cbnterm{t} =
  \termapp{\cbnterm{t_1}}{\termbang{\cbnterm{u}}}$, and $\Gamma =
  \ctxtsum{\Gamma_1}{\Gamma_2}{}$. Since $\Phi'$ is \cbnrelevant, then it is of
  the form \[
\Rule{
  \derivable{\Phi'_{r}}{\sequT{\cbntype{\Gamma_1}}{\assign{\cbnterm{t_1}}{\functtype{\cbntype{\M}}{\cbntype{\sigma}}}}{\cmult_1}{\cexp_1}{\csize_1}}{\SysTight}
  \quad
  \derivable{\Phi'_{u}}{\sequT{\cbntype{\Gamma_2}}{\assign{\termbang{\cbnterm{u}}}{\cbntype{\M}}}{\cmult_2}{\cexp_2}{\csize_2}}{\SysTight}
}{
  \sequT{\ctxtsum{\Gamma_1}{\Gamma_2}{}}{\assign{\termapp{t_1}{\termbang{\cbnterm{u}}}}{\cbntype{\sigma}}}{\cmult_1+\cmult_2+1}{\cexp_1+\cexp_2}{\csize_1+\csize_2}
}{\ruleBDArrowE}
  \] Then, by the \ih on $\Phi'_{t_1}$ there exists a type derivation
  $\derivable{\Phi_{t_1}}{\sequT{\Gamma_1}{\assign{t_1}{\functtype{\M}{\sigma}}}{\cmult_1}{\cexp_1}{\csize_1}}{\SysTightCBN}$.
  On the other hand, $\Phi'_u$ necessarily commes from a rule $\ruleBDDer$, and
  thus there exist type derivations
  $\derivable{\Phi^i_{u}}{\sequT{\cbntype{\Gamma^i_2}}{\assign{\cbnterm{u}}{\cbntype{\sigma_i}}}{\cmult^i_2}{\cexp^i_2}{\csize^i_2}}{\SysTight}$
  with $\cbntype{\M} = \intertype{\cbntype{\sigma_i}}{i \in I}$,
  $\cbntype{\Gamma_2} = +_{\iI}{\Gamma^i_2}$, $\cmult_2 = +_{\iI}{\cmult^i_2}$,
  $\cexp_2 = +_{\iI}{\cexp^i_2}$ and $\csize_2 = +_{\iI}{\csize^i_2}$. We
  conclude by applying $\ruleNDArrowE$ to derive $\Phi$ as follows \[
\Rule{
  \derivable{\Phi_{t_1}}{\sequT{\Gamma_1}{\assign{t_1}{\multiset{\functtype{\M}{\sigma}}}}{\cmult_1}{\cexp_1+1}{\csize_1}}{\SysTight}
  \quad
  \many{\derivable{\Phi^i_{u}}{\sequT{\Gamma^i_2}{\assign{u}{\sigma_i}}{\cmult^i_2}{\cexp^i_2}{\csize^i_2}}{\SysTightCBN}}{\iI}
}{
  \sequT{\Gamma}{\assign{t}{\sigma}}{\cmult}{\cexp}{\csize}
}{\ruleNTArrowE} \]
  
  \item $\ruleBDApp$. This case is similar to the previous one, resorting to
  the fact that $\ptight{\cbntype{\M}}$ implies $\ptight{\M}$.
  
  \item $\ruleBDArrowI$. Then $\cbnterm{t} = \termabs{x}{\cbnterm{u}}$ and $t
  = \termabs{x}{u}$ and since $\Phi'$ is \cbnrelevant, then it has a premise of
  the form
  $\derivable{}{\sequT{\cbntype{\Gamma};\assign{x}{\cbntype{\M}}}{\assign{\cbnterm{t}}{\cbntype{\tau}}}{\cmult}{\cexp}{\csize}}{\SysTight}$.
  The \ih\ gives
  $\derivable{}{\sequT{\Gamma;\assign{x}{\M}}{\assign{t}{\tau}}{\cmult}{\cexp}{\csize}}{\SysTightCBN}$.
  We thus conclude
  $\derivable{}{\sequT{\Gamma}{\assign{t}{\functtype{\M}{\tau}}}{\cmult}{\cexp}{\csize}}{\SysTightCBN}$
  by rule $\ruleNDArrowI$.
  
  \item $\ruleBDBang$. Then $\cbnterm{t} = \termbang{u}$, which is not
  possible by definition. Hence this case does not apply.

  \item $\ruleBDDer$. Then $\cbnterm{t} = \termder{(t')}$ which is not
  possible by definition. Hence, this case does not apply.
  
  \item $\ruleBDESubs$. This case is similar to $\ruleBTESubs$.
\end{itemize}
\end{proof}


This result is illustrated by the tight type derivations in Ex.~\ref{ex:t0cbn}
and~\ref{ex:traducciones} for the terms $t_0$ and $t'_0$ resp. Moreover,
tightness is preserved by the translation of contexts and types, hence:

\begin{corollary}
If
$\derivable{\Phi}{\sequT{\Gamma}{\assign{t}{\sigma}}{\cmult}{\cexp}{\csize}}{\SysTightCBN}$
is tight, then there exists $\cfnrml{p}$ such that $\cbnterm{t}
\rewriten{\bangweak}^{(\cmult,\cexp)} p$ with $\cmult$ $\mStep$-steps, $\cexp$
$\eStep$-steps, and $\wsize{p} = \csize$. Conversely, if
$\derivable{{\Phi'}}{\sequT{\cbntype\Gamma}{\assign{\cbnterm{t}}{\cbntype\sigma}}{\cmult}{\cexp}{\csize}}{\SysTight}$
is tight and \cbnrelevant, then there exists  $p \in \CBNNF$ such that $t
\rewriten{\callbyname}^{(\cmult,\cexp)} p$ with $\cmult$ $\mStep$-steps, $\cexp$
$\eStep$-steps, and $\nsize{p} = \csize$.
\end{corollary}

This result shows that not only from the tight type system $\SysTightCBN$
it is possible to extract exact measures for the image of the
CBN in the $\BangRev$-calculus, but more interestingly, also that from the
tight type system $\SysTight$ for the $\BangRev$-calculus it is possible to
extract exact measures for CBN. In this sense, the goal of encoding tight
typing in a unified framework is achieved.

\parrafo{Call-by-Value}
In contrast with the CBN case, the set of type for system
$\SysTightCBV$ is not a subset of that for $\SysTight$, and we need to
properly translate types. To that end, we introduce two
mutually dependent translations $\cbvtypeneg{\_}$ and $\cbvtypepos{\_}$: \[
\begin{array}{c@{\qquad}c}
\begin{array}{rcll}
\cbvtypeneg{\typetight}                     & \eqdef & \typetight & \text{if $\typetight \neq \typevar$} \\
\cbvtypeneg{\typevar}                       & \eqdef & \typeneutral \\
\cbvtypeneg{(\functtype{\M}{\sigma})}       & \eqdef & \functtype{\cbvtypeneg{\M}}{\cbvtypepos{\sigma}} \\
\cbvtypeneg{(\intertype{\sigma}{i \in I})}  & \eqdef & \intertype{\cbvtypeneg{\sigma_i}}{i \in I}
\end{array}
&
\begin{array}{rcll}
\cbvtypepos{\typetight}                     & \eqdef & \typetight & \text{if $\typetight \neq \typevar$} \\
\cbvtypepos{\typevar}                       & \eqdef & \multiset{\typeneutral} \\
\cbvtypepos{(\functtype{\M}{\sigma})}       & \eqdef & \functtype{\cbvtypeneg{\M}}{\cbvtypepos{\sigma}} \\
\cbvtypepos{(\intertype{\sigma}{i \in I})}  & \eqdef & \intertype{\cbvtypeneg{\sigma_i}}{i \in I}
\end{array}
\end{array} \]
Remark that $\cbvtypeneg{\M}= \cbvtypepos{\M}$ for every multitype $\M$.
Translation $\cbvtypeneg{\_}$ for a context $\Gamma$ is defined as expected:
$\cbvtypeneg{\Gamma} = \set{\assign{x_i}{\cbvtypeneg{\M_i}}}_{i \in I}$. To
translate terms, we resort to the function $\cbvterm{\_}$ defined in
Sec.~\ref{s:cbn-cbv-embeddings}. We also restrict $\SysTight$ derivations to
those that come from the translation of some $\SysTightCBV$ derivation. Indeed,
a $\SysTight$ derivation $\Phi$ is \emphdef{\cbvrelevant} if:
\begin{inparaenum}[(1)]
  \item all the contexts and types involved in $\Phi$
  are in the image of the translations $\cbvtypeneg{\_}$
  and $\cbvtypepos{\_}$ respectively; and 
  \item rule $\ruleBDDer$ is only applied to terms having a type of the form
  $\intertype{\functtype{\M}{\tau}}{}$. 
\end{inparaenum}
To state the preservation of typing derivations between systems $\SysTightCBV$
and $\SysTight$ we define measures over type derivations in both systems to
conveniently capture the relationship between the exponential steps in the
source and the target derivation. The intuition is that we need to compensate
for those $\dBang$-redexes of the $\BangRev$-calculus that might be introduced
when translating. For all the typing rules with two premises, we write
$\Phi_{t}$ and $\Phi_{u}$ for the first and second premise respectively. The
first measure for system $\SysTightCBV$ is given by induction on $\Phi$ as
follows: 
\begin{inparaenum}[(1)]
  \item for $\ruleVTAxiomVar$, $\countvr{\Phi} \eqdef 1$;
  \item for $\ruleVTAxiom$, $\ruleVTArrowI$ and $\ruleVDAxiom$,
  $\countvr{\Phi} \eqdef 0$;
  \item for $\ruleVTArrowE$, $\countvr{\Phi} \eqdef \countvr{\Phi_{t}} +
  \countvr{\Phi_{u}} - 1$ if $\pval{t}$;
  \item for $\ruleVDArrowE$ and $\ruleVDApp$, $\countvr{\Phi} \eqdef
  \countvr{\Phi_{t}} + \countvr{\Phi_{u}} + 1$ if $\neg\pval{t}$; and
  \item in any other case $\countvr{\Phi}$ is defined as the sum of the
  recursive calls over all premises.
\end{inparaenum}
The second measure for system $\SysTight$ is also defined by induction on
$\Phi$ as follows:
\begin{inparaenum}[(1)]
  \item for $\ruleBTBang$ and $\ruleBDAxiom$, $\inversevr{\Phi} \eqdef 0$;
  \item for $\ruleBTArrowE$, $\inversevr{\Phi} \eqdef \inversevr{\Phi_{t}} +
  \inversevr{\Phi_{u}} - 1$ if $\pval{t}$;
  \item for $\ruleBDArrowE$ and $\ruleBDApp$, $\inversevr{\Phi} \eqdef
  \inversevr{\Phi_{t}} + \inversevr{\Phi_{u}} + 1$ if $\neg\pval{t}$; and
  \item in any other case $\inversevr{\Phi}$ is defined as the sum of the
  recursive calls over all premises.
\end{inparaenum}
Then, we obtain:

\begin{theorem}\mbox{}
\begin{enumerate}
  \item\label{t:translation:value:to-bang} If
  $\derivable{\Phi}{\sequT{\Gamma}{\assign{t}{\sigma}}{\cmult}{\cexp}{\csize}}{\SysTightCBV}$,
  then
  $\derivable{\Phi'}{\sequT{\cbvtypeneg{\Gamma}}{\assign{\cbvterm{t}}{\cbvtypepos{\sigma}}}{\cmult}{\cexp'}{\csize}}{\SysTight}$
  is \cbvrelevant\ with $\cexp' = \cexp + \countvr{\Phi}$.

  \item\label{t:translation:value:from-bang} If
  $\derivable{\Phi'}{\sequT{\cbvtypeneg{\Gamma}}{\assign{\cbvterm{t}}{\cbvtypepos{\sigma}}}{\cmult}{\cexp'}{\csize}}{\SysTight}$
  is \cbvrelevant, then
  $\derivable{\Phi}{\sequT{\Gamma}{\assign{t}{\sigma}}{\cmult}{\cexp}{\csize}}{\SysTightCBV}$
  with $\cexp = \cexp' - \inversevr{\Phi'}$.
\end{enumerate}
\label{t:translation:value}
\end{theorem}

\begin{proof}
\begin{enumerate}
  \item By induction on $\Phi$ by analysing the last rule applied.
\begin{itemize}
  \item $\ruleVTAxiomVar$. Then
  $\derivable{\Phi}{\sequT{\assign{x}{\multiset{\typevar}}}{\assign{x}{\typevar}}{0}{0}{0}}{\SysTightCBV}$
  and $\cbvterm{t} = \termbang{x}$. We conclude by applying $\ruleBDAxiom$ and
  $\ruleBDBang$, obtaining
  $\derivable{\Phi'}{\sequT{\assign{x}{\multiset{\typeneutral}}}{\assign{\termbang{x}}{\multiset{\typeneutral}}}{0}{1}{0}}{\SysTight}$,
  which is trivially \cbvrelevant. The counters are as expected since
  $\countvr{\Phi} = 1$.
  
  \item $\ruleVTAxiom$. Then
  $\derivable{\Phi}{\sequT{}{\assign{x}{\typevalue}}{0}{0}{0}}{\SysTightCBV}$
  and $\cbvterm{t} = \termbang{x}$. We conclude by $\ruleBTBang$ with
  $\derivable{\Phi'}{\sequT{}{\assign{\termbang{x}}{\typevalue}}{0}{0}{0}}{\SysTight}$,
  which is trivially \cbvrelevant. Note that the counters are as expected since
  $\countvr{\Phi} = 0$.
  
  \item $\ruleVTArrowI$. Then
  $\derivable{\Phi}{\sequT{}{\assign{\termabs{x}{r}}{\typevalue}}{0}{0}{0}}{\SysTightCBV}$
  and $\cbvterm{t} = \termbang{\termabs{x}{\cbvterm{r}}}$. We conclude by
  $\ruleBTBang$ with
  $\derivable{\Phi'}{\sequT{}{\assign{\termbang{\termabs{x}{\cbvterm{r}}}}{\typevalue}}{0}{0}{0}}{\SysTight}$,
  which is trivially \cbvrelevant. Note that the counters are as expected since
  $\countvr{\Phi} = 0$.
  
  \item $\ruleVTArrowE$. Then $t = \termapp{r}{u}$, $\sigma = \typeneutral$
  and $\Phi$ is of the form \[
\Rule{
  \derivable{\Phi_{r}}{\sequT{\Gamma_1}{\assign{r}{\overline{\typevalue}}}{\cmult_1}{\cexp_1}{\csize_1}}{\SysTightCBV}
  \quad
  \derivable{\Phi_{u}}{\sequT{\Gamma_2}{\assign{u}{\overline{\typevar}}}{\cmult_2}{\cexp_2}{\csize_2}}{\SysTightCBV}
}{
  \sequT{\ctxtsum{\Gamma_1}{\Gamma_2}{}}{\assign{\termapp{r}{u}}{\typeneutral}}{\cmult_1+\cmult_2}{\cexp_1+\cexp_2}{\csize_1+\csize_2+1}
}{\ruleVTArrowE}
  \] where $\cexp= \cexp_1+\cexp_2$. By the \ih we have \cbvrelevant\
  derivations
  $\derivable{\Phi'_{r}}{\sequT{\cbvtypeneg{\Gamma_1}}{\assign{\cbvterm{r}}{\cbvtypepos{\overline{\typevalue}}}}{\cmult_1}{\cexp'_1}{\csize_1}}{\SysTight}$
  and
  $\derivable{\Phi'_{u}}{\sequT{\cbvtypeneg{\Gamma_2}}{\assign{\cbvterm{u}}{\cbvtypepos{\overline{\typevar}}}}{\cmult_2}{\cexp'_2}{\csize_2}}{\SysTight}$
  with $\cexp'_1 = \cexp_1 + \countvr{\Phi_{r}}$ and $\cexp'_2 = \cexp_2 +
  \countvr{\Phi_{u}}$. Moreover, since $u$ is typed with $\overline{\typevar}$,
  then it is typed with $\typeneutral$ or $\typevalue$, and thus
  $\cbvtypepos{\overline{\typevar}} = \overline{\typeabs}$. There are three
  possible cases for $r$:
  \begin{enumerate}
  \item $r = \ctxtapp{\ctxt{L}}{x}$. Then $\overline{\typevalue} = \typevar$
    with $\cbvtypepos{\overline{\typevalue}} = \multiset{\typeneutral}$.
    Moreover, by definition $\cbvterm{r} =
    \ctxtapp{\cbvterm{\ctxt{L}}}{\termbang{x}}$. Thus, we have
    $\derivable{\Phi'_{r}}{\sequT{\cbvtypeneg{\Gamma_1}}{\assign{\ctxtapp{\cbvterm{\ctxt{L}}}{\termbang{x}}}{\multiset{\typeneutral}}}{\cmult_1}{\cexp'_1}{\csize_1}}{\SysTight}$.
    With a simple induction on $\ctxt{L}$ we can show that there exists a
    \cbvrelevant\ type derivation
    $\derivable{\Phi''}{\sequT{\cbvtypeneg{\Gamma_1}}{\assign{\ctxtapp{\cbvterm{\ctxt{L}}}{x}}{\typeneutral}}{\cmult_1}{\cexp'_1-1}{\csize_1}}{\SysTight}$,
    that results from removing the occurrence of rule $\ruleBDBang$ in the base
    case. Finally, we conclude by applying $\ruleBTArrowE$ to derive a
    \cbvrelevant\ derivation $\Phi'$ as follows \[
\Rule{
  \derivable{\Phi''}{\sequT{\cbvtypeneg{\Gamma_1}}{\assign{\ctxtapp{\cbvterm{\ctxt{L}}}{x}}{\typeneutral}}{\cmult_1}{\cexp'_1-1}{\csize_1}}{\SysTight}
  \quad
  \derivable{\Phi'_{u}}{\sequT{\cbvtypeneg{\Gamma_2}}{\assign{\cbvterm{u}}{\overline{\typeabs}}}{\cmult_2}{\cexp'_2}{\csize_2}}{\SysTight}
}{
  \sequT{\cbvtypeneg{\Gamma}}{\assign{\cbvterm{t}}{\typeneutral}}{\cmult}{\cexp'}{\csize}
}{\ruleBTArrowE}
    \] where $\cexp' = \cexp'_1 - 1 + \cexp'_2 = \cexp_1 + \countvr{\Phi_{r}} -
    1 + \cexp_2 + \countvr{\Phi_{u}} = \cexp + \countvr{\Phi}$ as expected.
    \item $r = \ctxtapp{\ctxt{L}}{\termapp{r'}{r''}}$. Then, we have
    $\derivable{\Phi'_{r}}{\sequT{\cbvtypeneg{\Gamma_1}}{\assign{\cbvterm{r}}{\typeneutral}}{\cmult_1}{\cexp'_1}{\csize_1}}{\SysTight}$.
    We conclude by applying $\ruleBTDer$ and
    $\ruleBTArrowE$ to derive $\Phi'$ as follows \[
\Rule{
  \Rule{
    \derivable{\Phi'_{r}}{\sequT{\cbvtypeneg{\Gamma_1}}{\assign{\cbvterm{r}}{\typeneutral}}{\cmult_1}{\cexp'_1}{\csize_1}}{\SysTight}
  }{
    \sequT{\cbvtypeneg{\Gamma_1}}{\assign{\termder{(\cbvterm{r})}}{\typeneutral}}{\cmult_1}{\cexp'_1}{\csize_1}
  }{\ruleBTDer}
  \derivable{\Phi'_{u}}{\sequT{\cbvtypeneg{\Gamma_2}}{\assign{\cbvterm{u}}{\overline{\typeabs}}}{\cmult_2}{\cexp'_2}{\csize_2}}{\SysTight}
}{
  \sequT{\cbvtypeneg{\Gamma}}{\assign{\cbvterm{t}}{\typeneutral}}{\cmult}{\cexp'}{\csize}
}{\ruleBTArrowE}
    \] where $\cexp' = \cexp'_1 + \cexp'_2$. Hence, $\cexp' = \cexp'_1 +
    \cexp'_2 = \cexp_1 + \countvr{\Phi_{r}} + \cexp_2 + \countvr{\Phi_{u}} =
    \cexp + \countvr{\Phi}$ as expected. Notice that $\Phi'$ is \cbvrelevant\
    by the \ih and the fact that no $\ruleBDDer$ rule was added in the final
    step.

    \item $r = \ctxtapp{\ctxt{L}}{\termabs{x}{r'}}$. Then, it is necessarily
    the case that $\Phi$ has a \cbvrelevant\ subderivation for a judgement of
    the form
    $\sequT{\Gamma''}{\assign{\termabs{x}{r'}}{\overline{\typevalue}}}{\cmult''}{\cexp''}{\csize''}$
    for some appropriate $\Gamma''$, $\cmult''$, $\cexp''$ and $\csize''$.
    Since $\overline{\typevalue}$ is $\typevar$ or $\typeneutral$, this
    leads to a contradiction because abstractions cannot be assigned type
    $\typevar$ or $\typeneutral$ in system $\SysTightCBV$. Hence, this case
    does not apply.
  \end{enumerate}
  
  \item $\ruleVTESubs$. Then $t = \termsubs{x}{u}{r}$ and $\Phi$ is of the
  form \[
\Rule{
  \derivable{\Phi_{r}}{\sequT{\Gamma_1}{\assign{r}{\sigma}}{\cmult_1}{\cexp_1}{\csize_1}}{\SysTightCBV}
  \quad
  \derivable{\Phi_{u}}{\sequT{\Gamma_2}{\assign{u}{\typeneutral}}{\cmult_2}{\cexp_2}{\csize_2}}{\SysTightCBV}
}{
  \sequT{\ctxtsum{\Gamma_1}{\Gamma_2}{}}{\assign{\termsubs{x}{u}{r}}{\sigma}}{\cmult_1+\cmult_2}{\cexp_1+\cexp_2}{\csize_1+\csize_2}
}{\ruleVTESubs}
  \] with $\ptight{\Gamma_1(x)}$ and $\cexp= \cexp_1+\cexp_2$. By the \ih on
  $\Phi_{u}$ there exists a \cbvrelevant\ derivation
  $\derivable{\Phi'_{u}}{\sequT{\cbvtypeneg{\Gamma_2}}{\assign{\cbvterm{u}}{\typeneutral}}{\cmult_2}{\cexp'_2}{\csize_2}}{\SysTight}$
  with $\cexp'_2 = \cexp_2 + \countvr{\Phi_{u}}$. By the \ih again there exists
  a \cbvrelevant\ derivation
  $\derivable{\Phi'_{r}}{\sequT{\cbvtypeneg{\Gamma_1}}{\assign{\cbvterm{r}}{\cbvtypepos{\sigma}}}{\cmult_1}{\cexp'_1}{\csize_1}}{\SysTight}$
  with $\cexp'_1 = \cexp_1 + \countvr{\Phi_{r}}$. Thus, we conclude by applying
  $\ruleBTESubs$ with $\Phi'_{r}$, $\Phi'_{u}$ and $\ptight{\cbvtypeneg{\Gamma_1}(x)}$,
  hence obtaining the \cbvrelevant\ derivation
  $\derivable{\Phi'}{\sequT{\cbvtypeneg{\Gamma}}{\assign{\cbvterm{t}}{\cbvtypepos{\sigma}}}{\cmult}{\cexp'}{\csize}}{\SysTight}$
  where $\cexp' = \cexp'_1 + \cexp'_2 = \cexp_1 + \countvr{\Phi_{r}} + \cexp_2
  + \countvr{\Phi_{u}} = \cexp + \countvr{\Phi}$ as expected.
  
  \item $\ruleVDAxiom$. Then
  $\derivable{\Phi}{\sequT{\assign{x}{\M}}{\assign{x}{\M}}{0}{1}{0}}{\SysTightCBV}$
  with $\M = \intertype{\sigma_i}{i \in I}$. We conclude by $\ruleBDAxiom$ and
  $\ruleBDBang$, deriving a \cbvrelevant\ derivation $\Phi'$ as follows \[
\Rule{
  \left(\hspace{-1em}\raisebox{-1em}{
    \Rule{}{
      \sequT{\assign{x}{\multiset{\cbvtypeneg{\sigma_i}}}}{\assign{x}{\cbvtypeneg{\sigma_i}}}{0}{0}{0}
    }{\ruleBDAxiom}
  }\hspace{-0.5em}\right)_{i \in I}
}{
  \sequT{\assign{x}{\cbvtypeneg{\M}}}{\assign{\termbang{x}}{\cbvtypeneg{\M}}}{0}{1}{0}
}{\ruleBDBang}
  \] Note that the counters are as expected since $\countvr{\Phi} = 0$. We
  conclude since $\cbvtypeneg{\M}= \cbvtypepos{\M}$. 
  
  \item $\ruleVDArrowE$. Then $t = \termapp{r}{u}$ and $\Phi$ is of the form \[
\Rule{
  \derivable{\Phi_{r}}{\sequT{\Gamma_1}{\assign{r}{\multiset{\functtype{\M}{\sigma}}}}{\cmult_1}{\cexp_1}{\csize_1}}{\SysTightCBV}
  \quad
  \derivable{\Phi_{u}}{\sequT{\Gamma_2}{\assign{u}{\M}}{\cmult_2}{\cexp_2}{\csize_2}}{\SysTightCBV}
}{
  \sequT{\ctxtsum{\Gamma_1}{\Gamma_2}{}}{\assign{\termapp{r}{u}}{\sigma}}{\cmult_1+\cmult_2+1}{\cexp_1+\cexp_2-1}{\csize_1+\csize_2}
}{\ruleVDArrowE}
  \] where $\cmult = \cmult_1 + \cmult_2 + 1$, $\cexp = \cexp_1 + \cexp_2 - 1$,
  and $\csize = \csize_1 + \csize_2$. By the \ih on $\Phi_{u}$ there exists a
  \cbvrelevant\ derivation
  $\derivable{\Phi'_{u}}{\sequT{\cbvtypeneg{\Gamma_2}}{\assign{\cbvterm{u}}{\cbvtypepos{\M}}}{\cmult_2}{\cexp'_2}{\csize_2}}{\SysTight}$
  with $\cexp'_2 = \cexp_2 + \countvr{\Phi_{u}}$. And by the \ih on $\Phi_{r}$
  there exists a \cbvrelevant\ derivation
  $\derivable{\Phi'_{r}}{\sequT{\cbvtypeneg{\Gamma_1}}{\assign{\cbvterm{r}}{\multiset{\functtype{\cbvtypeneg{\M}}{\cbvtypepos{\sigma}}}}}{\cmult_1}{\cexp'_1}{\csize_1}}{\SysTight}$
  with $\cexp'_1 = \cexp_1 + \countvr{\Phi_{r}}$. We should consider two cases
  based on the shape of $r$:
  \begin{enumerate}
    \item $r = \ctxtapp{\ctxt{L}}{v}$. Then $\cbvterm{r} =
    \ctxtapp{\cbvterm{\ctxt{L}}}{\termbang{r'}}$ and $\cbvterm{t} =
    \termapp{\ctxtapp{\cbvterm{\ctxt{L}}}{r'}}{\cbvterm{u}}$. With a simple
    induction on $\ctxt{L}$ we can show from $\Phi'_{r}$ that there exists a
    \cbvrelevant\ type derivation
    $\derivable{\Phi''}{\sequT{\cbvtypeneg{\Gamma_1}}{\assign{\ctxtapp{\cbvterm{\ctxt{L}}}{r'}}{\functtype{\cbvtypeneg{\M}}{\cbvtypepos{\sigma}}}}{\cmult_1}{\cexp'_1-1}{\csize_1}}{\SysTight}$,
    that results from removing the occurrence of rule $\ruleBDBang$ in the base
    case. Finally, since $\cbvtypeneg{\M} = \cbvtypepos{\M}$, we conclude by
    applying $\ruleBDArrowE$ with $\Phi''$ and $\Phi_{u}$ to derive a
    \cbvrelevant\ derivation
    $\derivable{\Phi'}{\sequT{\cbvtypeneg{\Gamma}}{\assign{\cbvterm{t}}{\cbvtypepos{\sigma}}}{\cmult}{\cexp'}{\csize}}{\SysTight}$
    where $\cexp' = \cexp'_1 - 1 + \cexp'_2 = \cexp_1 + \countvr{\Phi_{r}} - 1
    + \cexp_2 + \countvr{\Phi_{u}} = \cexp + \countvr{\Phi}$ as expected.
    
    \item $r \neq \ctxtapp{\ctxt{L}}{v}$. Then $\cbvterm{t} =
    \termapp{\termder{(\cbvterm{r})}}{\cbvterm{u}}$. Since $\cbvtypeneg{\M} = \cbvtypepos{\M}$,
    we conclude by deriving $\Phi'$ as follows \[
\Rule{
  \Rule{
    \derivable{\Phi'_{r}}{\sequT{\cbvtypeneg{\Gamma_1}}{\assign{\cbvterm{r}}{\multiset{\functtype{\cbvtypeneg{\M}}{\cbvtypepos{\sigma}}}}}{\cmult_1}{\cexp'_1}{\csize_1}}{\SysTight}
  }{
    \sequT{\cbvtypeneg{\Gamma_1}}{\assign{\termder{(\cbvterm{r})}}{\functtype{\cbvtypeneg{\M}}{\cbvtypepos{\sigma}}}}{\cmult_1}{\cexp'_1}{\csize_1}
  }{\ruleBDDer}
  \derivable{\Phi'_{u}}{\sequT{\cbvtypeneg{\Gamma_2}}{\assign{\cbvterm{u}}{\cbvtypepos{\M}}}{\cmult_2}{\cexp'_2}{\csize_2}}{\SysTight}
}{
  \sequT{\cbvtypeneg{\Gamma}}{\assign{\cbvterm{t}}{\cbvtypepos{\sigma}}}{\cmult}{\cexp'}{\csize}
}{\ruleBDArrowE}
    \] where $\cexp' = \cexp'_1 + \cexp'_2 = \cexp_1 + \countvr{\Phi_{r}} +
    \cexp_2 + \countvr{\Phi_{u}} = \cexp + 1 + \countvr{\Phi_{r}} +
    \countvr{\Phi_{u}} = \cexp + \countvr{\Phi}$ as expected. Notice that the
    resulting derivation $\Phi'$ is \cbvrelevant since the added $\ruleBDDer$
    rule acts on a multiset of the required form.
  \end{enumerate}
  
  \item $\ruleVDApp$. This case is similar to the previous one.
  
  \item $\ruleVDArrowI$. Then $t = \termabs{x}{r}$ and $\Phi$ is of the form \[
\Rule{
  \many{
    \derivable{\Phi_i}{\sequT{\Gamma_i}{\assign{r}{\tau_i}}{\cmult_i}{\cexp_i}{\csize_i}}{\SysTightCBV}
  }{i \in I}
}{
  \sequT{\ctxtsum{}{\ctxtres{\Gamma_i}{x}{}}{i \in I}}{\assign{\termabs{x}{r}}{\intertype{\functtype{\Gamma_i(x)}{\tau_i}}{i \in I}}}{+_{i \in I}{\cmult_i}}{1+_{i \in I}{\cexp_i}}{+_{i \in I}{\csize_i}}
}{\ruleVDArrowI}
  \] where $\cexp = 1 +_{i \in I}{\cexp_i}$. By the \ih there exists a
  derivation
  $\derivable{\Phi'_i}{\sequT{\cbvtypeneg{\Gamma_i}}{\assign{\cbvterm{r}}{\cbvtypepos{\tau_i}}}{\cmult_i}{\cexp'_i}{\csize_i}}{\SysTight}$,
  which is \cbvrelevant, where $\cexp'_i = \cexp_i + \countvr{\Phi_{i}}$ for
  each $i \in I$. We then conclude by applying $\ruleBDArrowI$ and
  $\ruleBDBang$, hence deriving a \cbvrelevant\ derivation $\Phi'$ as
  follows \[
\Rule{
  \left(\hspace{-1em}\raisebox{-1em}{
    \Rule{
      \derivable{\Phi'_i}{\sequT{\cbvtypeneg{\Gamma_i}}{\assign{\cbvterm{r}}{\cbvtypepos{\tau_i}}}{\cmult_i}{\cexp'_i}{\csize_i}}{\SysTight}
    }{
      \sequT{\ctxtres{\cbvtypeneg{\Gamma_i}}{x}{}}{\assign{\termabs{x}{\cbvterm{r}}}{\functtype{\cbvtypeneg{\Gamma_i(x)}}{\cbvtypepos{\tau_i}}}}{\cmult_i}{\cexp'_i}{\csize_i}
    }{\ruleBDArrowI}
  }\hspace{-0.5em}\right)_{i \in I}
    }{
  \sequT{\ctxtsum{}{\ctxtres{\cbvtypeneg{\Gamma_i}}{x}{}}{i \in I}}{\assign{\cbvterm{t}}{\cbvtypepos{\intertype{\functtype{\Gamma_i(x)}{\tau_i}}{i \in I}}}}{+_{i \in I}{\cmult_i}}{1+_{i \in I}{\cexp'_i}}{+_{i \in I}{\csize_i}}
}{\ruleBDBang}
  \] taking $\cexp' = 1 +_{i \in I}{\cexp'_i} = 1
  +_{i \in I}{(\cexp_i + \countvr{\Phi_{i}})} = \cexp + \countvr{\Phi}$ as
  expected.
  
  \item $\ruleVDESubs$. This case is similar to $\ruleVTESubs$.
\end{itemize}
\end{enumerate}

\begin{enumerate}
  \item By induction on $\Phi'$ analysing the last rule applied.
\begin{itemize}
  \item $\ruleBTArrowE$. Then $t = \termapp{t_1}{u}$ with $\cbvterm{t} =
  \termapp{r}{\cbvterm{u}}$ for some term $r$, $\Gamma =
  \ctxtsum{\Gamma_1}{\Gamma_2}{}$, $\cbvtypepos{\sigma} = \typeneutral =
  \sigma$ and $\Phi'$ is of the form \[
\Rule{
  \derivable{\Phi'_{r}}{\sequT{\cbvtypeneg{\Gamma_1}}{\assign{r}{\typeneutral}}{\cmult_1}{\cexp'_1}{\csize_1}}{\SysTight}
  \quad
  \derivable{\Phi'_{u}}{\sequT{\cbvtypeneg{\Gamma_2}}{\assign{\cbvterm{u}}{\overline{\typeabs}}}{\cmult_2}{\cexp'_2}{\csize_2}}{\SysTight}
}{
  \sequT{\ctxtsum{\cbvtypeneg{\Gamma_1}}{\cbvtypeneg{\Gamma_2}}{}}{\assign{\cbvterm{t}}{\typeneutral}}{\cmult_1+\cmult_2}{\cexp'_1+\cexp'_2}{\csize_1+\csize_2+1}
}{\ruleBTArrowE}
  \] where $\cexp' = \cexp'_1 + \cexp'_2$. Moreover, since $\cbvterm{u}$ is
  typed with $\overline{\typeabs}$, then it is typed with $\typeneutral$ or
  $\typevalue$, \ie $\cbvtypepos{\overline{\typevar}}$. Then, by the \ih on
  $\Phi'_{u}$ we have
  $\derivable{\Phi_{u}}{\sequT{\Gamma_2}{\assign{u}{\overline{\typevar}}}{\cmult_2}{\cexp_2}{\csize_2}}{\SysTightCBV}$
  with $\cexp_2 = \cexp'_2 - \inversevr{\Phi'_{u}}$. There are two possible
  cases for $\cbvterm{t_1}$ and, hence, for $r$:
  \begin{enumerate}
    \item $\cbvterm{t_1} = \ctxtapp{\cbvterm{\ctxt{L}}}{\termbang{r'}}$ (\ie
    $\pval{t_1}$). Then, $r = \ctxtapp{\cbvterm{\ctxt{L}}}{r'}$. Moreover, by a
    simple induction on $\ctxt{L}$, from $\Phi'_{r}$ we get the \cbvrelevant\
    type derivation
    $\derivable{\Phi'_{t_1}}{\sequT{\cbvtypeneg{\Gamma_1}}{\assign{\ctxtapp{\cbvterm{\ctxt{L}}}{\termbang{r'}}}{\multiset{\typeneutral}}}{\cmult_1}{\cexp'_1+1}{\csize_1}}{\SysTight}$
    such that $\inversevr{\Phi'_{t_1}} = \inversevr{\Phi'_{r}}$. Note that
    $\multiset{\typeneutral} = \cbvtypepos{\typevar}$. Then, by the \ih we get
    $\derivable{\Phi_{t_1}}{\sequT{\Gamma_1}{\assign{t_1}{\typevar}}{\cmult_1}{\cexp_1+1}{\csize_1}}{\SysTightCBV}$
    with $\cexp_1 = \cexp'_1 - \inversevr{\Phi'_{t_1}} = \cexp'_1 -
    \inversevr{\Phi'_{r}}$. Finally, we conclude by applying $\ruleVTArrowE$ to
    derive $\Phi$ as follows \[
\Rule{
  \derivable{\Phi_{t_1}}{\sequT{\Gamma_1}{\assign{t_1}{\typevar}}{\cmult_1}{\cexp_1+1}{\csize_1}}{\SysTightCBV}
  \quad
  \derivable{\Phi_{u}}{\sequT{\Gamma_2}{\assign{u}{\overline{\typevar}}}{\cmult_2}{\cexp_2}{\csize_2}}{\SysTightCBV}
}{
  \sequT{\Gamma}{\assign{t}{\typeneutral}}{\cmult}{\cexp}{\csize}
}{\ruleVTArrowE}
    \] where $\cexp = \cexp_1 + 1 + \cexp_2 = \cexp'_1 -
    \inversevr{\Phi'_{r}} + 1 + \cexp'_2 - \inversevr{\Phi'_{u}} = \cexp'_1 +
    \cexp'_2 - (\inversevr{\Phi'_{r}} + \inversevr{\Phi'_{u}} - 1) =
    \cexp' - \inversevr{\Phi'}$ as expected.

    \item $\cbvterm{t_1} \neq \ctxtapp{\ctxt{L}}{\termbang{r'}}$ (\ie
    $\neg\pval{t_1}$). Then, $r = \termder{(\cbvterm{t_1})}$. Since $\Phi'$ is
    \cbvrelevant, it is necessarily the case where $\Phi'_{r}$ comes from an
    application of rule $\ruleBTDer$. Moreover, $\typeneutral =
    \cbvtypepos{\typeneutral}$, hence we have a \cbvrelevant\ type derivation
    $\derivable{\Phi'_{t_1}}{\sequT{\cbvtypeneg{\Gamma_1}}{\assign{\cbvterm{t_1}}{\cbvtypepos{\typeneutral}}}{\cmult_1}{\cexp'_1}{\csize_1}}{\SysTight}$
    satisfying the hypothesis of the theorem. By the \ih we get a type
    derivation
    $\derivable{\Phi_{t_1}}{\sequT{\Gamma_1}{\assign{t_1}{\typeneutral}}{\cmult_1}{\cexp_1}{\csize_1}}{\SysTightCBV}$
    with $\cexp_1 = \cexp'_1 - \inversevr{\Phi'_{t_1}} = \cexp'_1 -
    \inversevr{\Phi'_{r}}$. We conclude by applying $\ruleVTArrowE$ to derive
    $\Phi$ \[
\Rule{
  \derivable{\Phi_{t_1}}{\sequT{\Gamma_1}{\assign{t_1}{\typeneutral}}{\cmult_1}{\cexp_1}{\csize_1}}{\SysTightCBV}
  \quad
  \derivable{\Phi_{u}}{\sequT{\Gamma_2}{\assign{u}{\overline{\typevar}}}{\cmult_2}{\cexp_2}{\csize_2}}{\SysTightCBV}
}{
  \sequT{\Gamma}{\assign{t}{\typeneutral}}{\cmult}{\cexp}{\csize}
}{\ruleVTArrowE}
    \] where $\cexp = \cexp_1 + \cexp_2 = \cexp'_1 - \inversevr{\Phi'_{r}} +
    \cexp'_2 - \inversevr{\Phi'_{u}} = \cexp'_1 + \cexp'_2 -
    (\inversevr{\Phi'_{r}} + \inversevr{\Phi'_{u}}) = \cexp' -
    \inversevr{\Phi'}$ as expected.
  \end{enumerate}
  
  \item $\ruleBTArrowI$. Then $\cbvterm{t} = \termabs{x}{t'}$ which is not
  possible by definition. Hence, this case does not apply.
  
  \item $\ruleBTBang$. Then $\cbvterm{t} = \termbang{t'}$, $\cbvtypepos{\sigma}
  = \typevalue = \sigma$ and
  $\derivable{\Phi'}{\sequT{}{\assign{\termbang{t'}}{\typevalue}}{0}{0}{0}}{\SysTight}$.
  There are two possible cases: either 
  \begin{inparaenum}[(1)]
    \item $t = x = t'$; or
    \item $t = \termabs{x}{u}$ with $t' = \termabs{x}{\cbvterm{u}}$.
  \end{inparaenum}
  We resort to $\ruleVTAxiom$ or $\ruleVTArrowI$ resp. to conclude with
  $\derivable{\Phi}{\sequT{}{\assign{t}{\typevalue}}{0}{0}{0}}{\SysTightCBV}$.
  Note that the counters are as expected since $\inversevr{\Phi'} = 0$.

  \item $\ruleBTDer$. Then $\cbvterm{t} = \termder{(t')}$ which is not
  possible by definition. Hence, this case does not apply.
  
  \item $\ruleBTESubs$. Then $t = \termsubs{x}{u}{r}$ with $\cbvtypepos{t} =
  \termsubs{x}{\cbvterm{u}}{\cbvterm{r}}$, $\Gamma =
  \ctxtsum{\Gamma_1}{\Gamma_2}{}$ and $\Phi'$ is of the form \[
\Rule{
  \derivable{\Phi'_{r}}{\sequT{\cbvtypeneg{\Gamma_1};\assign{x}{\M'}}{\assign{\cbvterm{r}}{\cbvtypepos{\sigma}}}{\cmult_1}{\cexp'_1}{\csize_1}}{\SysTight}
  \quad
  \derivable{\Phi'_{u}}{\sequT{\cbvtypeneg{\Gamma_2}}{\assign{\cbvterm{u}}{\cbvtypepos{\typeneutral}}}{\cmult_2}{\cexp'_2}{\csize_2}}{\SysTight}
}{
  \sequT{\ctxtsum{\cbvtypeneg{\Gamma_1}}{\cbvtypeneg{\Gamma_2}}{}}{\assign{\cbvtypepos{t}}{\cbvtypepos{\sigma}}}{\cmult_1+\cmult_2}{\cexp'_1+\cexp'_2}{\csize_1+\csize_2}
}{\ruleBTESubs}
  \] with $\ptight{\M'}$ and $\cexp' = \cexp'_1 + \cexp'_2$. Moreover, since
  $\Phi'$ is \cbvrelevant, $\M'$ is of the form $\cbvtypeneg{\M}$. Then, by the
  \ih there exists type derivations
  $\derivable{\Phi_{r}}{\sequT{\Gamma_1;\assign{x}{\M}}{\assign{r}{\sigma}}{\cmult_1}{\cexp_1}{\csize_1}}{\SysTightCBV}$
  and
  $\derivable{\Phi_{u}}{\sequT{\Gamma_2}{\assign{u}{\typeneutral}}{\cmult_2}{\cexp_2}{\csize_2}}{\SysTightCBV}$
  with $\cexp_1 = \cexp'_1 - \inversevr{\Phi'_{r}}$ and $\cexp_2 = \cexp'_2 -
  \inversevr{\Phi'_{u}}$ resp. Finally, since $\ptight{\cbvtypeneg{\M}}$ implies
  $\ptight{\M}$, we conclude by applying rule $\ruleVTESubs$ with $\Phi_{r}$,
  $\Phi_{u}$ obtaining the type derivation
  $\derivable{\Phi}{\sequT{\Gamma}{\assign{t}{\sigma}}{\cmult}{\cexp}{\csize}}{\SysTightCBV}$
  where $\cexp = \cexp_1 + \cexp_2 = \cexp'_1 - \inversevr{\Phi'_{r}} +
  \cexp'_2 - \inversevr{\Phi'_{u}} = \cexp' - \inversevr{\Phi'}$ as expected.
  
  \item $\ruleBDAxiom$. Then $\cbvterm{t} = x$ which is not possible by
  definition. Hence, this case does not apply.
  
  \item $\ruleBDArrowE$. Then $t = \termapp{t_1}{u}$
  with $\cbvterm{t} = \termapp{r}{\cbvterm{u}}$ for some term $r$, $\Gamma =
  \ctxtsum{\Gamma_1}{\Gamma_2}{}$. Moreover, since $\Phi'$ is \cbvrelevant,
  then it is of the form \[
\Rule{
  \derivable{\Phi'_{t_1}}{\sequT{\cbvtypeneg{\Gamma_1}}{\assign{r}{\functtype{\cbvtypeneg{\M}}{\cbvtypepos{\sigma}}}}{\cmult_1}{\cexp'_1}{\csize_1}}{\SysTight}
  \quad
  \derivable{\Phi'_{u}}{\sequT{\cbvtypeneg{\Gamma_2}}{\assign{\cbvterm{u}}{\cbvtypeneg{\M}}}{\cmult_2}{\cexp'_2}{\csize_2}}{\SysTight}
}{
  \sequT{\ctxtsum{\Gamma_1}{\Gamma_2}{}}{\assign{\termapp{r}{\cbvterm{u}}}{\cbvtypepos{\sigma}}}{\cmult_1+\cmult_2+1}{\cexp'_1+\cexp'_2}{\csize_1+\csize_2}
}{\ruleBDArrowE}
  \] where $\cexp' = \cexp'_1 + \cexp'_2$. Moreover, by definition
  $\cbvtypeneg{\M} = \cbvtypepos{\M}$. Then, by the \ih on $\Phi'_{u}$ there
  exists a type derivation
  $\derivable{\Phi_{u}}{\sequT{\Gamma_2}{\assign{u}{\M}}{\cmult_2}{\cexp_2}{\csize_2}}{\SysTightCBV}$
  with $\cexp_2 = \cexp'_2 - \inversevr{\Phi_{u}}$. We should consider two
  cases based on the shape of $\cbvterm{t_1}$ and, hence, of $r$:
  \begin{enumerate}
    \item $\cbvterm{t_1} = \ctxtapp{\cbvterm{\ctxt{L}}}{\termbang{r'}}$ (\ie
    $\pval{t_1}$). Then, $r = \ctxtapp{\cbvterm{\ctxt{L}}}{r'}$. Moreover, by a
    simple induction on $\ctxt{L}$, from $\Phi'_{r}$ we get
    $\derivable{\Phi'_{t_1}}{\sequT{\cbvtypeneg{\Gamma_1}}{\assign{\ctxtapp{\cbvterm{\ctxt{L}}}{\termbang{r'}}}{\multiset{\functtype{\cbvtypeneg{\M}}{\cbvtypepos{\sigma}}}}}{\cmult_1}{\cexp'_1+1}{\csize_1}}{\SysTight}$
    \cbvrelevant\ such that $\inversevr{\Phi'_{t_1}} = \inversevr{\Phi'_{r}}$.
    Note that
    $\multiset{\functtype{\cbvtypeneg{\M}}{\cbvtypepos{\sigma}}} =
    \cbvtypepos{\multiset{\functtype{\M}{\sigma}}}$. Then, by the \ih we get
    $\derivable{\Phi_{t_1}}{\sequT{\Gamma_1}{\assign{t_1}{\multiset{\functtype{\M}{\sigma}}}}{\cmult_1}{\cexp_1+1}{\csize_1}}{\SysTight}$
    with $\cexp_1 = \cexp'_1 - \inversevr{\Phi_{t_1}} = \cexp'_1 -
    \inversevr{\Phi_{r}}$. We conclude by applying $\ruleVDArrowE$ to derive
    $\Phi$ as follows \[
\Rule{
  \derivable{\Phi_{t_1}}{\sequT{\Gamma_1}{\assign{t_1}{\multiset{\functtype{\M}{\sigma}}}}{\cmult_1}{\cexp_1+1}{\csize_1}}{\SysTight}
  \quad
  \derivable{\Phi_{u}}{\sequT{\Gamma_2}{\assign{u}{\M}}{\cmult_2}{\cexp_2}{\csize_2}}{\SysTightCBV}
}{
  \sequT{\Gamma}{\assign{t}{\sigma}}{\cmult}{\cexp}{\csize}
}{\ruleVTArrowE}
    \] where $\cexp = \cexp_1 + 1 + \cexp_2 - 1 = \cexp'_1 -
    \inversevr{\Phi'_{r}} + \cexp'_2 - \inversevr{\Phi'_{u}} = \cexp'_1 +
    \cexp'_2 - (\inversevr{\Phi'_{r}} + \inversevr{\Phi'_{u}}) =
    \cexp' - \inversevr{\Phi'}$ as expected.
    
    \item $\cbvterm{t_1} \neq \ctxtapp{\ctxt{L}}{\termbang{r'}}$ (\ie
    $\neg\pval{t_1}$). Then, $r = \termder{(\cbvterm{t_1})}$. Moreover,
    $\Phi'_{r}$ necessarily comes from an application of rule $\ruleBDDer$, and
    $\multiset{\functtype{\cbvtypeneg{\M}}{\cbvtypepos{\sigma}}} =
    \cbvtypepos{\multiset{\functtype{\M}{\sigma}}}$. Thus, we have a
    \cbvrelevant\ type derivation
    $\derivable{\Phi'_{t_1}}{\sequT{\cbvtypeneg{\Gamma_1}}{\assign{\cbvterm{t_1}}{\cbvtypepos{\multiset{\functtype{\M}{\sigma}}}}}{\cmult_1}{\cexp'_1}{\csize_1}}{\SysTight}$
    satisfying the hypothesis of the theorem. By the \ih we get a type
    derivation
    $\derivable{\Phi_{t_1}}{\sequT{\Gamma_1}{\assign{t_1}{\multiset{\functtype{\M}{\sigma}}}}{\cmult_1}{\cexp_1}{\csize_1}}{\SysTightCBV}$
    with $\cexp_1 = \cexp'_1 - \inversevr{\Phi'_{t_1}} = \cexp'_1 -
    \inversevr{\Phi'_{r}}$. We conclude by applying $\ruleVDArrowE$ to derive
    $\Phi$ \[
\Rule{
  \derivable{\Phi_{t_1}}{\sequT{\Gamma_1}{\assign{t_1}{\multiset{\functtype{\M}{\sigma}}}}{\cmult_1}{\cexp_1}{\csize_1}}{\SysTightCBV}
  \quad
  \derivable{\Phi_{u}}{\sequT{\Gamma_2}{\assign{u}{\M}}{\cmult_2}{\cexp_2}{\csize_2}}{\SysTightCBV}
}{
  \sequT{\Gamma}{\assign{t}{\typeneutral}}{\cmult}{\cexp}{\csize}
}{\ruleVTArrowE}
    \] where $\cexp = \cexp_1 + \cexp_2 - 1 = \cexp'_1 - \inversevr{\Phi'_{r}}
    + \cexp'_2 - \inversevr{\Phi'_{u}} - 1 = \cexp'_1 + \cexp'_2 -
    (\inversevr{\Phi'_{r}} + \inversevr{\Phi'_{u}} + 1) = \cexp' -
    \inversevr{\Phi'}$ as expected.
  \end{enumerate}
  
  \item $\ruleBDApp$. This case is similar to the previous one, resorting to
  the fact that $\ptight{\cbvtypeneg{\M}}$ implies $\ptight{\M}$.
  
  \item $\ruleBDArrowI$. Then $\cbvterm{t} = \termabs{x}{t'}$ which is not
  possible by definition. Hence, this case does not apply.
  
  \item $\ruleBDBang$. Then $\cbvterm{t} = \termbang{t'}$, $\Gamma =
  \ctxtsum{}{\Gamma_i}{i \in I}$, $\sigma = \intertype{\sigma_i}{i \in I}$ and
  $\Phi'$ is of the form \[
\Rule{
  \many{
    \derivable{\Phi'_i}{\sequT{\cbvtypeneg{\Gamma_i}}{\assign{t'}{\cbvtypeneg{\sigma_i}}}{\cmult_i}{\cexp'_i}{\csize_i}}{\SysTight}
  }{i \in I}
}{
  \sequT{\cbvtypeneg{\Gamma}}{\assign{\termbang{t'}}{\intertype{\cbvtypeneg{\sigma_i}}{i \in I}}}{+_{i \in I}{\cmult_i}}{1+_{i \in I}{\cexp'_i}}{+_{i \in I}{\csize_i}}
}{\ruleBDBang}
  \] where $\cexp' = 1 +_{i \in I}{\cexp'_i}$. There are two possible cases for
  $t'$:
  \begin{enumerate}
    \item $t' = x$ (\ie $t = x$). Then, $\Phi'_i$ necessarily comes from
    $\ruleBDAxiom$ for each $i \in I$, \ie
    $\many{\derivable{\Phi'_i}{\sequT{\assign{x}{\multiset{\cbvtypeneg{\sigma_i}}}}{\assign{x}{\cbvtypeneg{\sigma_i}}}{0}{0}{0}}{\SysTight}}{i \in I}$.
    Hence, $\Gamma = \assign{x}{\intertype{\sigma_i}{i \in I}}$, $\cmult =
    \csize = 0$ and $\cexp' = 1$. We conclude by $\ruleVDAxiom$ with
    $\derivable{\Phi}{\sequT{\assign{x}{\intertype{\sigma_i}{i \in I}}}{\assign{x}{\intertype{\sigma_i}{i \in I}}}{0}{1}{0}}{\SysTightCBV}$,
    since $\cbvtypeneg{\Gamma(x)} = \intertype{\cbvtypeneg{\sigma_i}}{i \in I}
    = \cbvtypepos{\sigma}$. Note that the counters are as expected since
    $\inversevr{\Phi'} = 0$.

    \item $t' = \termabs{x}{\cbvterm{r}}$ (\ie $t = \termabs{x}{r}$). Note
    that, for every $i \in I$, $\cbvtypeneg{\sigma_i} = \typeabs$ leads to a
    contradiction, since $\typeabs$ is not in the image of the translation.
    Hence, $\Phi'_i$ necessarily comes from $\ruleBDArrowI$, \ie $\sigma_i =
    \functtype{\M_i}{\tau_i}$ and $\Phi'_i$ is \[
\Rule{
  \sequT{\cbvtypeneg{\Gamma_i};\assign{x}{\cbvtypeneg{\M}}}{\assign{\cbvterm{r}}{\cbvtypepos{\tau_i}}}{\cmult_i}{\cexp'_i}{\csize_i}
}{
  \sequT{\cbvtypeneg{\Gamma_i}}{\assign{\termabs{x}{\cbvterm{r}}}{\cbvtypeneg{\sigma_i}}}{\cmult_i}{\cexp'_i}{\csize_i}
}{\ruleBDArrowI}
    \] for each $i \in I$. Then, by the \ih on we get
    $\many{\derivable{\Phi_i}{\sequT{\Gamma_i;\assign{x}{\M}}{\assign{r}{\tau_i}}{\cmult_i}{\cexp_i}{\csize_i}}{\SysTightCBV}}{i \in I}$
    with $\many{\cexp_i = \cexp'_i - \inversevr{\Phi'_i}}{i \in I}$. We
    conclude by applying rule $\ruleVDArrowI$ to derive $\Phi$ as follows \[
\Rule{
  \many{
    \derivable{\Phi_i}{\sequT{\Gamma_i;\assign{x}{\M}}{\assign{r}{\tau_i}}{\cmult_i}{\cexp_i}{\csize_i}}{\SysTightCBV}
  }{i \in I}
}{
  \sequT{\Gamma}{\assign{t}{\intertype{\functtype{\M_i}{\tau_i}}{i \in I}}}{+_{i \in I}{\cmult_i}}{1+_{i \in I}{\cexp_i}}{+_{i \in I}{\csize_i}}
}{\ruleVDArrowI}
  \] where $\cexp = 1 +_{i \in I}{\cexp_i} = 1
  +_{i \in I}{(\cexp'_i - \inversevr{\Phi'_i})} = 1 +_{i \in I}{\cexp'_i} -
  (+_{i \in I}{\inversevr{\Phi'_i}}) = \cexp' - \inversevr{\Phi'}$ as expected.
  \end{enumerate}
  
  \item $\ruleBDDer$. Then $\cbvterm{t} = \termder{(t')}$ which is not
  possible by definition. Hence, this case does not apply.
  
  \item $\ruleBDESubs$. This case is similar to $\ruleBTESubs$.
\end{itemize}
\end{enumerate}
\end{proof}


\begin{sidewaysfigure}
\input{example-tight-derivation-cbv-translation.tex}
\caption{Type derivation in System $\SysTight$ for $\cbvterm{t_0} =
\cbvterm{(\termapp{\termapp{\Kterm}{(\termapp{z}{\id})}}{(\termapp{\id}{\id})})} =
\termapp{\termder{(\termapp{(\termabs{x}{\termbang{\termabs{y}{\termbang{x}}}})}{(\termapp{z}{\termbang{\id'}})})}}{(\termapp{\id'}{\termbang{\id'}})}$,
with $\id' = \termabs{w}{\termbang{w}}$.}
\label{f:t0cbpv}
\end{sidewaysfigure}

Fig.~\ref{f:t0cbpv} illustrates this result over the term $t_0$ in
Ex.~\ref{ex:t0cbv}. As a further (simpler) example, consider the CBV term
$\termapp{(\termabs{x}{x})}{y}$ (whose derivation $\Phi$ is on the left) and
its translation $\termapp{(\termabs{x}{\termbang{x}})}{\termbang{y}}$ into the
$\BangRev$-calculus (whose derivation $\Phi'$ is on the right):
\begin{center}
{\small
$
\Rule{
  \Rule{
    \Rule{
      \vphantom{
        \Rule{}{
          \sequT{\assign{x}{\multiset{\typeneutral}}}{\assign{x}{\typeneutral}}{0}{0}{0}
        }{}    
      }
    }{
      \sequT{\assign{x}{\multiset{\typevar}}}{\assign{x}{\typevar}}{0}{0}{0}
    }{}
  }{
    \sequT{}{\assign{\termabs{x}{x}}{\multiset{\functtype{\multiset{\typevar}}{\typevar}}}}{0}{1}{0}
  }{}
  \!\!
  \Rule{}{
    \sequT{\assign{y}{\multiset{\typevar}}}{\assign{y}{\multiset{\typevar}}}{0}{1}{0}
  }{}
}{
  \sequT{\assign{y}{\multiset{\typevar}}}{\assign{\termapp{(\termabs{x}{x})}{y}}{\typevar}}{1}{1}{0}
}{}
\enspace
\Rule{
  \Rule{
    \Rule{
      \Rule{}{
        \sequT{\assign{x}{\multiset{\typeneutral}}}{\assign{x}{\typeneutral}}{0}{0}{0}
      }{}
    }{
        \sequT{\assign{x}{\multiset{\typeneutral}}}{\assign{\termbang{x}}{\multiset{\typeneutral}}}{0}{1}{0}
    }{}
  }{
    \sequT{}{\assign{\termabs{x}{\termbang{x}}}{\functtype{\multiset{\typeneutral}}{\multiset{\typeneutral}}}}{0}{1}{0}
  }{}
  \!\!
  \Rule{
    \Rule{}{
      \sequT{\assign{y}{\multiset{\typeneutral}}}{\assign{y}{\typeneutral}}{0}{0}{0}
    }{}
  }{
      \sequT{\assign{y}{\multiset{\typeneutral}}}{\assign{\termbang{y}}{\multiset{\typeneutral}}}{0}{1}{0}
  }{}
}{
  \sequT{\assign{y}{\multiset{\typeneutral}}}{\assign{\termapp{(\termabs{x}{\termbang{x}})}{\termbang{y}}}{\multiset{\typeneutral}}}{1}{2}{0}
}{}
$}
\end{center}
Notice that $\ptight{\Gamma}$ if and only if $\ptight{\cbvtypeneg{\Gamma}}$.
Unfortunately, this is not sufficient to translate a tight
derivation in $\SysTightCBV$ into a tight derivation in $\SysTight$. Indeed,
the variable $x$ of type $\typevar$ translates to $\termbang{x}$ of type
$\multiset{\typeneutral}$. However, this only happens
if the derived type is $\typevar$, which is an auxiliary type of system
$\SysTightCBV$ used to identify variables that are not used as values because
they will be applied to some argument. In the case of
\emphdef{\bangrelevant\
$\SysTightCBV$ derivations}, defined as not deriving type $\typevar$,
tightness is indeed preserved. As a consequence, as for CBN, is it possible to
study CBV in the unified framework of the $\BangRev$-calculus and extract exact
measures for it by resorting to \relevant\ tight derivations:

\begin{corollary}
If
$\derivable{\Phi}{\sequT{\Gamma}{\assign{t}{\sigma}}{\cmult}{\cexp}{\csize}}{\SysTightCBV}$
is tight and \bangrelevant, then there
exists $\cfnrml{p}$ such that $\cbvterm{t}
\rewriten{\bangweak}^{(\cmult,\cexp)} p$ with $\cmult$ $\mStep$-steps, $\cexp +
\countvr{\Phi}$ $\eStep$-steps, and $\wsize{p} = \csize$. Conversely, if
$\derivable{\Phi'}{\sequT{\cbvtypeneg{\Gamma}}{\assign{\cbvterm{t}}{\cbvtypepos{\sigma}}}{\cmult}{\cexp'}{\csize}}{\SysTight}$
is tight and \cbvrelevant, then there exists $p \in \CBVNF$ such that
$t \rewriten{\callbyvalue}^{(\cmult,\cexp)} p$ with $\cmult$ $\mStep$-steps,
$\cexp' - \inversevr{\Phi'}$ $\eStep$-steps, and $\valsize{p} = \csize$.
\end{corollary}


\section{Conclusion}
\label{s:conclusion}

Following recent works exploring the power of CBPV, we develop a technique for
deriving tight type systems for CBN/CBV as special cases of a single tight type
system for the $\BangRev$-calculus, a subcalculus of CBPV inspired by Linear
Logic.

The idea to study semantical and operational properties in a CBPV framework in
order to transfer them to CBN/CBV has so far be
exploited in different
works~\cite{Levy06,Ehrhard16,EhrhardG16,GuerrieriM18,ChouquetT20,BucciarelliKRV20,SantoPU19}.
Moreover, relational models for CBN/CBV can be derived from the relational
model for CBPV, resulting in non-idempotent intersection type systems for them,
that provide upper bounds for the length of normalization
sequences~\cite{BucciarelliKRV20}. However, the challenging quest of a (tight)
quantitative type system for CBV, giving \emph{exact measures} for the 
length of normalization sequences instead of \emph{upper bounds}, and being at
the same time encodable in CBPV, has been open. None of the existing
proposals~\cite{AccattoliG18,LeberlePhD} could be defined/explained within such an
approach. In particular, the tight type systems that we propose for CBN/CBV
give independent exact measures for the length of multiplicative and
exponential reduction to normal form, as well as the size of these normal
forms.

Different topics deserve future attention. One of them is the study of
\emph{strong} reduction for the $\BangRev$-calculus, which allows to reduce
terms under \emph{all} the constructors, including bang. Appropriate
encodings of strong CBN and strong CBV should follow. Linear (head) reduction,
as well as other more sophisticated semantics like GOI also deserve some
attention.The tight systems presented in this work could also be used to
understand bounded computation.


\newpage
\bibliography{main.bbl}

\end{document}
